\documentclass{lmcs} %%% last changed 2014-08-20
\pdfoutput=1

\usepackage[utf8]{inputenc}

\usepackage{lastpage}
\lmcsdoi{17}{3}{5}
\lmcsheading{}{\pageref{LastPage}}{}{}%
{Apr.~02,~2020}{Jul.~20,~2021}{}

\usepackage{booktabs}
\usepackage{hyperref}


\usepackage{
  rawfonts,
  url,
  tikz,
  amssymb,
  amsmath,
  tabularx,
  amsthm,
  xspace,
  colortbl, 
  stmaryrd, 
  bm,
  algorithm,
  algpseudocode
}

\usetikzlibrary{arrows,decorations,shapes,automata,positioning,decorations.pathmorphing}
\tikzset{
  every picture/.style = {
    thick,
    >=stealth',
    node distance = 1.5em and 3em,
  }
  ,
  cross line/.style = {
    preaction = {
      draw=white,
      -,
      line width=4pt
    }
  }
  ,
  state/.style = {
    rectangle,
    rounded corners = 5pt,
    font = \footnotesize,
    draw,
    minimum width = 1em,
    minimum height = 1em
  }
  , % labels of states
  label-state/.style = {
    sloped,
    font = \scriptsize,
    label distance = -2pt
  }
  , % labels of edges
  label-edge/.style = {
    font = \scriptsize,    
    label distance = -2pt
  }
}
\allowdisplaybreaks %will allow page breaks in math mode.
\usepackage{color}

\usepackage[draft]{fixme}

\theoremstyle{plain}
\newtheorem{theorem}[thm]{Theorem}
\newtheorem{proposition}[thm]{Proposition}
\newtheorem{lemma}[thm]{Lemma}
\newtheorem{corollary}[thm]{Corollary}
\theoremstyle{definition}
\newtheorem{definition}[thm]{Definition}

\theoremstyle{remark}
 \newtheorem{example}[thm]{Example}

\usepackage[textwidth=0.89in,textsize=scriptsize]{todonotes}

\newcommand{\bisim}{\leftrightarroweq}

\newcommand{\filt}{\leftrightsquigarrow}

\newcommand{\fsigma}{\filt_\Sigma}

\newcommand{\NN}{\mathbb{N}}

\newcommand{\nat}{\NN}

\newcommand{\xml}{\text{XML}\xspace}

\newcommand{\nexptime}{\textsc{NExpTime}\xspace}

\newcommand{\pspace}{\textsc{PSpace}\xspace}
\newcommand{\expspace}{\textsc{ExpSpace}\xspace}
\newcommand{\nexpspace}{\textsc{NExpSpace}\xspace}

\newcommand{\midd}{\mathrel{\,\mid\,}}
\newcommand{\eqdef}                     % equal by definition
  {\stackrel{\scriptscriptstyle \mathrm{def}}{=}}

\newcommand{\md}{{\sf md}}

\newcommand{\tup}[1]{\langle #1 \rangle}
\newcommand{\set}[1]{\{#1\}}

\newcommand{\Coloneqq}{\mathrel{{\mathop:}{\mathop:}}=}

\newcommand{\xpath}{\text{XPath}\xspace}

\newcommand{\xpathd}{\ensuremath{\xpath_{=}}\xspace}

\newcommand{\hxpath}{\text{HXPath}\xspace}
\DeclareMathSymbol{\shortminus}{\mathbin}{AMSa}{"39}
\newcommand{\hxpv}{\hxpath_=(\shortminus)}
\newcommand{\hxpd}{\hxpath_=}
\newcommand{\hxps}{\hxpath_=(\shortminus,\dows)}

\newcommand{\hxprt}{\hxpath_=(*)}

\newcommand{\hxpe}{\hxpath_=(\subseteq)}

\newcommand{\hyb}{\mathcal{HL}\xspace}
\newcommand{\hl}{\hyb(@)}

\newcommand{\dowa}{\xspace\mathsf{a}\xspace}
\newcommand{\dowb}{\xspace\mathsf{b}\xspace}
\newcommand{\dowam}{\xspace\mathsf{a^{\shortminus}}\xspace}

\newcommand{\dows}{\xspace\mathsf{s}\xspace}

\newcommand{\sib}{\dows}

\newcommand{\OMIT}[1]{}

\newcommand{\mlabel}{\textit{Label}}
\newcommand{\madd}{\textit{add}}

\newcommand{\Zig}{\textbf{Zig}} %Comando tambien usado en la seccion de nodos
\newcommand{\Zag}{\textbf{Zag}} %Comando tambien usado en la seccion de nodos
\newcommand{\Nom}{\textbf{Nom}} 
\newcommand{\Harm}{\textbf{Harmony}} 
\newcommand{\lra}{\leftrightarrow}
\newcommand{\lab}{\mathsf{Prop}}
\renewcommand{\prop}{\lab}
\newcommand{\nom}{\mathsf{Nom}}

\newcommand{\mo}{\mathsf{Mod}}
\newcommand{\eq}{\mathsf{Eq}}

\newcommand{\no}{\mathit{nom}}
\newcommand{\lbl}{\mathit{V}}

\newcommand{\ahxpd}{\mathsf{HXP}}
\newcommand{\ah}{\mathsf{H}}

\renewcommand{\k}{K}

\newcommand{\fc}{\textsf{FC}}

\newcommand{\eqrel}{\sim}
\newcommand{\ra}{\rightarrow}

\renewcommand{\iff}{\xspace\mbox{iff}\xspace}

\newcommand{\Ra}{\Rightarrow}
\newcommand{\cls}[1]{[#1]}

\newcommand{\model}{\mathcal{M}}
\newcommand{\fr}{\mathcal{F}}

\newcommand{\ct}{\mathfrak{C}_{tree}}
\newcommand{\ctm}{\mathfrak{C}_{forest^{-}}}

\newcommand{\evnode}{\model_\Gamma,\Delta_i}
\newcommand{\evpath}{\model_\Gamma,\Delta_i,\Delta_j}

\newcommand{\uphr}{\upharpoonright}
\newcommand{\restmodel}[3]{#1_{\upharpoonright(#2,#3)}}
\newcommand{\formequiv}{\equiv}
\newcommand{\size}[1]{{\mid}#1{\mid}}

\newcommand{\Rho}{\mathrm{P}}

\newcommand{\st}{\mathsf{ST}}
\newcommand{\sub}{\mathsf{Sub}}

\newcommand{\itree}[3]{\restmodel{#1}{#2}{#3}^{IT}}
\newcommand{\ftree}[3]{\restmodel{#1}{#2}{#3}^{FT}}

\newcommand{\only}[2]{\mathsf{Only}^{#1}_{#2}}
\newcommand{\lin}{\mathsf{Lin}}

\usepgflibrary{arrows}
\usetikzlibrary{arrows,automata,shapes,decorations.markings,decorations.pathmorphing,
backgrounds,fit,calc,shadows.blur}                                                   

\tikzset{
    >=stealth',
    ar/.style={
           shorten <=2pt,
           shorten >=2pt,}
}

\tikzstyle{world} = [shape=circle,fill=white,inner sep=2pt,draw=black,blur shadow={shadow xshift=0ex,shadow yshift=0ex,shadow scale=1.1}]
\tikzstyle{blackworld} = [shape=circle,fill=black,inner sep=2pt,draw=black,blur shadow={shadow xshift=0ex,shadow yshift=0ex,shadow scale=1.1}]

\newdimen\proofrulebreadth \proofrulebreadth=.05em
\newdimen\proofdotseparation \proofdotseparation=1.25ex
\newdimen\proofrulebaseline \proofrulebaseline=2ex
\newcount\proofdotnumber \proofdotnumber=3
\let\then\relax
\def\hfi{\hskip0pt plus.0001fil}
\mathchardef\squigto="3A3B
\newif\ifinsideprooftree\insideprooftreefalse
\newif\ifonleftofproofrule\onleftofproofrulefalse
\newif\ifproofdots\proofdotsfalse
\newif\ifdoubleproof\doubleprooffalse
\let\wereinproofbit\relax
\newdimen\shortenproofleft
\newdimen\shortenproofright
\newdimen\proofbelowshift
\newbox\proofabove
\newbox\proofbelow
\newbox\proofrulename
\def\shiftproofbelow{\let\next\relax\afterassignment\setshiftproofbelo 
w\dimen0 }
\def\shiftproofbelowneg{\def\next{\multiply\dimen0 by-1 }%
\afterassignment\setshiftproofbelow\dimen0 }
\def\setshiftproofbelow{\next\proofbelowshift=\dimen0 }
\def\setproofrulebreadth{\proofrulebreadth}

\def\prooftree{% NESTED ZERO (\ifonleftofproofrule)
\ifnum  \lastpenalty=1
\then   \unpenalty
\else   \onleftofproofrulefalse
\fi
\ifonleftofproofrule
\else   \ifinsideprooftree
         \then   \hskip.5em plus1fil
         \fi
\fi
\bgroup% NESTED ONE (\proofbelow, \proofrulename, \proofabove,
\setbox\proofbelow=\hbox{}\setbox\proofrulename=\hbox{}%
\let\justifies\proofover\let\leadsto\proofoverdots\let\Justifies\proofoverdbl
\let\using\proofusing\let\[\prooftree
\ifinsideprooftree\let\]\endprooftree\fi
\proofdotsfalse\doubleprooffalse
\let\thickness\setproofrulebreadth
\let\shiftright\shiftproofbelow \let\shift\shiftproofbelow
\let\shiftleft\shiftproofbelowneg
\let\ifwasinsideprooftree\ifinsideprooftree
\insideprooftreetrue
\setbox\proofabove=\hbox\bgroup$\displaystyle % NESTED TWO
\let\wereinproofbit\prooftree
\shortenproofleft=0pt \shortenproofright=0pt \proofbelowshift=0pt
\onleftofproofruletrue\penalty1
}

\def\eproofbit{% NESTED TWO
\ifx    \wereinproofbit\prooftree
\then   \ifcase \lastpenalty
         \then   \shortenproofright=0pt  % 0: some other object, no indentation
         \or     \unpenalty\hfil         % 1: empty hypotheses, just glue
         \or     \unpenalty\unskip       % 2: just had a tree, remove glue
         \else   \shortenproofright=0pt  % eh?
         \fi
\fi
\global\dimen0=\shortenproofleft
\global\dimen1=\shortenproofright
\global\dimen2=\proofrulebreadth
\global\dimen3=\proofbelowshift
\global\dimen4=\proofdotseparation
\global\count255=\proofdotnumber
$\egroup  % NESTED ONE
\shortenproofleft=\dimen0
\shortenproofright=\dimen1
\proofrulebreadth=\dimen2
\proofbelowshift=\dimen3
\proofdotseparation=\dimen4
\proofdotnumber=\count255
}

\def\proofover{% NESTED TWO
\eproofbit % NESTED ONE
\setbox\proofbelow=\hbox\bgroup % NESTED TWO
\let\wereinproofbit\proofover
$\displaystyle
}%
\def\proofoverdbl{% NESTED TWO
\eproofbit % NESTED ONE
\doubleprooftrue
\setbox\proofbelow=\hbox\bgroup % NESTED TWO
\let\wereinproofbit\proofoverdbl
$\displaystyle
}%
\def\proofoverdots{% NESTED TWO
\eproofbit % NESTED ONE
\proofdotstrue
\setbox\proofbelow=\hbox\bgroup % NESTED TWO
\let\wereinproofbit\proofoverdots
$\displaystyle
}%
\def\proofusing{% NESTED TWO
\eproofbit % NESTED ONE
\setbox\proofrulename=\hbox\bgroup % NESTED TWO
\let\wereinproofbit\proofusing
\kern0.3em$
}

\def\endprooftree{% NESTED TWO
\eproofbit % NESTED ONE
   \dimen5 =0pt% spread of hypotheses
\dimen0=\wd\proofabove \advance\dimen0-\shortenproofleft
\advance\dimen0-\shortenproofright
\dimen1=.5\dimen0 \advance\dimen1-.5\wd\proofbelow
\dimen4=\dimen1
\advance\dimen1\proofbelowshift \advance\dimen4-\proofbelowshift
\ifdim  \dimen1<0pt
\then   \advance\shortenproofleft\dimen1
         \advance\dimen0-\dimen1
         \dimen1=0pt
         \ifdim  \shortenproofleft<0pt
         \then   \setbox\proofabove=\hbox{%
                         \kern-\shortenproofleft\unhbox\proofabove}%
                 \shortenproofleft=0pt
         \fi
\fi
\ifdim  \dimen4<0pt
\then   \advance\shortenproofright\dimen4
         \advance\dimen0-\dimen4
         \dimen4=0pt
\fi
\ifdim  \shortenproofright<\wd\proofrulename
\then   \shortenproofright=\wd\proofrulename
\fi
\dimen2=\shortenproofleft \advance\dimen2 by\dimen1
\dimen3=\shortenproofright\advance\dimen3 by\dimen4
\ifproofdots
\then
         \dimen6=\shortenproofleft \advance\dimen6 .5\dimen0
         \setbox1=\vbox to\proofdotseparation{\vss\hbox{$\cdot$}\vss}%
         \setbox0=\hbox{%
                 \advance\dimen6-.5\wd1
                 \kern\dimen6
                 $\vcenter to\proofdotnumber\proofdotseparation
                         {\leaders\box1\vfill}$%
                 \unhbox\proofrulename}%
\else   \dimen6=\fontdimen22\the\textfont2 % height of maths axis
         \dimen7=\dimen6
         \advance\dimen6by.5\proofrulebreadth
         \advance\dimen7by-.5\proofrulebreadth
         \setbox0=\hbox{%
                 \kern\shortenproofleft
                 \ifdoubleproof
                 \then   \hbox to\dimen0{%
                         $\mathsurround0pt\mathord=\mkern-6mu%
                         \cleaders\hbox{$\mkern-2mu=\mkern-2mu$}\hfill
                         \mkern-6mu\mathord=$}%
                 \else   \vrule height\dimen6 depth-\dimen7 width\dimen0
                 \fi
                 \unhbox\proofrulename}%
         \ht0=\dimen6 \dp0=-\dimen7
\fi
\let\doll\relax
\ifwasinsideprooftree
\then   \let\VBOX\vbox
\else   \ifmmode\else$\let\doll=$\fi
         \let\VBOX\vcenter
\fi
\VBOX   {\baselineskip\proofrulebaseline \lineskip.2ex
         \expandafter\lineskiplimit\ifproofdots0ex\else-0.6ex\fi
         \hbox   spread\dimen5   {\hfi\unhbox\proofabove\hfi}%
         \hbox{\box0}%
         \hbox   {\kern\dimen2 \box\proofbelow}}\doll%
\global\dimen2=\dimen2
\global\dimen3=\dimen3
\egroup % NESTED ZERO
\ifonleftofproofrule
\then   \shortenproofleft=\dimen2
\fi
\shortenproofright=\dimen3
\onleftofproofrulefalse
\ifinsideprooftree
\then   \hskip.5em plus 1fil \penalty2
\fi
}

\begin{document}

\title[Axiomatizing Hybrid XPath with Data]{Axiomatizing Hybrid XPath with Data}
%\titlecomment{{\lsuper*}OPTIONAL comment concerning the title, \eg, 
 % if a variant or an extended abstract of the paper has appeared elsewhere.}

\author[C.~Areces]{Carlos Areces}	%required
\address{Universidad Nacional de C\'ordoba \& CONICET, Argentina}	%required
%\thanks{thanks 1, optional.}	%optional

\author[R.~Fervari]{Raul Fervari}	%optional
%\address{Universidad Nacional de C\'ordoba \& CONICET, Argentina}	%required
\email{carlos.areces@unc.edu.ar}  %optional
\email{rfervari@unc.edu.ar}  %optional

%% etc.

%% required for running head on odd and even pages, use suitable
%% abbreviations in case of long titles and many authors:

%%%%%%%%%%%%%%%%%%%%%%%%%%%%%%%%%%%%%%%%%%%%%%%%%%%%%%%%%%%%%%%%%%%%%%%%%%%

%% the abstract has to PRECEDE the command \maketitle:
%% be sure not to issue the \maketitle command twice!

\begin{abstract}
  In this paper we introduce sound and strongly complete
  axiomatizations for XPath with data constraints extended with hybrid
  operators. First, we present $\hxpd$, a multi-modal version of XPath
  with data, extended with nominals and the hybrid operator $@$. Then,
  we introduce an axiomatic system for $\hxpd$, and we prove it is
  strongly complete with respect to the class of abstract data models,
  i.e., data models in which data values are abstracted as equivalence
  relations. We prove a general completeness result similar to the one
  presented in, e.g.,~\cite{BtC06}, that ensures that certain
  extensions of the axiomatic system we introduce are also complete.
  The axiomatic systems that can be obtained in this way cover a large
  family of hybrid XPath languages over different classes of frames,
  for which we present concrete examples. In addition, we investigate
  axiomatizations over the class of tree models, structures widely
  used in practice.  We show that a strongly complete, finitary, first-order 
  axiomatization of hybrid \xpath over trees does not exist,
  and we propose two alternatives to deal with this issue.  We finally
  introduce filtrations to investigate the status of decidability of
  the satisfiability problem for these languages.
\end{abstract}

\maketitle

\section{XPath as a Modal Logic with Data Tests}
\label{sec:intro}
XPath is, arguably, the most widely used query language for the
eXtensible Markup Language (XML).  Indeed, XPath is implemented in
XSLT~\cite{wad00} and XQuery~\cite{wwwc02} and it is used in many
specification and update languages (e.g., Saxon). It is,
fundamentally, a general purpose language for addressing, searching,
and matching pieces of an XML document.  It is an open standard and
constitutes a World Wide Web Consortium (W3C)
Recommendation~\cite{xpath:w3c}.  In~\cite{kost:xpat15} XPath is
adapted to be used
%\citeauthor{kost:xpat15}~\citeyear{kost:xpat15} adapt XPath to be used
as a powerful query language over knowledge bases.
Core-XPath~\cite{GKP05} is the fragment of XPath 1.0 containing the
navigational behaviour of XPath.  It can express properties of the
underlying tree structure of the XML document, such as the label (tag
name) of a node, but it cannot express conditions on the actual data
contained in the attributes. In other words, it is essentially a
\emph{classical modal
  logic}~\cite{blackburn2001modal,HML06}. For instance the path
expressions \texttt{child, parent} and \texttt{descendant-or-self} are
basically the modal operators $\lozenge$, $\lozenge^-$ and
$\lozenge^*$ respectively.
% (or, in the \xpath notation, they correspond
%to the $\dow$, $\upw$ and $\rtdow$ axes respectively).

Core-XPath has been well studied from a modal logic point of view.  For
instance, its satisfiability problem is known to be decidable even in
the presence of Document Type Definitions
(DTDs)~\cite{M04,BFG08}. Moreover, it is known that it is equivalent,
in terms of expressive power, to first-order logic with two variables
(FO2) over an appropriate signature on trees~\cite{MdR05}, and that it
is strictly less expressive than Propositional Dynamic Logic (PDL)
with converse over trees~\cite{BK08}. Sound and complete
axiomatizations for Core-XPath have been introduced
in~\cite{CateM09,cateLM10}. It has been argued that the study of \xpath
from a modal logic point of view is not merely a theoretical challenge
but it has also concrete applications.  For instance, by investigating
the \emph{expressive power} of \xpath fragments it is possible to
determine if an integrity constraint can be simplified.  Also, query
containment and query equivalence --two fundamental tasks in query
optimization-- can be reduced to the \emph{satisfiability problem}. In
particular these results have direct impact on, e.g., security~\cite{fan04},
type checking~\cite{MartensN07} and consistency of XML
specifications~\cite{Arenas2002}. %, just to name a few examples.

However, from a database perspective, Core-XPath falls short %is not expressive enough 
to define the most important construct in a query language: the
\emph{join}. Without the ability to relate nodes based on the actual
data values of the attributes, the logic's expressive power is
inappropriate for many applications.  The extension of Core-XPath with
(in)equality tests between attributes of elements in an XML document
is named Core-Data-XPath in~\cite{BojanczykMSS09}. Here, we will call
this logic $\xpathd$. Models of $\xpathd$ are usually data trees which can be
seen as XML documents. A data tree is a tree whose nodes contain a
label from a finite alphabet and a data value from an infinite domain.
%(see Figure~\ref{fig:datatree} for an example). 
From a modal logic perspective, these data trees are a particular
class of \emph{relational models}. In recent years other data structures have
been considered, and $\xpathd$ has been used to query these
structures. In particular, in this article we
consider arbitrary \emph{data graphs} as models for $\xpathd$.  Data
graphs are the underlying mathematical structure in \emph{graph
  databases} (see, e.g.,~\cite{LV12,Robinson:2013,AnglesG08}) and it
is important to study the metalogical properties of languages to query
this particular kind of models (see,
e.g.,~\cite{LibkinMV16,AbriolaBFF18}).  In this respect, we focus on a
variant of $\xpathd$ which provides us with several interesting
expressivity features.

%For instance, in Figure~\ref{fig:data-graph} the
%expression \emph{``there is a one-step descendant and a two-steps
%  descendant sharing the same data value''} is satisfied at $x$, given
%the presence of $u$ and $z$. The expression {\em ``there are two
%  children with different data value''} is also true at $x$, because
%$y$ and $z$ have different data.

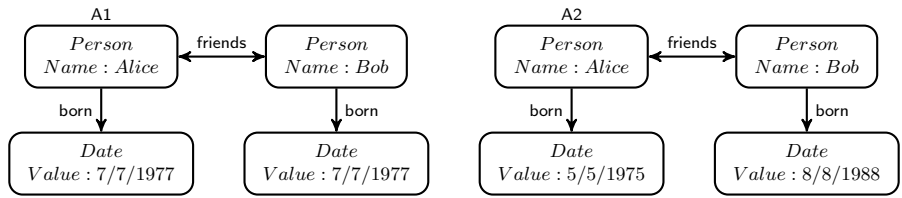
\begin{figure}[h]
  \centering
     \begin{tikzpicture}[->,scale=0.8, every node/.style={scale=0.8}]
     
      \node [state, label = {[label-state]above:$\mathsf{A1}$}] (w1) {$\begin{array}{c}Person \\  
      Name: Alice\end{array}$};      
      \node [state, right = of w1] (w2) {$\begin{array}{c}Person \\  
      Name: Bob\end{array}$};  
      
       \node [state,  right = of w2, label = {[label-state]above:$\mathsf{A2}$}] (w3) {$\begin{array}{c}Person \\  
      Name: Alice\end{array}$};      
      \node [state, right = of w3] (w4) {$\begin{array}{c}Person \\  
      Name: Bob\end{array}$};

      \node [state, below = of w1] (ww1) {$\begin{array}{c}Date\\  
      Value: 7/7/1977\end{array}$};      
      \node [state, below = of w2] (ww2) {$\begin{array}{c}Date\\  
      Value: 7/7/1977\end{array}$};      
      
      \node [state,  below = of w3] (ww3) {$\begin{array}{c}Date\\  
      Value: 5/5/1975\end{array}$};      
    
      \node [state,  below = of w4] (ww4) {$\begin{array}{c}Date\\  
      Value: 8/8/1988\end{array}$};  
      
      \path (w1) edge[<->] node [label-edge, above] {$\mathsf{friends}$} (w2);
      \path (w3) edge[<->] node [label-edge, above] {$\mathsf{friends}$} (w4);   
     
      \path (w1) edge node [label-edge, left] {$\mathsf{born}$} (ww1);
      \path (w2) edge node [label-edge, right] {$\mathsf{born}$} (ww2);

      \path (w3) edge node [label-edge, left] {$\mathsf{born}$} (ww3);
      \path (w4) edge node [label-edge, right] {$\mathsf{born}$} (ww4);

    \end{tikzpicture}
    
  \caption{Example of a graph-structured database.} 
  \label{fig:data-graph}
\end{figure}

Figure~\ref{fig:data-graph} shows an example of a data graph, where we
can see that edges carry labels from a finite alphabet (such as
``friends''), and whose nodes carry data values, usually presented as a
``\emph{key:\,value}'' pair, with \emph{key} being a label from a
finite alphabet (e.g., ``Name'') and \emph{value} being a data value
from an infinite domain (e.g., ``Alice'').

One of the characteristic advantages of graph databases is that they
directly exhibit the \emph{topology} of the data.  Most of the
approaches for querying these databases have focused on exploring the
topology of the underlying graph exclusively. However, little
attention has been paid to queries that check how the actual data
contained in the graph nodes relates with the topology.  $\xpathd$
allows comparison of data values by equality or inequality, even
though it does not grant access to the concrete data values themselves.

The main characteristic of $\xpathd$ is to allow formulas of the form
$\tup{\alpha=\beta}$ and $\tup{\alpha\neq\beta}$, where $\alpha,\beta$
are path expressions that navigate the graph using accessibility
relations: successor, predecessor, and for the particular case of tree-like 
models, 
descendant, next-sibling, etc.; and
can make tests in intermediate nodes.  In addition, in this work we
consider the novel `$@$' navigation operator that allows us to jump to named
nodes.  The formula $\tup{\alpha=\beta}$ (respectively
$\tup{\alpha\neq\beta}$) is true at a node $x$ of a data graph if
there are nodes $y, z$ that can be reached by paths satisfying the 
expressions $\alpha$ and $\beta$ respectively, and such that the data value of $y$ is
equal (respectively different) to the data value of $z$.

The benefits of having data comparisons is described in the following example.
Consider a social network which stores the information about the
friendship relation between persons, and also stores the date of birth
of each individual in the network. This scene can be represented as in
Figure~\ref{fig:data-graph}. Topological information is not sufficient
to distinguish the Alice on the left, {\sf A1} from the Alice on the
right, {\sf A2} (formally, these two nodes are bisimilar for the basic
modal logic). However, the expression `the person that is friend of
someone born the same day'' expressible in $\xpathd$, is true only at
{\sf A1}.

Recent articles investigate $\xpathd$ from a modal perspective. For
example, satisfiability and evaluation are discussed
in~\cite{FigPhD,FigueiraS11,Figueira12ACM}, while model theory and
expressivity are studied
in~\cite{ADF14,ICDT14,ICDT14Jair,ADF17IC,KR16,AbriolaBFF18}.  Query
power and complexity of several $\xpathd$ fragments are investigated
in~\cite{LibkinMV16,FigueiraD:2018}.  In~\cite{BaeldeLS15}, a Gentzen-style sequent
calculus is given for a very restricted fragment of $\xpathd$, named
DataGL. In DataGL, data comparisons are allowed only between the
evaluation point and its successors. An extension of the equational
axiomatic system from~\cite{cateLM10} is introduced
in~\cite{AbriolaDFF17}, allowing downward navigation and
equality/inequality tests.

In this article, we will focus on the proof theory of $\xpathd$
extended with hybrid operators.  It has been argued that the inclusion
of hybrid operators in a modal language leads to a better proof
theory~\cite{brau:hybr}.  In particular, it can lead to general
completeness results for axiomatizations by enabling Henkin style
completeness proofs~\cite{blackburn2001modal,BtC06}.  Moreover, we
will focus on strong completeness results (i.e., the axiomatizations
we introduce characterize the consequence relation
$\Gamma \vdash \varphi$ and not only theoremhood $\vdash \varphi$).

First results using hybridization
techniques were introduced in~\cite{ArecesF16} for a basic axiomatic
system for $\xpathd$. In~\cite{ArecesFS17} nominals are
used as labels in a tableaux system.

\paragraph{\bf Contribution.}
\label{subsec:contrib}
We will introduce a Hilbert-style axiomatization for $\xpathd$ with
forward navigation for multiple accessibility relations, multiple
equality/inequality comparisons, and where node expressions are
extended with nominals (special labels that are valid in only one
node), and path expressions are extended with the hybrid operator $@$
(allowing the navigation to some particular named node). We call this
logic {\em Hybrid XPath with Data} (denoted $\hxpd$). The language we consider
is a variant of \xpath fragments for data graphs in the literature.
It improves the expressive power of the fragment $\xpathd$
from~\cite{KR16,AbriolaBFF18}, e.g., by enabling us to express global
properties. On the other  hand, it has similarities with the
fragment GXPath$_{core}$(eq)
from~\cite{LibkinMV16}. GXPath$_{core}$(eq) allows us to express data
properties of nodes that are reachable in a finite number of steps,
due to the reflexive-transitive navigation on atomic steps. However,
the hybrid operator $@$ in $\hxpd$ is enough to capture some of its
expressivity by allowing us to `jump' to nodes that are not accessible
in one step, or that are simply not connected with the current
node. One of the advantages of our approach, is that hybrid operators
add only first-order behaviour, whereas reflexive-transitive closure
needs second-order expressivity.

We will take advantage of hybrid operators to prove completeness using
a Henkin-style model construction as presented
in~\cite{Gold1984,blackburn2001modal,BtC06,SchroderP10}. In fact, the
use of the so-called unorthodox inference rules (i.e., inference
rules involving side conditions) will help us define an axiom system
whose completeness can be automatically extended to more expressive
logics.

\medskip Summing up, our contributions in this article are the
following:
\begin{itemize}
\item We define a sound and strongly complete axiomatic system for
  multi-modal $\hxpd$, refining the axioms and rules first presented
  in~\cite{ArecesF16}.

\item We give a Henkin-style completeness proof, proving that the system is
  strongly complete for any extension given by pure axioms (formulas
  involving only nominals) and existential rules (instantiations of
  first-order formulas with quantification pattern $\forall\exists$
  followed by an $\hxpd$-formula). This gives us a strong completeness
  result with respect to a wide variety of frame classes.

\item Even though the main focus of our work is axiomatizing node
  expressions, we also provide as a by-product strong completeness results for
  path expressions.

\item We discuss concrete extensions of $\hxpd$ in which completeness
  is automatically obtained. In particular we explore backward and
  sibling navigation, and equality inclusion.

\item We show that a strongly complete, finitary, first-order
  axiomatization of hybrid XPath over trees does not exist. Then we
  introduce an infinite pure axiomatic system that is strongly complete
  for a slightly larger class of tree models. We also prove that this 
  system is weakly complete over tree models. As a corollary of the 
  method we use, we conclude that the satisfiability problem 
  for $\hxpd$ on trees is decidable in \expspace.

\item Finally, we prove that the satisfiability problem for the
  multi-modal logic $\hxpd$ and some of its pure extensions is
  decidable in \nexptime, using filtrations to establish a bounded
  finite model property. To our knowledge, the best known lower bound 
  can be obtained by \pspace-hardness of mono-modal hybrid \xpath with forward
  navigation~\cite{ArecesFS17}.
\end{itemize}

\paragraph{\bf Related Work.}
\label{subsec:related}
It is argued in~\cite{BtC06} that hybrid logic is a natural setting
for deductive systems involving Kripke semantics.  By having nominals
in the language, it is possible to imitate the role of first-order
constants to obtain a Henkin-style completeness proof. Moreover,
hybrid operators provide all the necessary machinery to get general
completeness results for a wide variety of systems.  The claim is that
allowing the use of \emph{unorthodox rules}, we can define a basic
axiomatic system that can be extended with \emph{pure
  axioms} and \emph{existential rules}, which lets us obtain
complete extensions with respect to diverse classes of frames, automatically. 

One of the first proof systems for data-aware fragments of \xpath was
introduced in~\cite{BaeldeLS15}. It is a Gentzen-style system for a
very restricted language named \emph{DataGL}, interpreted on finite data trees.
In \emph{DataGL} formulas of the form $\Diamond_=\varphi$ are read as: 
\emph{``the current node
  has the same data as a descendant where $\varphi$ holds''}.  They
introduced a sound and weakly complete sequent calculus for \emph{DataGL} and
established \pspace-completeness for its validity problem.  

An equational system which extends the one presented for navigational
fragments of \xpath in~\cite{cateLM10} was introduced
in~\cite{AbriolaDFF17}.  The authors consider first the fragment of
$\xpathd$ without inequality tests, over data trees.  
Therein, the proof of weak completeness relies on a
normal form theorem for both node and path expressions, and on a
canonical model construction for consistent formulas in normal form
inspired by~\cite{fine75}. When inequality is also taken into account,
similar ideas can be used but the canonical model construction and the
corresponding completeness proof become quite complex.

A Hilbert-style axiomatization for $\xpathd$ extended with hybrid
operators and backward navigation was first introduced in~\cite{ArecesF16},
using ideas from hybrid logic: nominals play the role of first-order
constants in a Henkin-style completeness proof. The system we present
herein is a vast extension of the work from~\cite{ArecesF16}, and it
automatically encompasses a large family of logics.

In~\cite{ArecesFS17} nominals are used as labels in a tableaux system for 
a restricted variant of $\hxpd$ with only one accessibility relation and one 
equality relation. The tableaux calculus
is shown to be sound, complete and terminating, and it is used to
prove that the satisfiability problem for such logic is \pspace-complete.

\paragraph{\bf Organization.}
\label{subsec:org}
The article is organized as follows. %First we introduce basic notions
%from modal and hybrid logic in Section~\ref{sec:modlogic}.  
In Section~\ref{sec:prelim} we introduce the syntax, semantics
and a notion of bisimulation for $\hxpd$.  We define the axiomatic system $\ahxpd$ in
Section~\ref{sec:axiom} and we prove its completeness. We finish that section by
introducing extensions of $\ahxpd$ and their complete axiomatizations
with pure axioms and existential rules. In Section~\ref{sec:tree} we introduce an axiomatic 
system for \xpath over the class of data trees.
% , the most common structures used in practice. 
Then, we use filtrations to show
that the satisfiability problem of $\hxpd$ is decidable in
Section~\ref{sec:filt}. To
conclude, in Section~\ref{sec:final} we discuss the results and
introduce future lines of research.

%%% Local Variables:
%%% mode: latex
%%% TeX-master: "lmcs20"
%%% End:

% \section{Modal Logics}
% \label{sec:modlogic}
% \input{modlogic}

\section{Hybrid XPath with Data}
\label{sec:prelim}
In this section we introduce the syntax and semantics for the logic we
call Hybrid XPath with Data ($\hxpd$ for short).  
Then, we present a notion of bisimulation for $\hxpd$.

\subsection{Syntax and Semantics of $\hxpd$}
\label{subsec:synsem}

Thoughout the text, let $\lab$ be a countable set of propositional
symbols; let $\nom$ be a countable set of nominals such that 
$\lab \cap \nom = \emptyset$; and let $\mo$ and $\eq$ be sets 
of modal and equality symbols, respectively.

\begin{definition}[Syntax]\label{def:hxpsyn}
  The sets $\mathsf{PExp}$ of path expressions (which we will note as
  $\alpha$, $\beta$, $\gamma$, $\ldots$) and $\mathsf{NExp}$ of node
  expressions (which we will note as $\varphi$, $\psi$, $\theta$,
  $\ldots$) of $\hxpd$ are defined by mutual recursion as follows:
\begin{align*}
  \mathsf{PExp} \; & \Coloneqq \; {\dowa}  \midd @_i \midd [\varphi]  \midd
  \alpha\beta    \\
  \mathsf{NExp} \; &\Coloneqq \; p \midd i \midd \neg \varphi  \midd \varphi \wedge \psi 
  \midd  \tup{\alpha=_e \beta} \midd \tup{\alpha \neq_e \beta}, 
\end{align*}
\noindent where $p \in \lab$, $i \in \nom$, $\dowa \in \mo$, $e\in\eq$,
$\varphi,\psi \in \mathsf{NExp}$, and $\alpha,\beta\in\mathsf{PExp}$.
We will refer to members of $\mathsf{PExp} \cup \mathsf{NExp}$ as
$\emph{expressions}$. In what follows, when referring to expressions
of $\hxpd$ we will reserve the term \emph{formula} for members of
$\mathsf{NExp}$.   We say that an expression is \emph{pure} if it 
does not contain propositional symbols. 
% Let $\varepsilon$ be an expression, the language of $\varepsilon$ (notation $\lang{\varepsilon}$) is defined as $\lang{\varepsilon}=\set{p \in \prop\cup\nom \mid p \mbox{ appears in } \varepsilon}$. 
\end{definition}

Notice that path expressions occur in node expressions in \emph{data
comparisons} of the form $\tup{\alpha =_e \beta}$ and
$\tup{\alpha \neq_e \beta}$, while node expressions occur in path
expressions in \emph{tests} like $[\varphi]$.

In what follows we will always use $*$ for $=_e$ and $\not =_e$, when
the particular operator used is not relevant. Other Boolean operators
are defined as usual: $\varphi\vee\psi:=\neg(\neg\varphi\wedge\neg\psi)$, and 
$\varphi\ra\psi:=\neg\varphi\vee\psi$. Below we define other operators as
abbreviations.

\begin{definition}[Abbreviations]
  Let $\alpha,\beta,\delta$ be path expressions, $\gamma_1, \gamma_2$
  path expressions or the empty string, $\varphi$ a node expression,
  $i$ a nominal, and $p$ an arbitrary symbol in $\lab$.  We define the
  following expressions:
\begin{center}
\begin{tabular}{rcl}%\hline
%\multicolumn{3}{l}{\footnotesize \bf Definitions} \\
  \toprule
$\top$ & $:=$ & $p\vee\neg p$ \\

$\bot$ & $:=$ & $\neg\top$ \\
$\epsilon$ & $:=$ & $[\top]$ \\ 
 $[\alpha=_e\beta]$ & $:=$ & $\neg\tup{\alpha\neq_e\beta}$\\

 $[\alpha\neq_e\beta]$ & $:=$ & $\neg\tup{\alpha=_e\beta}$ \\
~~$\tup{\alpha}\varphi$ & $:=$ & $\tup{\alpha[\varphi] =_e \alpha[\varphi]}$ 
\\
$[\alpha]\varphi$ & $:=$ & $\neg\tup{\alpha}\neg\varphi$~~ \\
$@_i \varphi$ & $:=$ &$\tup{@_i}\varphi$ \\ 
$\tup{\gamma_1(\alpha\cup\beta)\gamma_2*\delta}$ & $:=$ & $\tup{\gamma_1\alpha\gamma_2*\delta}\vee\tup{\gamma_1\beta\gamma_2*\delta}$\\
$\tup{\delta*\gamma_1(\alpha\cup\beta)\gamma_2}$ & $:=$ & 
$\tup{\delta*\gamma_1\alpha\gamma_2}\vee\tup{\delta*\gamma_1\beta\gamma_2}$ \\
\bottomrule
\end{tabular}
\end{center}

% \begin{center}
% \colorbox{blue!20}{
% \begin{tabular}{@{~~}rcl@{~~}}\hline
% \multicolumn{3}{l}{\footnotesize \bf Path Expressions}\\\hline
%  $\epsilon$ & $:=$ & $[\top]$ \\
%    $\tup{\gamma_1(\alpha\cup\beta)\gamma_2*\delta}$ & $:=$ & $\tup{\gamma_1\alpha\gamma_2*\delta}\vee\tup{\gamma_1\beta\gamma_2*\gamma_3}$\\
%   $\tup{\delta*\gamma_1(\alpha\cup\beta)\gamma_2}$ & $:=$ & 
% $\tup{\delta*\gamma_1\alpha\gamma_2}\vee\tup{\delta*\gamma_1\beta\gamma_2}$
%    \\
%    \hline
% \end{tabular}
% }
% \end{center}
\end{definition}

As a corollary of the definition of the semantic relation below, the
diamond and box expressions $\tup{\alpha}\varphi$ and
$[\alpha]\varphi$ will have their classical meaning, and the same will
be true for hybrid formulas of the form $@_i\varphi$.  Notice
that we use $@_i$ both as a path expression and as a modality; the
intended meaning will always be clear in context. Notice also that,
following the standard notation in XPath logics and in modal logics,
the $[\ ]$ operation is overloaded: for $\varphi$ a node expression
and $\alpha$ a path expression, both $[\alpha]\varphi$ and
$\alpha[\varphi]$ are well-formed expressions; the former is a node
expression where $[\alpha]$ is a box modality, the latter is a path
expression where $[\varphi]$ is a test.

We now define the structures that will be used to evaluate
formulas in the language. 

\begin{definition}[Hybrid Data Models]\label{def:models}
  % Let $\lab$ (the set of propositional symbols or labels) and $\nom$
  % (the set of nominals) be two infinite countable sets such that $\nom\cap\prop=\emptyset$. Let $\mo$ be a
  % set of modal symbols, and $\eq$ a set of equality symbols.  
  An \emph{abstract hybrid data model} is a tuple $\model=\tup{M,$
    $\{\eqrel_e\}_{e\in\eq},\{R_{\dowa}\}_{\dowa\in\mo},\lbl,\no}$,
  where $M$ is a non-empty set of elements; for each $e\in\eq$,
  ${\eqrel_e} \subseteq M \times M$ is an equivalence relation between
  elements of $M$; for each $\dowa\in\mo$,
  ${R_{\dowa}} \subseteq M \times M$ is the associated accessibility
  relation; $\lbl: M \ra 2^\lab$ is the valuation function; and
  $\no: \nom \ra M$ is a function that assigns nominals to certain
  elements. Given $m\in M$, $(\model,m)$ is called a {\em pointed 
  model} (parentheses are usually dropped).

  An {\em abstract data frame} (or just \emph{frame}) is a tuple
  $\fr=\tup{M,\set{\eqrel_e}_{e\in\eq},\set{R_{\dowa}}_{\dowa\in\mo}}$.  Let
  $\mathfrak{F}$ be a class of frames, its \emph{class of
    models} is
  $\mathsf{Mod}(\mathfrak{F}) =
  \{\tup{M,\set{\eqrel_e}_{e\in\eq},\set{R_{\dowa}}_{\dowa\in\mo},\lbl,\no}
  \mid \tup{M,\{\eqrel_e\}_{e\in\eq},\{R_{\dowa}\}_{\dowa\in\mo},} \in
  \mathfrak{F}, \text{ for $\lbl,\no$ arbitrary}\}$.

  A {\em concrete hybrid data model} is a tuple
  $\model=\tup{M,D, \{R_{\dowa}\}_{\dowa\in\mo},\lbl,\no,data}$, where $M$ is a
  non-empty set of elements; $D$ is a non-empty set of data; for each
  $\dowa\in\mo$, ${R_{\dowa}} \subseteq M \times M$ is the associated
  accessibility relation; $\lbl: M \ra 2^\lab$ is the valuation
  function; $\no: \nom \ra M$ is a function which names some nodes; and
  $data: \eq \times M \ra D$ is a function which assigns a data
  value to each node of the model (for each type of equality
  considered in the model).
\end{definition}

Concrete data models are most commonly used in applications, where we
encounter data from an infinite alphabet (e.g., alphabetic strings)
associated to the nodes in a semi-structured database, and different
ways of comparing this data.  It is easy to see that to each concrete
data model we can associate an abstract data model where
data is replaced by an equivalence relation that links all nodes with
the same data, for each equality symbol from $\eq$. 
Vice-versa, each abstract data model can be
``concretized'' by assigning to each pair node and symbol from $\eq$, 
its equivalence data class as data.  We will prove soundness and 
completeness over the class of abstract data models and, as a corollary, 
obtain completeness over concrete data models.

\begin{definition}[Semantics]
Let $\model=\tup{M,\{\eqrel_e\}_{e\in\eq}, \{R_{\dowa}\}_{\dowa\in\mo},\lbl,\no}$ be an abstract data model, and $m,n\in M$.
We define the semantics of $\hxpd$ as follows:
$$
\begin{array}{r@{ \ \ \iff \ \ }l}
\model,m,n\models \dowa & m R_{\dowa} n \\
\model,m,n\models @_i & \no(i)=n \\
\model,m,n\models [\varphi] & m=n \mbox{ and } \model,m\models \varphi \\
\model,m,n\models \alpha\beta & \mbox{there is $l\in M$ s.t. } \model,m,l\models\alpha \mbox{ and } \model,l,n\models\beta \\
\model,m\models p & p\in\lbl(m) \\
\model,m\models i & \no(i)=m \\
\model,m\models \neg\varphi & \model,m\not\models\varphi \\
\model,m\models \varphi\wedge\psi & \model,m\models\varphi \mbox{ and } \model,m\models\psi\\
\model,m\models\tup{\alpha=_e\beta} & \mbox{there are $n,l\in M$ s.t. } \model,m,n\models\alpha,~\model,m,l\models\beta \mbox{ and } n\eqrel_e l\\
\model,m\models\tup{\alpha\neq_e\beta} & \mbox{there are $n,l\in M$ s.t. } \model,m,n\models\alpha,~\model,m,l\models\beta \mbox{ and } n\not\eqrel_e l.
\end{array}
$$
We say a node expression $\varphi$ (path expression $\alpha$) is
\emph{satisfiable} if there exists $\model,m$ ($\model,m,n$) such that
$\model,m\models\varphi$ ($\model,m,n\models \alpha$).  Satisfiability
for sets of node and path expressions is defined in the obvious way.
Let $\Gamma \cup \set{\varphi}$ be a set of node expressions,
$\Gamma\models\varphi$ if for all $\model,m$, 
$\model,m\models\Gamma$ implies $\model,m\models\varphi$.  For a class
$\mathfrak{C}$ of models, $\Gamma \models_{\mathfrak{C}} \varphi$ if
for all $\model \in \mathfrak{C}$ and $m$ in the domain of $\model$,
$\model,m\models\Gamma$ implies $\model,m\models\varphi$. 
When $\Gamma=\emptyset$ we write $\models\varphi$ (respectively 
$\models_{\mathfrak{C}} \varphi$).
\end{definition}

As mentioned, it is a straightforward exercise to show that modal and
hybrid operators have their intended meaning.

\begin{proposition}\label{coro:hybridclassic}
Let $\model$ be an abstract hybrid data model and $m$ a state in $\model$. Then
$$
\begin{array}{r@{\,}l@{~~\iff~~}l}
	\model,m&\models @_i\varphi & \model,\no(i)\models\varphi	\\
	\model,m&\models\tup{\dowa}\varphi  & \mbox{for some $n\in M$ s.t. $m R_{\dowa} n$, } \model,n\models\varphi \\
	\model,m&\models[\dowa]\varphi  & \mbox{for all $n\in M$, $m R_{\dowa} n$ implies } \model,n\models\varphi.
\end{array}
$$	
\end{proposition}

The addition of the hybrid operators to XPath increases its expressive
power.  The following examples should serve as illustration.

\begin{example}
We list below  some $\hxpd$ expressions together with their intuitive
meaning:
\medskip

\noindent
\begin{tabular}{l@{\ \ \ }l}
  $\alpha[i]$ & There exists an $\alpha$ path between the current points of evaluation; the \\
              & second node is named $i$.\\
  $@_i\alpha$ & There exists an $\alpha$ path between the 
                node named $i$ and some other node.\\
  $\tup{@_i =_e @_j}$ & The node named $i$ has the same data as the node
                        named $j$, w.r.t.  the \\ 
                        & data field $e$.\\
  $\tup{\alpha =_e @_i\beta}$ & There exists a node accessible from the
                                current point of evaluation by an \\
              &  $\alpha$ path that has the same data
                than a node accessible from the point \\
              &  named $i$ by a $\beta$ path, w.r.t. the data field $e$.\\
\end{tabular}
\end{example}

The next example illustrates the expressivity gained by adding hybrid operators 
into the language, in a concrete example.

\begin{example}
 Consider the following queries in the data graph from Figure \ref{fig:data-graph}
 (we omit the subdindex in the symbols $=$ and $\neq$, since we deal with only one
 equality relation). 
  \medskip

%\def\arraystretch{1.2}
%\noindent
% \begin{tabular}{l}
%   \hline
%  $\tup{@_{\sf A1} {\sf born}[Value] = @_{\sf A1} {\sf friends} \  {\sf born} [Value] }$ \\ 
%  \rowcolor{lightgray!30}{The person with the id  {\sf A1} is friend of someone born in the same day.} \\
%  \hline
%  $[@_{\sf A2} {\sf born}[Value] \neq @_{\sf A2} {\sf friends} \  {\sf born}[Value] ]$ \\
%  \rowcolor{lightgray!30}{All friends of the person with id {\sf A2} were born in a day different to hers.} \\
%  \hline
%  $\tup{@_{\sf A1}[Name] = @_{\sf A2} [Name]} \wedge \tup{@_{\sf A1} {\sf born}[Value] \neq @_{\sf A2}  {\sf born} [Value] }$ \\ 
%  \rowcolor{lightgray!30}{The persons with id {\sf A1} and {\sf A2} have  the same value (since they are  both named  Alice)} \\
%  \rowcolor{lightgray!30}{and they were born in different days.} \\
%  \hline
%  \end{tabular}

\noindent
  \begin{tabular}{@{\ \ }l@{\ \ \ \ }l}
   $\tup{@_{\sf A1} {\sf born}[Value] = @_{\sf A1} {\sf friends} \  {\sf born} [Value] }$ & The person with the id  {\sf A1} is friend of   \\ 
   & someone born in the same day. \\[.5em]
   $[@_{\sf A2} {\sf born}[Value] \neq @_{\sf A2} {\sf friends} \  {\sf born}[Value] ]$ & All friends of the person with id {\sf A2}   \\ 
   & were born in a day different to hers. \\[.5em]
   $\tup{@_{\sf A1}[Name] = @_{\sf A2} [Name]}$ & The persons with id {\sf A1} and {\sf A2} have 
   \\ $ \ \ \ \ \wedge \tup{@_{\sf A1} {\sf born}[Value] \neq @_{\sf A2}  {\sf born} [Value] }$
   & the same value (both are named Alice) \\ 
   & and they were born in different days.
   \end{tabular}

  \medskip
  In the last item, the conjunct $\tup{@_{\sf A1}[Name] = @_{\sf A2} [Name]}$
  states that the nodes named by the nominals {\sf A1} and {\sf A2} contain a
  data value corresponding to the key {\em Name}, expressed by the test $[Name]$; 
  also, these data values coincide, without referring to the actual value.
  
  It is worth noting that with classical $\xpathd$ navigation, the two
  first properties can be expressed only if we are evaluating the
  formulas in the node with id {\sf A1} and {\sf A2},
  respectively. This is due the fact that $\xpathd$ only allows local
  navigation. Moreover, the third formula is not expressible in
  $\xpathd$, even in the presence of transitive closure operators, as
  it involves unconnected components of the graph.
\end{example}

Before concluding this section, we will introduce two classical notions that 
will be useful in the rest of the paper: the \emph{modal depth} and the set of \emph{subformulas} of a formula. These definitions are given on node and path expressions
seen as strings, in particular they consider the empty string $\lambda$ which is 
not in the language. However, these definitions work as intended.

\begin{definition}\label{def:modaldepth}
  We define the \emph{modal depth} of an expression, by mutual recursion as follows.
  $$
  \begin{array}{l@{~ = ~}l@{\qquad}l@{~ = ~}l}
 \md(\lambda) & 0 & \md(p) & 0  \\
 \md(\dowa\alpha) & 1 + \md(\alpha) & \md(\neg\varphi) & \md(\varphi) \\
 \md([\varphi]\alpha) & \max(\md(\varphi),\md(\alpha)) & \md(\varphi\wedge\psi) & \max(\md(\varphi),\md(\psi)) \\ 
 \md(@_i\alpha) & \md(\alpha) &  \md(\tup{\alpha*\beta}) & \max(\md(\alpha),\md(\beta)),
  \end{array}
  $$
  \noindent where $\lambda$ represents the empty string, $\dowa\in\mo$, $p\in\prop\cup\nom$, $i\in\nom$, $*\in\{=_e,\neq_e\}$, $\varphi,\psi\in\mathsf{NExp}$ and $\alpha,\beta\in\mathsf{PExp}$.
 \end{definition}

\begin{definition}\label{def:subformula}
  We define the set of \emph{subformulas} of an expression, by mutual recursion as follows.
  $$
  \begin{array}{@{~}l@{\,=\,}ll@{\,=\,}l@{~}}
  \sub(\lambda) & \emptyset & \sub(p) & \set{p}  \\
  \sub(\dowa\alpha) & \set{\tup{\dowa\alpha}\top}\cup\sub(\alpha) & \sub(\neg\varphi) & \set{\neg\varphi}\cup\sub(\varphi) \\
  \sub([\varphi]\alpha) & \set{\tup{[\varphi]\alpha}\top} \cup \sub(\varphi)\cup\sub(\alpha) & \sub(\varphi\wedge\psi) & \set{\varphi\wedge\psi}\cup \sub(\varphi) \cup \sub(\psi) \\ 
  \sub(@_i\alpha) & \set{\tup{@_i\alpha}\top,i}\cup\sub(\alpha) &  \sub(\tup{\alpha*\beta}) & \set{\tup{\alpha*\beta}}\cup\sub(\alpha)\cup\sub(\beta),
  \end{array}
  $$
  \noindent where $\lambda$ represents the empty string, $\dowa\in\mo$,  $p\in\prop\cup\nom$, $i\in\nom$, $*\in\{=_e,\neq_e\}$, $\varphi,\psi\in\mathsf{NExp}$ and $\alpha,\beta\in\mathsf{PExp}$.
  \end{definition}

\subsection{A Note on Bisimulations}

In the rest of the paper we will sometimes need the notion of
bisimulation for $\hxpd$. Bisimulation for \xpath on data trees was
investigated in~\cite{ICDT14,ICDT14Jair}. In~\cite{KR16,AbriolaBFF18}
it is shown that the same notion can be used in arbitrary models. It is simple to extend these definitions to take into account hybrid
operators (see, e.g.,~\cite{arec:hybr01,blackburn2001modal}).

\begin{definition}[Bisimulations]
\label{def:bisimulations}
Let
$\model=\tup{M,\{\eqrel_e\}_{e\in\eq},
  \{{R_{\dowa}}\}_{\dowa\in\mo},\lbl,\no}$
and
$\model'=\tup{M',\{\eqrel'_e\}_{e\in\eq},
  \{R'_{\dowa}\}_{\dowa\in\mo},\lbl',\no'}$
be two abstract hybrid data models, $m\in M$, $m'\in M'$. An
\emph{$\hxpd$-bisimulation} between $\model,m$ and $\model',m'$
is a relation
$Z\subseteq M\times M'$ such that $mZm'$, and when $lZl'$ we have:
\begin{itemize}
  \item $(\Harm)$ $\lbl(l)=\lbl'(l')$ and $l = \no(i)$ iff $l' = \no'(i)$ (for all $i\in\nom$).
  \item $(\Zig_=)$ If there are paths
  $\pi_1=lR_{\dowa_1} h_1R_{\dowa_2}\ldots R_{\dowa_{n_1}} h_{n_1}$ and 
  $\pi_2=lR_{\dowb_1} k_1R_{\dowb_2}\ldots R_{\dowb_{n_2}} k_{n_2}$ (for $\dowa_i,\dowb_i\in\mo$),
  then there are paths
  $\pi'_1=l'R'_{\dowa_1} h'_1R'_{\dowa_2}\ldots R'_{\dowa_{n_1}} h'_{n_1}$ and 
  $\pi'_2=l'R'_{\dowb_1} k'_1R'_{\dowb_2}\ldots \linebreak R'_{\dowb_{n_2}} k'_{n_2}$ such that
  \begin{enumerate}
  \item $h_iZh'_i$ for $1 \le i \le n_1$.
  \item $k_iZk'_i$ for $1 \le i \le n_2$.
  \item $h_{n_1}\eqrel_e k_{n_2} ~ \iff ~ h'_{n_1}\eqrel'_e k'_{n_2}$,
    for all $e\in\eq$.
  \end{enumerate}
\item $(\Zag_=)$ If there are paths
  $\pi'_1=l'R'_{\dowa_1} h'_1R'_{\dowa_2}\ldots R'_{\dowa_{n_1}} h'_{n_1}$ and
  $\pi'_2=l'R'_{\dowb_1} k'_1R'_{\dowb_2}\ldots \linebreak R'_{\dowb_{n_2}} k'_{n_2}$ (for $\dowa_i,\dowb_i\in\mo$), 
  then there are paths
  $\pi_1=lR_{\dowa_1} h_1R_{\dowa_2}\ldots R_{\dowa_{n_1}} h_{n_1}$ and
  $\pi_2=lR_{\dowb_1} k_1R_{\dowb_2}\ldots R_{\dowb_{n_2}} k_{n_2}$ such that
  conditions 1, 2 and 3 above are verified.
 \item $(\Nom)$ For all $i\in\nom$, $\no(i)Z\no'(i)$.
\end{itemize}
We write $\model,m\bisim\model',m'$ (and say that $\model,m$ and $\model',m'$ are $\hxpd$-bisimilar) if there is an $\hxpd$-bisimulation $Z$ such that $mZm'$.
% Let $S\subseteq \prop\cup\nom$, we write $\model,m\bisim_S \model',m'$ $Z$ when $Z$ is an $\hxpd$-bisimulation such as $mZm'$, but where \textbf{(Harmony)} and \textbf{(Nom)} are restricted to the symbols in $S$.
\end{definition}

The notion of $\hxpd$-bisimulation extends the one for basic modal logic~\cite{blackburn2001modal}.  As usual, \textbf{(Harmony)} takes care of atomic information (both propositional symbols and nominals).  The standard \textbf{(Zig)} and \textbf{(Zag)} should be strengthened to take into account 
that XPath modalities deal with two paths at the same time, and moreover equality of data values can be checked at their ending points.   Finally \textbf{(Nom)} takes care 
of the $@$ operator. 

It is straightforward to prove that $\hxpd$-bisimilar models 
satisfy the same $\hxpd$-formulas (see~\cite{AbriolaBFF18}).

\begin{proposition}
  \label{prop:invariance}
  Let $\model,m$ and $\model',m'$ be two pointed data models. If
  $\model,m\bisim\model',m'$ then for any $\hxpd$-formula $\varphi$,
  $\model,m\models\varphi$ if and only if $\model',m'\models\varphi$.
\end{proposition}

We can extend the notion of $\ell$-bisimulation introduced in~\cite{ICDT14Jair} in a similar way. 
An $\ell$-bisimulation checks model equivalence up to a certain modal depth. 

\begin{definition}[$\ell$-bisimulations]
  \label{def:hbisimulations}
  Let
  $\model=\tup{M,\{\eqrel_e\}_{e\in\eq},
    \{{R_{\dowa}}\}_{\dowa\in\mo},\lbl,\no}$
  and
  $\model'=\tup{M',\{\eqrel'_e\}_{e\in\eq},
    \{R'_{\dowa}\}_{\dowa\in\mo},\lbl',\no'}$
  be two abstract hybrid data models, $m\in M$, $m'\in M'$. We say that 
  $\model,m$ and $\model',m'$ are  \emph{$\ell$-bisimilar} for $\hxpd$ 
  (written $\model,m\bisim_\ell\model',m'$) if there exists a family of relations
  $(Z_j)_{j\leq\ell}$ in $M\times M'$ such that $mZ_\ell m'$, and for all $j\leq\ell$, when $lZ_jl'$ we have:
  \begin{itemize}
    \item $(\Harm)$ $\lbl(l)=\lbl'(l')$ and $l = \no(i)$ iff $l' = \no'(i)$  (for all $i\in\nom$).
    \item $(\Zig_=)$ If there are paths
    $\pi_1=lR_{\dowa_1} h_1R_{\dowa_2}\ldots R_{\dowa_{n_1}} h_{n_1}$ and 
    $\pi_2=lR_{\dowb_1} k_1R_{\dowb_2}\ldots R_{\dowb_{n_2}} k_{n_2}$ (for $\dowa_i,\dowb_i\in\mo$), $n_1,n_2 \leq j$,
    then there are paths
    $\pi'_1=l'R'_{\dowa_1} h'_1R'_{\dowa_2}\ldots R'_{\dowa_{n_1}} h'_{n_1}$ and 
    $\pi'_2=l'R'_{\dowb_1} k'_1R'_{\dowb_2}\ldots R'_{\dowb_{n_2}} k'_{n_2}$ such that
    \begin{enumerate}
    \item $h_iZ_{(j-n_1)+i}h'_i$ for $1 \le i \le n_1$.
    \item $k_iZ_{(j-n_2)+i}k'_i$ for $1 \le i \le n_2$.
    \item $h_{n_1}\eqrel_e k_{n_2} ~ \iff ~ h'_{n_1}\eqrel'_e k'_{n_2}$,
      for all $e\in\eq$.
    \end{enumerate}
  \item $(\Zag_=)$ If there are paths
    $\pi'_1=l'R'_{\dowa_1} h'_1R'_{\dowa_2}\ldots R'_{\dowa_{n_1}} h'_{n_1}$ and
    $\pi'_2=l'R'_{\dowb_1} k'_1R'_{\dowb_2}\ldots R'_{\dowb_{n_2}} k'_{n_2}$ (for $\dowa_i,\dowb_i\in\mo$), $n_1,n_2 \leq j$,
    then there are paths
    $\pi_1=lR_{\dowa_1} h_1R_{\dowa_2}\ldots R_{\dowa_{n_1}} h_{n_1}$ and
    $\pi_2=lR_{\dowb_1} k_1R_{\dowb_2}\ldots R_{\dowb_{n_2}} k_{n_2}$ such that
    conditions 1, 2 and 3 above are verified.
   \item $(\Nom)$ For all $i\in\nom$, $\no(i)Z_{\ell}\no'(i)$.
  \end{itemize}
  \end{definition}
  
  %The notion of $\ell$-bisimulation restricts the notion of bisimulation up to 
  %some level $\ell$. 
  Again, it is straightforward to prove the next proposition.

  \begin{proposition}
    \label{prop:linvariance}
    Let $\model,m$ and $\model',m'$ be two pointed data models. If
    $\model,m\bisim_\ell\model',m'$ then for any $\hxpd$-formula $\varphi$ such that $\md(\varphi)\leq \ell$,
    $\model,m\models\varphi$ if and only if $\model',m'\models\varphi$.
  \end{proposition}

% In a similar way, we obtain the next proposition.
% \begin{proposition}
% \label{prop:sinvariance}
% Let $S\subseteq \prop\cup\nom$, and let $\model,m$ and $\model',m'$ be two pointed models such as $\model,m\bisim_S \model',m'$.
% Then, for any $\hxpd$-formula $\varphi$ such that $\lang{\varphi}\subseteq S$,
% $\model,m\models\varphi$ if and only if $\model',m'\models\varphi$.
% \end{proposition}

%%% Local Variables: 
%%% mode: latex
%%% TeX-master: "lmcs20"
%%% End: 

\section{Axiomatization}
\label{sec:axiom}
%!TEX root = jair18.tex
In this section we introduce the axiomatic system $\ahxpd$ for
$\hxpd$. It is an extension of an axiomatic system for the hybrid
logic $\hl$ which adds nominals and the $@$ operator to the basic
modal language (see~\cite{blackburn2001modal}).  In particular, we
include axioms to handle data equality and inequality.  We will prove
that $\ahxpd$ is sound and strongly complete with respect to the class of
abstract hybrid data models.  Moreover, we will show that the system
is strongly complete for any extension with some particular kind of axioms and
inference rules.  This result will be helpful in order to
automatically obtain complete axiomatic systems for several natural
extensions of the language $\hxpd$.  The system is geared for
inference over node expressions, but we will discuss path expressions in
Section~\ref{subsec:pathcomp}. We remind the reader that when referring to $\hxpd$
we use the term \emph{formula} for node expressions.  In what
follows, we use $*$ in axioms which hold for both $=_e$ and $\neq_e$.

\subsection{The Basic Axiomatic System $\ahxpd$}
\label{subsec:axiomatic}
We present axioms and rules step by step, providing brief comments to
help the reader understand their role.  In all axioms and rules
$\varphi$, $\psi$ and $\theta$ are node expressions; $\alpha$,
$\beta$, $\gamma$ and $\eta$ are path expressions; and $i$, $j$ and
$k$ are nominals.
%\footnote{
  More precisely, we provide \emph{axiom and
    rule schemes}, i.e., they can be instantiated with arbitrary path
  and node expressions, respecting typing.  Hence, $\varphi$, $\psi$
  and $\theta$ can be instantiated with node expressions; $\alpha$,
  $\beta$, $\gamma$ and $\eta$ with path expressions; and $i$, $j$ and
  $k$ with nominals.
  %}. 

\begin{definition}[Theorems, Syntactic Consequence, Consistency]
  For $\mathsf{A}$ an axiomatic system, \emph{$\varphi$ is a theorem
    of $\mathsf{A}$} (notation: $\vdash_\mathsf{A} \varphi$) if it is
  either an instantiation of an axiom of $\mathsf{A}$, or it can be
  derived from an axiom of $\mathsf{A}$ in a finite number of steps by application of the rules of
  $\mathsf{A}$.
  Let $\Gamma$ be a set of node expressions. We write $\Gamma\vdash_\mathsf{A}\varphi$ and say that
  \emph{$\varphi$ is a syntactic consequence of $\Gamma$ in
    $\mathsf{A}$} if there exists a finite set
  $\Gamma'\subseteq\Gamma$ such that
  $\vdash_\mathsf{A} \bigwedge\Gamma' \ra \varphi$ (where
  $\bigwedge\emptyset = \top$).  We say that $\Gamma$ is \emph{consistent
    (for $\mathsf{A}$)} if $\Gamma\not\vdash_\mathsf{A}\bot$.  We
  write $\vdash$ instead of $\vdash_\mathsf{A}$ if $\mathsf{A}$ is
  clear from context.
\end{definition}

In addition to a complete set of axioms and rules for propositional
logic, $\ahxpd$ includes generalizations of the {\em K} axiom and the {\em
  Necessitation} rule for the basic modal logic to handle modalities
with arbitrary path expressions.

\begin{center}
  \def\arraystretch{1.1}
\begin{tabular}{ll}
\toprule 
\multicolumn{2}{l}{\bf Axiom and rule for classical modal logic}  \\
\midrule 
\begin{minipage}{0.45\textwidth}
\begin{tabular}{ll}
\emph{K} & $\vdash[\alpha](\varphi \ra \psi) \ra ([\alpha]\varphi \ra [\alpha]\psi)$	
\end{tabular}
\end{minipage}
&
\begin{minipage}{0.45\textwidth}
$$
\prooftree
\vdash \varphi
\justifies
\vdash [\alpha]\varphi 
\using
\emph{Nec}
\endprooftree \\ [3pt]
$$
\end{minipage}\\
\bottomrule 
\end{tabular}
\end{center}

Then we introduce generalizations of the inference rules for the hybrid
logic $\hl$. 

\begin{center}
\begin{tabular}{lll}
\toprule 
\multicolumn{3}{l}{\bf Hybrid rules}  \\
\midrule 
\begin{minipage}{0.3\textwidth}
$$
\prooftree
\vdash @_j \varphi
\justifies
\vdash  \varphi 
\using name
\endprooftree
$$
\end{minipage}
&
\begin{minipage}{0.3\textwidth}
$$
\prooftree
\vdash @_i\tup{\dowa}j \wedge \tup{@_j\alpha*\beta} \ra \theta
\justifies
\vdash \tup{@_i\dowa\alpha*\beta} \ra \theta
\using  paste 
\endprooftree
$$
\end{minipage}\\

\\

\multicolumn{3}{l}{$j$ is a nominal different from $i$ that does not occur
in $\varphi,\alpha,\beta,\theta$.}\\
\bottomrule
\end{tabular} 
\end{center}

Now we introduce axioms to handle $@$. Notice that $@_i$ is a path
expression and as a result, some of the standard hybrid axioms for $@$
have been generalized.  In particular, the $\k$ axiom and \emph{Nec}
rule above also apply to $@_i$. In addition, we provide axioms to
ensure that the relation induced by $@$ is a congruence.

\begin{center}
\centering
\def\arraystretch{1.1}
\begin{tabular}{l|l}
\toprule 
{\bf Axioms for $@$} & {\bf Congruence for $@$} \\
\midrule
\begin{tabular}{l@{~~}l}
%$\k_@$ & $@_i (\varphi\ra \psi) \ra (@_i \varphi \ra @_i \psi)$ \\
\emph{$@$-self-dual} & $\vdash\neg @_i \varphi \lra @_i\neg \varphi$ \\
\emph{$@$-intro} & $\vdash i \ra (\varphi \lra @_i \varphi)$ \\
\end{tabular}
&
\begin{tabular}{l@{~~}l}
\emph{$@$-refl} & $\vdash @_i i$ \\
\emph{agree} & $\vdash \tup{@_j@_i\alpha*\beta} \lra \tup{@_i \alpha*\beta}$ \\
\emph{back} & $\vdash \tup{\gamma @_i \alpha*\beta} \ra \tup{@_i \alpha *\beta}$
\end{tabular} \\
\bottomrule 
\end{tabular}
\end{center}

Axioms involving the classical XPath operators can be found below.  
First we introduce axioms to handle
complex path expressions in data comparisons. Then we introduce
axioms to handle data tests.

\begin{center}
\centering
\def\arraystretch{1.1}
\begin{tabular}{l}
\toprule
{\bf Axioms for paths} \\
\midrule
\begin{tabular}{l@{~~~~}l}
\emph{comp-assoc} & $\vdash\tup{(\alpha\beta)\gamma * \eta} \lra \tup{\alpha(\beta\gamma)*\eta} $ \\
\emph{comp-neutral} & $\vdash\tup{\alpha\beta*\gamma} \lra
\tup{\alpha\epsilon\beta*\gamma}$ ($\alpha$ or $\beta$ can be empty) \\
\emph{comp-dist}  & $\vdash\tup{\alpha\beta}\varphi \lra \tup{\alpha}\tup\beta\varphi$ \\
\end{tabular}\\
\midrule 
{\bf Axioms for data} \\
  \midrule
\begin{tabular}{l@{~~~~}l}
\emph{equal} & $\vdash\tup{\epsilon=_e\epsilon}$\\
\emph{distinct} & $\vdash\neg\tup{\epsilon\neq_e\epsilon}$ \\
\emph{$@$-data} & $\vdash\neg\tup{@_i{=_e}@_j} \lra \tup{@_i{\neq_e}@_j}$ \\
\emph{$\epsilon$-trans} & $\vdash\tup{\epsilon=_e\alpha}\wedge\tup{\epsilon=_e\beta}\ra\tup{\alpha=_e\beta}$\\
\emph{$*$-comm} & $\vdash\tup{\alpha*\beta} \lra \tup{\beta*\alpha}$\\
\emph{$*$-test} & $\vdash\tup{[\varphi]\alpha*\beta} \lra \varphi \wedge \tup{\alpha*\beta}$ \\
\emph{$@*$-dist} & $\vdash\tup{@_i \alpha * @_i \beta} \ra @_i\tup{\alpha*\beta}$ \\
\emph{subpath} & $\vdash\tup{\alpha\beta * \gamma} \ra \tup{\alpha}\top$ \\
\emph{comp$*$-dist} & $\vdash\tup{\alpha}\tup{\beta*\gamma} \ra \tup{\alpha\beta*\alpha\gamma}$
\end{tabular}\\
\bottomrule
\end{tabular}
\end{center}

It is a straightforward exercise to see that the axiomatic system $\ahxpd$
is sound. However, we will provide a more general statement of soundness later 
on.

Below we prove that some useful theorems and rules can be derived within $\ahxpd$.

\begin{proposition} \label{tab:deriv} The following are
  theorems and derived rules of $\ahxpd$, and will be used (explicitly
  or implicitly) in the rest of the paper.

\begin{center}
\begin{tabular}{l@{~~~}l}
\toprule 
test-dist &  $\vdash \tup{[\varphi]=_e[\psi]} \lra \varphi \wedge \psi$  \\
test-$\bot$ & $\vdash \tup{[\varphi]\neq_e[\psi]} \lra \bot$  \\
$@$-swap & $\vdash @_i\tup{\alpha*@_j\beta}  \lra  @_j\tup{\beta*@_i\alpha}$ \\
$@$-intro' & $\vdash (i\wedge \varphi) \ra @_i \varphi$ \\ 
$@$-sym. & $\vdash @_ij \ra @_ji$ \\
K$_@^{-1}$ & $\vdash(@_i\varphi\ra @_i\psi) \ra @_i(\varphi\ra\psi)$ \\
nom & $\vdash @_i j \wedge \tup{@_i\alpha*\beta} \ra \tup{@_j\alpha*\beta}$ \\
bridge & $\vdash \tup{\alpha}i\wedge@_i\varphi \ra \tup{\alpha}\varphi$ \\
name' & If $\vdash i\ra \varphi$ then $\vdash\varphi$, if $i$ does not appear in $\varphi$ \\
\bottomrule
\end{tabular}
\end{center}

\end{proposition}

\begin{proof}
\noindent{\em (test-dist and test-$\bot$).} Let $*$ be $=_e$ or $\neq_e$. Then:
\begin{align*}
&  \vdash \tup{[\varphi]*[\psi]}   \lra  \tup{[\varphi]\epsilon*[\psi]} \tag{\emph{comp-neutral}} \\
& \vdash \tup{[\varphi]\epsilon*[\psi]}   
 \lra   \varphi \wedge \tup{\epsilon*[\psi]} \tag{\emph{$*$-test}}\\
& \vdash \varphi \wedge \tup{\epsilon*[\psi]}
 \lra  \varphi \wedge\tup{[\psi]*\epsilon} \tag{\emph{$*$-comm}} \\
& \vdash \varphi \wedge\tup{[\psi]*\epsilon}
 \lra  \varphi \wedge \tup{[\psi]\epsilon*\epsilon} \tag{\emph{comp-neutral}} \\
& \vdash \varphi \wedge \tup{[\psi]\epsilon*\epsilon}
 \lra  \varphi \wedge \psi \wedge \tup{\epsilon*\epsilon} \tag{\emph{$*$-test}}
\end{align*}
Replacing $*$ by $=_e$ we get $\varphi \wedge \psi$ by \emph{equal}. Replacing it
by $\neq_e$ we get $\tup{\epsilon\neq_e\epsilon}$, and given
\emph{distinct} we can derive $\bot$.

%\smallskip

\noindent{\em ($@$-swap).} 
\begin{align*}
&\vdash @_i\tup{\alpha=_e@_j\beta}  \lra \tup{@_i\alpha=_e@_i@_j\beta} \tag{\emph{$@=$-dist}} \\
&\vdash \tup{@_i\alpha=_e@_i@_j\beta} \lra \tup{@_i@_j\beta=_e@_i\alpha} \tag{\emph{$=$-comm}}  \\
&\vdash \tup{@_i@_j\beta=_e@_i\alpha} \lra \tup{@_j\beta=_e@_i\alpha} \tag{\emph{agree}} \\
&\vdash \tup{@_j\beta=_e@_i\alpha} \lra \tup{@_i\alpha=_e@_j\beta} \tag{\emph{$=$-comm}} \\
&\vdash \tup{@_i\alpha=_e@_j\beta} \lra \tup{@_j@_i\alpha=_e@_j\beta} \tag{\emph{agree}} \\
&\vdash \tup{@_j@_i\alpha=_e@_j\beta} \lra @_j\tup{@_i\alpha=_e\beta} \tag{\emph{agree}}  \\
&\vdash @_j\tup{@_i\alpha=_e\beta} \lra @_j\tup{\beta=_e@_i\alpha}  \tag{\emph{$=$-comm}} 
\end{align*}

\noindent{\em ($@$-intro').} Direct from $@$-intro. 

\noindent{\em ($@$-sym.).} 
\begin{align*}
&\vdash (j \wedge i) \ra @_ji \tag{\emph{$@$-intro'}} \\
&\vdash (@_ij \wedge @_ii) \ra @_i@_ji \tag{\emph{Nec, K}} \\
&\vdash @_ij \ra @_ji \tag{\emph{$@$-refl., back}} 
\end{align*} 

\noindent{\em (K$_@^{-1}$).}
\begin{align*}
&\vdash (@_i\neg(\varphi\ra\psi)\ra @_i\varphi) \wedge (@_i\neg(\varphi\ra\psi)\ra @_i\neg\psi) \tag{\emph{prop, Nec, K}} \\
&\vdash @_i\neg(\varphi\ra\psi)\ra (@_i\varphi \wedge @_i\neg\psi) \tag{\emph{prop}} \\
&\vdash \neg @_i(\varphi\ra\psi)\ra (@_i\varphi \wedge \neg@_i\psi) \tag{\emph{$@$-self-dual}} \\
&\vdash (@_i\varphi\ra @_i\psi) \ra @_i(\varphi\ra\psi) \tag{\emph{prop}}
\end{align*} 

\noindent{\em (nom).}
\begin{align*}
&\vdash (j \wedge \tup{@_j\beta*\alpha}) \ra @_j\tup{@_j\beta*\alpha} \tag{\emph{$@$-intro'}} \\
&\vdash (@_ij \wedge @_i\tup{@_j\beta*\alpha}) \ra @_i@_j\tup{@_j\beta*\alpha} \tag{\emph{Nec, K}} \\
&\vdash (@_ij \wedge @_j\tup{@_i\alpha*\beta}) \ra @_i@_j\tup{@_j\alpha*\beta} \tag{\emph{$@$-swap, $*$-comm}} \\
&\vdash (@_j@_ij \wedge @_j\tup{@_i\alpha*\beta}) \ra @_j\tup{@_j\alpha*\beta} \tag{\emph{agree}} \\
&\vdash @_j(@_ij \wedge \tup{@_i\alpha*\beta} \ra \tup{@_j\alpha*\beta}) \tag{\emph{K$_@^{-1}$}} \\
&\vdash @_ij \wedge \tup{@_i\alpha*\beta} \ra \tup{@_j\alpha*\beta} \tag{\emph{name}}
\end{align*} 
\smallskip

\noindent{\em (bridge).} Using contrapositive, {\em bridge} is equivalent to
$\tup{\alpha}i \wedge [\alpha]\varphi \ra @_i\varphi$. Using the 
modal theorem $\vdash \tup{\alpha}\varphi\wedge[\alpha]\psi \ra
\tup{\alpha}(\varphi\wedge\psi)$, we reason:
\begin{align*}
&\vdash \tup{\alpha}i \wedge [\alpha]\varphi \ra \tup{\alpha}(i\wedge\varphi)\\ 
&\vdash \tup{\alpha}(i\wedge\varphi) \ra  \tup{\alpha}(@_i\varphi) \tag{\emph{$@$-intro'}} \\
&\vdash \tup{\alpha}(@_i\varphi)\ra  @_i\varphi \tag{\emph{back}}
\end{align*}

\noindent{\em (name').} Suppose $\vdash i \ra \varphi$. Then:
\begin{align*}
&\vdash @_i(i \ra \varphi) \tag{\emph{Nec}} \\
&\vdash @_ii\ra @_i\varphi \tag{\emph{K}} \\
&\vdash @_i\varphi \tag{\emph{$@$-refl.}} \\
&\vdash \varphi \tag{\em name}
\end{align*}
\end{proof}

\subsection{Extended Axiomatic Systems} % $\ahxpd+\Pi+\Rho$}
\label{subsec:extended}
In the next section we will prove that not only $\ahxpd$ is complete
with respect to the class of all models, but that extensions of $\ahxpd$
with \emph{pure axioms} and \emph{existential saturation rules}
preserve completeness with respect to the corresponding class of
models. In this section we present such extensions.

The use of nominals, and in particular pure axioms and existential saturation rules, allows us to characterize 
classes of models that are not definable without them. For instance, the 
hybrid axiom $@_i \neg \tup{\dowa} i$ forces the accessibility relation
related to $\dowa$ in the model to be \emph{irreflexive},
which cannot be expressed in the basic modal logic. In this way, we can express 
properties about the underlying topology of a data model and also impose restrictions about the data fields. For instance, 
the axiom $\tup{@_i=_e@_j}\ra\tup{@_i=_d@_j}$ for \emph{data inclusion} 
that will be discussed in Section~\ref{subsec:pure},
expresses that if two nodes coincide in the data field $e$, then they 
must also coincide in the data field $d$.

\paragraph{\bf Standard Translation.} In order to establish frame
conditions associated to pure axioms and existential saturation rules
we define the \emph{standard translation} of $\hxpd$ into first-order
logic, by mutual recursion between \textsf{NExp} and \textsf{PExp}.

\begin{definition}
\label{def:standtr}
The correspondence language for expressions of $\hxpd$ is a relational
language with a unary relation symbol $P_i$ for each $p_i \in \lab$, a
binary relation symbol $R_{\dowa}$ for each $\dowa\in\mo$ and a binary relation
symbol $D_e$ for each $e\in\eq$. Moreover, for each nominal $i\in\nom$ we will associate an indexed variable $x_i$. The function $\st'_x$ from
$\hxpd$-formulas into its correspondence language is defined by
$$
\begin{array}{lcl}
\st'_x(i) & = & x=x_i \\
\st'_x(p) & = & P(x) \\
\st'_x(\neg\varphi) &=& \neg\st'_x(\varphi) \\
\st'_x(\varphi\wedge\psi) & = & \st'_x(\varphi)\wedge\st'_x(\psi)\\
\st'_x(\tup{\alpha=_e\beta}) & =& \exists y,z. ~\st'_{x,y}(\alpha) \wedge \st'_{x,z}(\beta) \wedge D_e(y,z) \\
\st'_x(\tup{\alpha\neq_e\beta}) & =& \exists y,z. ~\st'_{x,y}(\alpha) \wedge \st'_{x,z}(\beta) \wedge \neg D_e(y,z),
\end{array}
$$
where $y,z$ are first-order variables which have not been used yet in
the translation, and where $\st'_{x,y}$ is defined as
$$
\begin{array}{lcl}
\st'_{x,y}(\dowa) & = & R_{\dowa}(x,y) \\
\st'_{x,y}(@_i) & =&  y=x_i \\
\st'_{x,y}([\varphi]) & = & x=y \wedge \st'_y(\varphi) \\
\st'_{x,y}(\alpha\beta) & = & \exists z. ~\st'_{x,z}(\alpha)\wedge\st'_{z,y}(\beta),
\end{array}
$$
with $z$ not used yet in the translation.  Finally, let $Eq(e)$ be the
first-order formula stating that the relation $D_e$ is an equivalence
relation (i.e., reflexive, symmetric and transitive), we define the 
\emph{standard translation} $\st_x$ of a formula $\varphi$ as
$$
\st_x(\varphi) = \bigwedge \{Eq(e) \mid {=}_e \mbox{ appearing in } \varphi \} \wedge \st'_x(\varphi).
$$ 
\end{definition}

Since the standard translation mimics the semantic clauses for
$\hxpd$-formulas, the following proposition holds.

\begin{proposition}
\label{prop:stpreserves}
Let $\varphi$ be an $\hxpd$-formula, and let 
$\model=\tup{M,\{\eqrel_e\}_{e\in\eq},\{R_{\dowa}\}_{\dowa\in\mo},\lbl, \\ \no}$
be an abstract hybrid data model with $m\in M$.
Then $\model,m\models\varphi$ if and only if
$\model,g\models\st_x(\varphi)$, where $g$ is an arbitrary
first-order assignment such that $g(x_i) = \no(i)$ and $g(x)=m$.
\end{proposition}

Notice, in the proposition above, that
$\model=\tup{M,\{\eqrel_e\}_{e\in\eq},\{R_{\dowa}\}_{\dowa\in\mo},\lbl,\no}$
is used to interpret $\st_x(\varphi)$ with $\lbl(p)$ interpreting the
unary relation symbol $P$, each accessibility relation $R_{\dowa}$
interpreting the binary relation symbol $R_{\dowa}$ and each relation
$\eqrel_e$ interpreting the binary relation $D_e$.

\paragraph{\bf Pure Axioms and Existential Saturation Rules.}
We have now all the ingredients needed to define pure axioms and existential saturation rules,
together with their associated frame conditions.

\begin{definition}[Pure Axioms]
  We say that a formula (or in particular, an axiom) is \emph{pure} if 
  it does not contain any
  occurrences of propositional symbols.  Let $\Pi$ be a set of pure
  axioms, we define $\fc(\Pi)$, the \emph{frame condition associated
    to $\Pi$} as the universal closure of the standard translation of
  the axioms in $\Pi$, i.e.,
$$
\fc(\Pi)=\bigwedge \set{\forall x.\forall {x_{i_1}} \ldots \forall{x_{i_n}}.\st_x(\varphi) \mid \varphi\in\Pi,
  i_1,\ldots,i_n \text{ all the nominals in $\varphi$}}. 
$$
\end{definition}

\begin{definition}[Existential Saturation Rules]\label{def:existrule} 
  Let $\varphi(i_1,\dots,i_n,j_1,\ldots,j_m)$ be an $\hxpd$ formula
  with no propositional symbols such that
  $i_1,\dots,i_n,j_1,\ldots,j_m$ is an enumeration of all nominals
  appearing in $\varphi$.  An {\em existential saturation rule} is a
  rule of the form
$$
\begin{array}{l@{\ \ \ }l}
\prooftree
\vdash \varphi(i_1,\dots,i_n,j_1,\ldots,j_m) \ra \psi
\justifies
\vdash \psi
\endprooftree &
\begin{minipage}{7cm} \small
provided that $j_1,\ldots,j_m$ do not occur in $\psi$.
\end{minipage}
\end{array}
$$ 
\noindent 
We will call $\varphi$ the \emph{head} of the rule, $i_1,\ldots,i_n$
its \emph{universally quantified nominals} (notation:
$\varphi_\forall$), and $j_1,\ldots,j_m$ its \emph{existentially
  quantified nominals} (notation: $\varphi_\exists$)\footnote{Notice
  that given a particular $\varphi$, different existential saturation
  rules can be defined depending on which nominals are listed in its side
  condition.}.
Let $\Rho$ be a set of existential saturation
rules, we define $\fc(\Rho)$, the \emph{frame condition associated to
  $\Rho$} as follows:
$$
\begin{array}{l@{\ }l}
\fc(\Rho)=\bigwedge \{\forall x.\forall x_{i_1} \ldots \forall x_{i_n}.\exists
  x_{j_1}\ldots \exists x_{j_m}.\st_x(\varphi) \mid 
& \text{$\varphi$ is the head of a rule $\rho \in \Rho$,}  \\
& \text{$i_k \in  \varphi_\forall$ and $j_l \in \varphi_\exists$\}}.
\end{array}
$$
\end{definition}

Let $\Pi$ be a set of pure formulas and let $\Rho$ be a set of
existential saturation rules, in what follows we write
$\ahxpd+\Pi+\Rho$ for the axiomatic system $\ahxpd$ extended with
$\Pi$ as additional axioms and $\Rho$ as additional inference rules.
In the coming sections we will show that this system is strongly
complete with respect to the class of models based on frames satisfying the frame condition $\fc(\Pi)\wedge\fc(\Rho)$.

\begin{example}
  The axiom $p\ra[\dowa]\tup{\dowb}p$ is not pure, since it contains
  the propositional symbol $p$. On the other hand,
  $\Pi=\set{i\ra[\dowa]\tup{\dowb}i, i\ra[\dowb]\tup{\dowa}i}$ is
  a set of pure axioms.  Its corresponding frame condition $\fc(\Pi)$
  indicates that $R_{\dowb}$ is the inverse of $R_{\dowa}$.  Indeed, $\fc(\Pi)$
  is equivalent to $\forall x.\forall y.(R_{\dowa}(x,y) \leftrightarrow R_{\dowb}(y,x))$.

  Consider the formula $@_i\tup{\dowa}i$. It can be used to define
  two different existential saturation rules:
$$
\begin{array}{l@{\ \ \ \ }l}
\rho_1 = \prooftree
\vdash @_i \tup{\dowa}i \ra \psi
\justifies
\vdash \psi
\endprooftree 
&
\mbox{\small (no side condition)} 
\end{array}
$$
and 
$$
\begin{array}{l@{\ \ \ \ }l}
\rho_2 = \prooftree
\vdash @_i \tup{\dowa}i \ra \psi
\justifies
\vdash \psi
\endprooftree 
&
\mbox{\small where $i$ does not occur in $\psi$.} \\
\end{array}
$$
$\fc(\{\rho_1\})$ is equivalent to $\forall x. R_{\dowa}(x,x)$ while
$\fc(\{\rho_2\})$ is equivalent to $\exists x.R_{\dowa}(x,x)$.

Consider now the rule corresponding to the \emph{Church-Rosser property} discussed in~\cite{BtC06}:
$$
\begin{array}{l@{\ \ \ \ }l}
\prooftree
\vdash (@_i \tup{\dowa}j \wedge @_i\tup{\dowa}k \ra @_j \tup{\dowa}l \wedge @_k\tup{\dowa}l) \ra \psi
\justifies
\vdash \psi
\endprooftree 
&
\mbox{\small where $l$ does not occur in $\psi$.}
\end{array}
$$
The frame condition corresponding to this rule is:
$$
\forall x.\forall x_i.\forall x_j.\forall x_k.\exists x_l.\st_x(@_i \tup{\dowa}j \wedge @_i\tup{\dowa}k \ra @_j \tup{\dowa}l \wedge @_k \tup{\dowa}l),
$$
which is equivalent to $\forall x_i.\forall x_j.\forall x_k.\exists x_l.(R_{\dowa}(x_i,x_j) \wedge R_{\dowa}(x_i,x_k) \to R_{\dowa}(x_j,x_l) \wedge R_{\dowa}(x_k,x_l))$.
\end{example}

It is not difficult to see that any pure axiom $\pi$ is equivalent to
the rule that uses $\pi$ as head without side conditions (i.e.,
$\pi_\exists$ is empty).  In that sense, any axiomatic system $\ahxpd
+ \Pi + \Rho$ is equivalent to some $\ahxpd + \Rho'$, as pure axioms
do not introduce additional expressive power.  On the other
hand, properties that mix both universal and existential
quantification like the Church-Rosser property mentioned above cannot
be captured using only pure axioms.  However, axioms are simpler than rules with side
conditions; and pure axioms are expressive enough to characterize many
interesting properties.

%%% Local Variables: 
%%% mode: latex
%%% TeX-master: "lmcs20"
%%% End: 

\subsection{Soundness and Completeness}
\label{subsec:complete}
It is a fairly straightforward exercise to prove that the axioms and
rules of $\ahxpd$ are sound with respect to the class of all
models. Similarly, any set $\Pi$ of pure axioms is sound with respect
to the class of models obtained from frames satisfying $\fc(\Pi)$, and any set $\Rho$ of existential saturation rules is sound with respect to the class of frames 
satisfying $\fc(\Rho)$ (the proof is similar to the one provided in~\cite{BtC06}).

\begin{theorem}[Soundness]
  Let $\Pi$ be a set of pure axioms and let $\Rho$ be a set of existential 
  saturation rules.
    All the axioms and rules from $\ahxpd$ are valid over the class of abstract 
  hybrid data models satisfying the frame condition $\fc(\Pi)\wedge\fc(\Rho)$. 
\end{theorem}
 
Now we will devote ourselves to show that the axiomatic system is also strongly complete. 
The completeness argument follows the lines of the completeness proof
for $\hl$ and similar approaches (see, e.g.,~\cite{Gold1984,BtC06,SchroderP10}), 
which is a Henkin-style proof with
nominals playing the role of first-order constants. 

We will prove in this section that the system $\ahxpd + \Pi + \Rho$ is
\emph{strongly complete} with respect to the class of abstract hybrid
data models obtained from frames satisfying $\fc(\Pi) \wedge \fc(\Rho)$.  More
precisely, given a particular extension $\ahxpd + \Pi + \Rho$, we will
show that if $\mathfrak{C} = \mathsf{Mod(\mathfrak{F})}$ for
$\mathfrak{F}$ the class of frames satisfying
$\fc(\Pi) \wedge \fc(\Rho)$, then $\Gamma\models_\mathfrak{C}\varphi$
implies $\Gamma\vdash_{\ahxpd + \Pi + \Rho}\varphi$, where $\Gamma\cup\set{\varphi}$
is a set of $\hxpd$-formulas. 
%
%
%However, that
%notion can be reformulated in the following way (we write $\models$
%and $\vdash$ for simplicity, instead of the formally correct %$\models_\mathfrak{C}$
%and $\vdash_{\ahxpd + \Pi + \Rho}$, respectively):
%\begin{align*}
%& \Gamma\models\varphi \mbox { implies } \Gamma\vdash \varphi \\
%\Leftrightarrow  ~~ & \Gamma\not\vdash\varphi \mbox{ implies } \Gamma\not\models\varphi \\
%\Leftrightarrow  ~~ & \Gamma\not\vdash\neg\varphi \ra \bot \mbox{ implies } \Gamma\not\models\varphi \\
%\Leftrightarrow  ~~ & \Gamma\cup\set{\neg\varphi}\not\vdash \bot \mbox{ implies } \Gamma\not\models\varphi \\
%\Leftrightarrow  ~~ & \Gamma\cup\set{\neg\varphi} \mbox{ consistent implies } \Gamma\not\models\varphi \\
%\Leftrightarrow  ~~ & \Gamma\cup\set{\neg\varphi} \mbox{ consistent implies there exists a model $\model,w$ s.t. } \model,w\models\Gamma\cup\{\neg\varphi\}.
%\end{align*}
%
Or, equivalently, we need to show that every consistent set of formulas (for
$\ahxpd + \Pi + \Rho$) is satisfiable in some abstract hybrid data
model (in $\mathfrak{C}$).
Recall that subscripts in relations $\vdash$ and $\models$ are ommited when they are clear from the context.

\begin{definition}\label{def:mcs}
  Let $\Gamma$ be a set of formulas, we say that $\Gamma$ is an
  $\ahxpd + \Pi + \Rho$ \emph{maximal consistent set} (MCS for
  short) if and only if $\Gamma\nvdash\bot$ and for all
  $\varphi\notin\Gamma$ we have $\Gamma\cup\{\varphi\}\vdash\bot$.
\end{definition}

\begin{proposition}\label{prop:mcs}
Let $\Gamma$ be an MCS. Then, the following facts hold:
\begin{enumerate}
	\item If $\set{i,\varphi}\subseteq\Gamma$ then $@_i\varphi\in\Gamma$.
	\item If $@_i\tup{\alpha=_e\beta}\in\Gamma$ then $\tup{@_i\alpha=_e@_i\beta}\in\Gamma$.
	\item If $\tup{\alpha=_e@_i\beta}\in\Gamma$ then $\tup{\alpha=_e@_j@_i\beta}\in\Gamma$.
\end{enumerate}	
\end{proposition}

\begin{proof}
Item 1 is a consequence of \emph{$@$-intro'}, 2 follows from \emph{$@{=}$-dist} and 
3 can be proved using {\em agree} and {\em $=$-comm}.	
\end{proof}

The next fact follows from the definition of MCS, as expected:

\begin{fact}\label{fact:mcs}
Let $\Gamma$ be an MCS. Then for all $\varphi$, either
$\varphi\in\Gamma$ or $\neg\varphi\in\Gamma$.	
\end{fact}

So far, we presented an axiom system together with the standard tools
for proving its completeness. We also introduced non-orthodox rules
(i.e., rules with side conditions), which will play a crucial role in the
Henkin-style model we will build for proving completeness.  
The \emph{paste} rule expresses that path expressions
can control what happens in accessible states from a named state. 
The \emph{name} rule says that if $\varphi$ is provable to hold
in an arbitrary state named by $j$, then $\varphi$ is also provable.
Now we introduce some properties that will be required in the construction of 
the Henkin-style model.

\begin{definition}[Named, Pasted and $\rho$-saturated sets]\label{def:pasted}
  Let $\Sigma$ be a set of $\hxpd$-formulas. 
\begin{itemize}
  \item We say that $\Sigma$ is \emph{named} if
  for some nominal $i$ we have that $i \in \Sigma$ (and we will 
  say that $\Sigma$ is named by $i$). 
  \item We say that $\Sigma$ is {\em pasted} if the following holds:
	\begin{enumerate}
	\item $\tup{@_i\dowa\alpha =_e\beta}\in \Sigma$
	implies that for some nominal $j$, $@_i\tup{\dowa}j\wedge\tup{@_j\alpha =_e \beta}\in\Sigma$. 
	\item $\tup{@_i\dowa\alpha \not =_e \beta}\in \Sigma$
	implies that for some nominal $j$,
	$@_i\tup{\dowa}j\wedge\tup{@_j\alpha \not =_e\beta}\in\Sigma$.
	\end{enumerate}
      \item For $\varphi$ a formula and
        $i_1,\ldots,i_n,j_1,\ldots,j_n \in \nom$, let
        $\varphi[i_1/j_1,\ldots,j_1/j_n]$ be the simultaneous
        substitution of $i_k$ by $j_k$, $1 \le k \le n$ in $\varphi$.
        Let $\rho$ be an existential saturation rule with head
        $\varphi$, $\varphi_\forall = i_1,\ldots,i_n$ and
        $\varphi_\exists = j_1,\ldots,j_m$.  Then $\Sigma$ is
        \emph{$\rho$-saturated} if for all nominals
        $i'_1,\ldots,i'_n \in \Sigma$,
        $\varphi[i_1/i'_1,\ldots,i_n/i'_n,j_i/j'_n,\ldots,j_m/j'_m]\in\Sigma$
        for some $j'_1,\ldots,j'_m \in \nom$. For $\Rho$ a set of
        existential saturation rules, \emph{$\Sigma$ is
          $\Rho$-saturated} if it is $\rho$-saturated for all
        $\rho \in \Rho$.
\end{itemize}
\end{definition}

Now we are going to prove a crucial property in our completeness
proof: the \emph{Extended Lindenbaum Lemma}. Intuitively, it says that
the rules of $\ahxpd + \Pi + \Rho$ allow us to extend MCSs to \emph{named and
  pasted} MCSs, provided we enrich the language with new
nominals. This lemma will be useful to obtain the models we need from
an MCS.

\begin{lemma}[Extended Lindenbaum Lemma]\label{lemma:lindenbaum}
Let $\nom'$ be a (countably) infinite set of nominals disjoint from $\nom$, and
let $\hxpd'$ be the language obtained by adding these new nominals to $\hxpd$.
Then, every consistent set of formulas in $\hxpd$ can be extended to
a named and pasted MCS in $\hxpd'$.
\end{lemma}

\begin{proof}
  Enumerate $\nom'$ and let $k$ be the first nominal in the
  enumeration. Given $\Sigma$ a consistent set in $\hxpd$,
  $\Sigma\cup\{k\}$ is consistent, otherwise for some conjunction
  $\theta$ from $\Sigma$, $\vdash k\ra \neg\theta$.  By the {\em name'}
  rule, $\vdash\neg\theta$, contradicting the consistency of $\Sigma$.

  Now enumerate all formulas in $\hxpd'$. Define
  $\Sigma^0 = \Sigma \cup \{k\}$ and suppose we have defined
  $\Sigma^n$, for $n\geq 0$.  Let $\varphi_{n+1}$ be the $(n+1)^{th}$
  formula in the enumeration of $\hxpd'$. Define $\Sigma_{n+1}$ as
  follows.  If $\Sigma^n\cup\{\varphi_{n+1}\}$ is inconsistent, then
  $\Sigma^{n+1}=\Sigma^n$. Otherwise:
\begin{enumerate}
	\item $\Sigma^{n+1}=\Sigma^n\cup\{\varphi_{n+1}\}$ if $\varphi_{n+1}$ is not of the form
	$\tup{@_i\dowa\alpha*\beta}$.
	\item $\Sigma^{n+1}=\Sigma^n \cup\{\varphi_{n+1}\}\cup\{@_i\tup{\dowa}j\wedge\tup{@_j\alpha*\beta}\}$,
	if $\varphi_{n+1}$ is of the form $\tup{@_i\dowa\alpha*\beta}$.
	Here $j$ is the first nominal in the enumeration that does not occur in any formula of $\Sigma^n$ or 
	$\tup{@_i\dowa\alpha*\beta}$.
\end{enumerate}

Let $\Sigma^\omega=\bigcup_{n\geq 0}\Sigma^n$. This set is named (by $k$), maximal and pasted.
Furthermore, it is consistent as a direct consequence of the {\em paste} rule.
\end{proof}

\begin{lemma}[Rule Saturation Lemma]\label{lemma:saturation}
  Let $\nom'$ be a (countably) infinite set of nominals disjoint from
  $\nom$, and let $\hxpd'$ be the language obtained by adding these
  new nominals to $\hxpd$.  Let $\Pi$ be a set of pure axioms, and
  $\Rho$ be a set of existential saturation rules.  Then, every
  consistent set of formulas in $\hxpd$ can be extended to a named,
  pasted and $\Rho$-saturated MCS in $\hxpd'$.
\end{lemma}

\begin{proof}
  Let $\Sigma$ be a consistent set of formulas in $\hxpd$.  First, we
  will show that $\Sigma$ can be extended to some $\Sigma^+$ such that
  for every rule $\rho\in \Rho$ with head $\varphi$,
  $\varphi_\forall = i_1\ldots i_n$, and
  $\varphi_\exists = j_1\ldots j_m$, and for all
  $i'_1,\ldots,i'_n \in\Sigma$, there are nominals
  $j'_1,\ldots,j'_m$ such that
  $\varphi(i_1/i'_1,\ldots,i_n/i'_k,j_1/j'_1,\ldots,j'_m)\in\Sigma^+$.

  Let $\nom'$ be a (countably) infinite set of nominals disjoint from
  $\nom$.  Given $\Sigma$ and $\Rho$, say that a pair $(\rho,\bar i)$
  is well-formed if $\rho \in \Rho$ with head $\varphi$, and $\bar i$
  is a sequence of nominals in $\Sigma$ with length equal to the
  length of $\varphi_\forall$.  For $\Sigma$ and $\Rho$ countable, the
  set of well-formed pairs is countable and can be enumerated. Define
  $\Sigma^0 = \Sigma$.  Let $(\rho_{n+1}, i'_1 \ldots i'_k)$ be the
  $(n+1)^{th}$ pair in the enumeration.  Define $\Sigma^{n+1}$ as
  follows.

\begin{itemize}
\item[]
  $\Sigma^{n+1}=\Sigma^n \cup
  \{\varphi(i_1/i'_1,\ldots,i_k/i'_k,j_1/j'_1,\ldots,j_m/j'_m)\}$
  where $\varphi$ is the head of $\rho_{n+1}$,
  $\varphi_\forall = i_1,\ldots,i_k$,
  $\varphi_\exists = j_1,\ldots,j_m$, and $j'_1,\ldots,j'_m$ are the
  first $m$ nominals in $\nom'$ that do not occur in any formula of $\Sigma^n$.
\end{itemize}

We define $\Sigma^+ = \bigcup_{n\geq 0} \Sigma^n$, which is consistent
and extends $\Sigma$ as described above.

Now consider any consistent set of formulas $\Gamma$.  Let
$\Gamma^0=\Gamma$, and for all $n\geq 0$ let $\Gamma^{n+1}$ be a named
and pasted MCS extending $(\Gamma^n)^+$ (which exists by
Lemma~\ref{lemma:lindenbaum}). Then we have the following chain of
inclusions:
$$
\Gamma=\Gamma^0 \subseteq (\Gamma^0)^+ \subseteq \Gamma^{1} \subseteq (\Gamma^1)^+ \subseteq \ldots
$$
Let $\Gamma^\omega=\bigcup_{n\geq 0}\Gamma^n$. Then $\Gamma^\omega$ is
a named, pasted and $\Rho$-saturated MCS in $\hxpd'$. As we used only
countably many new nominals, this set is also countable.
\end{proof}

%\begin{corollary}\raul{por si lo necesitamos}
%Let $\Sigma$ be a named and pasted $\ahxpd$-MCS, then if $@_i\tup{\delta\alpha=\beta}\in\Sigma$
%implies $@_i\tup{\delta}j \wedge @_i\tup{@_j\alpha=\beta}\in\Sigma$.
%\end{corollary}

From a named and pasted MCS we can extract a model:

\begin{definition}[Extracted Model]\label{def:extracted}
  Let $\Gamma$ be a named and pasted MCS\footnote{For this definition,
    it is irrelevant whether $\Gamma$ is $\Rho$-saturated.}.  Define
  $\model_\Gamma=\tup{M, \{\eqrel_e\}_{e\in\eq}, \{R_{\dowa}\}_{\dowa\in\mo},
    \lbl, \no}$, the \emph{extracted model} from $\Gamma$, as
	\begin{itemize}
	\item $M=\set{\Delta_i \mid \Delta_i = \set{j \mid @_ij \in \Gamma}}$ 
	\item $\Delta_i R_{\dowa} \Delta_j$ iff $@_i\tup{\dowa}j \in \Gamma$
	\item $p\in\lbl(\Delta_i)$ iff $@_ip\in\Gamma$
	\item $\no(i)=\Delta_i$
	\item $\Delta_i \eqrel_e \Delta_j$ iff $@_i\tup{\epsilon=_e@_j}\in\Gamma$.
	\end{itemize}	
\end{definition}

We need to prove that  $\model_\Gamma$ is well defined, and that it is actually
an abstract hybrid data model.

\begin{proposition} \ 
\begin{enumerate}
\item $\Delta_i = \Delta_j$ implies $@_i\varphi \in \Gamma$ iff $@_j\varphi \in \Gamma$. 

\item For all $e\in\eq$, $\eqrel_e$ is an equivalence relation.
\end{enumerate}
\end{proposition}

\begin{proof}
Item 1 ensures that the definition of $\model_\Gamma$ does not depend of
the particular nominal taken as representative of $\Delta_i$.  The
property follows directly from \emph{bridge}.

For item 2, we prove:

        \noindent {\em- Reflexivity:}  We need to show that $\Delta_i\eqrel_e\Delta_i$, which by definition is equivalent to
        $@_i\tup{\epsilon=_e@_i}\in\Gamma$.
        By  \emph{equal} we have $\tup{\epsilon=_e\epsilon}\in\Gamma$,
        and applying \emph{Nec} we get $@_i\tup{\epsilon=_e\epsilon}\in\Gamma$. Also, by \emph{comp=-dist}, $\tup{@_i=@_i}\in\Gamma$, and by \emph{agree} $\tup{@_i=_e@_i@_i}\in\Gamma$. Then, by \emph{$@=$-dist},
        $@_i\tup{\epsilon=_e@_i}\in\Gamma$, as wanted.

        \noindent {\em- Symmetry:} $\Delta_i\eqrel_e\Delta_j$ iff
        $@_i\tup{\epsilon=_e@_j}\in\Gamma$.  By {\em neutral} and {\em
          $=$-comm} we get
        $@_i\tup{\epsilon=_e@_j\epsilon}\in\Gamma$. Then, by {\em
          $@$-swap}
        $@_j\tup{\epsilon=_e@_i\epsilon}\in\Gamma$. Therefore (by
        \emph{neutral}) $@_j\tup{\epsilon=_e@_i}\in\Gamma$.

	\noindent {\em- Transitivity:} Suppose
        $\Delta_i\eqrel_e\Delta_j$ and $\Delta_j\eqrel_e\Delta_k$, iff
        $@_i\tup{\epsilon=_e@_j}\in\Gamma$ and
        $@_j\tup{\epsilon=_e@_k}\in\Gamma$.  This means that we have
        (by \emph{$@$-swap}) $@_j\tup{\epsilon=_e@_i}\in\Gamma$, and
        $@_j\tup{\epsilon=_e@_k}\in\Gamma$.  Then
        $@_j(\tup{\epsilon=@_i}\wedge\tup{\epsilon=_e@_k})\in\Gamma$ (let us call this fact $(\dagger)$ for later use).
        On the other hand,
        $\tup{\epsilon=_e@_i}\wedge\tup{\epsilon=_e@_k}\ra\tup{@_i=_e@_k}\in\Gamma$
        because of {\em $\epsilon$-trans}. Then (by {\em
          Nec})
        $@_j(\tup{\epsilon=_e@_i}\wedge\tup{\epsilon=_e@_k}\ra\tup{@_i=_e@_k})\in\Gamma$,
        and by {\em K} and distributivity of $@$ with respect to
        $\wedge$,
        $@_j\tup{\epsilon=_e@_i}\wedge@_j\tup{\epsilon=_e@_k}\ra
        @_j\tup{@_i=_e@_k}\in\Gamma$.
        By Modus Ponens with $(\dagger)$ we have
        $@_j\tup{@_i=_e@_k}\in\Gamma$, and by {\em comp=-dist}
        $\tup{@_j@_i=_e@_j@_k}\in\Gamma$. Using {\em back} and {\em
          agree} we get $\tup{@_i=_e@_i@_k}\in\Gamma$. Hence by
        \emph{$@{=}$-dist}, $@_i\tup{\epsilon=_e@_k}\in\Gamma$, which
        gives us $\Delta_i\eqrel_e\Delta_k$.
\end{proof}

\begin{proposition}\label{prop:addprop}
Let $\model_\Gamma=\tup{M,\{\eqrel_e\}_{e\in\eq},\{R_{\dowa}\}_{\dowa\in\mo},\lbl,\no}$ be the extracted 
model, for some $\Gamma$.
Then, 
\begin{enumerate}
%	 \item $\Delta_i\ra\Delta_j$ if and only if $\tup{\upw}i\in\Delta_j$, and
	 \item $\Delta_i\not\eqrel_e\Delta_j$ if and only if $@_i\tup{\epsilon\neq_e @_j}\in\Gamma$,
	 \item if $\Delta_i R_{\dowa} \Delta_j$ then for all $@_j\varphi\in\Gamma$, $@_i\tup{\dowa}\varphi\in\Gamma$.
\end{enumerate}
\end{proposition}

\begin{proof}
%Item $1$ uses the same argument as for $\hl$ in addition to the
%axioms for $\upw$; i
Item $1$ follows from  {\em $@$-data}; for item $2$
suppose $\Delta_i R_{\dowa}\Delta_j$, then we have $@_i\tup{\dowa}j\in\Delta_i$. 
Let $@_j\varphi\in\Gamma$, then by \emph{comp-dist} we have
$\tup{@_i\dowa}j \in\Gamma$, hence by $bridge$ we get $\tup{@_i\dowa}\varphi\in\Gamma$.
Therefore, by {\em $@*$-dist}, $@_i\tup{\dowa}\varphi\in\Gamma$.
\end{proof}

Now, given a named and pasted MCS $\Gamma$ we can prove the
following lemma:

\begin{lemma}[Existence Lemma]\label{lemma:existence} 
Let $\Gamma$ be an MCS and  $\model_\Gamma=\langle M,\{\eqrel_e\}_{e\in\eq}, 
\{R_{\dowa}\}_{\dowa\in\mo}, \lbl$, $\no\rangle$ be
the extracted model from $\Gamma$. Let $\Delta_i\in M$, then
\begin{enumerate}
 	\item $@_i\tup{\dowa\alpha=_e\beta}\in\Gamma$ implies there exists
          $\Delta_j\in M$ s.t.\ $\Delta_i R_{\dowa} \Delta_j$ and $@_j\tup{\alpha=_e@_i\beta}\in\Gamma$.
 	\item $@_i\tup{\dowa\alpha\neq_e\beta}\in\Gamma$ implies there exists
          $\Delta_j\in M$ s.t.\ $\Delta_i R_{\dowa} \Delta_j$  and $@_j\tup{\alpha\neq_e@_i\beta}\in\Gamma$. 	
    \item $@_i\tup{@_j\alpha=_e@_k\beta}\in\Gamma$ implies 
    %there exists  $\Sigma\in M$ s.t.
    \ $@_j\tup{\alpha=_e@_k\beta}\in\Gamma$.
 	\item $@_i\tup{@_j\alpha\not=_e@_k\beta}\in\Gamma$ implies 
 	%there exists     $\Sigma\in M$ s.t.
 	\ $@_j\tup{\alpha\not=_e@_k\beta}\in\Gamma$.
 \end{enumerate}  
\end{lemma}

\begin{proof}
  First, we discuss $1$ ($2$ is similar). 
  We suppose $@_i\tup{\dowa\alpha=_e\beta}\in\Gamma$.  
  Then, by {\em $@{=}$-dist},
  $\tup{@_i\dowa\alpha=_e@_i\beta}\in\Gamma$. Because $\Gamma$ is
  pasted, $@_i\tup{\dowa}j\wedge\tup{@_j\alpha=_e@_i\beta}\in\Gamma$.
  As $\Gamma$ is an MCS, $@_i\tup{\dowa}j\in\Gamma$ and
  $\tup{@_j\alpha=_e@_i\beta}\in\Gamma$.  By definition of $\model_\Gamma$
  we obtain $\Delta_i R_{\dowa} \Delta_j$, and by {\em agree}, we have
  $\tup{@_j\alpha=_e@_j@_i\beta}\in\Gamma$. Then,
  $@_j\tup{\alpha=_e@_i\beta}\in\Gamma$ by {\em $@{=}$-dist}. 

  \medskip For $3$ ($4$ is similar) assume
  $@_i\tup{@_j\alpha=_e@_k\beta}\in\Gamma$, then by {\em comp=-dist},
  $\tup{@_i@_j\alpha=_e@_i@_k\beta}\in\Gamma$. Applying {\em agree}
  twice, we have $\tup{@_j\alpha=_e@_j@_k\beta}\in\Gamma$, then by
  {\em $@{=}$-dist} $@_j\tup{\alpha=_e@_k\beta}\in\Gamma$.
\end{proof}

\begin{corollary}\label{coro:canonrel}
Let $\Gamma$ be an MCS	and let
$\model_\Gamma=\tup{M, \{\eqrel_e\}_{e\in\eq}, \{R_{\dowa}\}_{\dowa\in\mo}, \lbl, \no}$ 
be the extracted model, and $\Delta_i\in M$. If $@_i\tup{\dowa\alpha}\varphi\in\Gamma$,
then there exists $\Delta_j\in M$ such that $\Delta_i R_{\dowa} \Delta_j$ and
$@_j\tup{\alpha}\varphi\in\Gamma$.
\end{corollary}

\begin{proof}
  Let $\Delta_i\in M$.  By hypothesis,
  $@_i\tup{\dowa\alpha[\varphi]=_e\dowa\alpha[\varphi]}\in\Gamma$,
  then by Lemma~\ref{lemma:existence} there exists $\Delta_j \in M$
  such that $\Delta_i R_{\dowa} \Delta_j$ and
  $@_j\tup{\alpha[\varphi]=_e@_i\alpha[\varphi]}\in\Gamma$.  By {\em
    comp{=}-dist},
  $\tup{@_j\alpha[\varphi]=_e@_j@_i\alpha[\varphi]}\in\Gamma$, and by
  {\em comp-neutral} and {\em subpath} we get
  $\tup{@_j\alpha[\varphi]}\top\in\Gamma$.  Then, using {\em
    comp-dist}, {\em comp-assoc} and {\em =-test}, we have
  $@_j\tup{\alpha}\varphi\in\Gamma$.
\end{proof}

Now we are ready to prove the Truth Lemma that states that membership
in an MCS generating an extracted model is equivalent to being true at
a state in the extracted model. First let us introduce a notion of
size for node and path expressions, which we will use in the inductive
cases of the proof.

\begin{definition}\label{def:size}
We define inductively the size of a path and node expression (notation $\size\cdot$)
as follows:
$$
\begin{array}{l@{~=~}l@{\quad \quad \quad }l@{~=~}l}
\size{\dowa} & 2 & \size p & 1, ~p\in\lab\cup\nom	\\
\size{@_i} & 1 & \size{\neg\varphi} & \size\varphi \\
\size{[\varphi]} & 1+\size\varphi & \size{\varphi\wedge\psi} & \size\varphi + \size\psi \\
\size{\alpha\beta} & \size\alpha + \size\beta & \size{\tup{\alpha*\beta}} & \size\alpha + \size\beta,
\end{array}
$$	
where $\alpha,\beta$ are path expressions and $\varphi,\psi$ are node expressions.
\end{definition}

\begin{lemma}[Truth Lemma]\label{lemma:truth} 
	Let $\model_\Gamma=\tup{M, \{\eqrel_e\}_{e\in\eq}, \{R_{\dowa}\}_{\dowa\in\mo}, \lbl, \no}$ 
	be the extracted model from an MCS $\Gamma$, and let $\Delta_i\in M$. Then, for any formula
	$\varphi$,
	$$
	\model_\Gamma,\Delta_i\models\varphi ~~ \iff ~~ @_i\varphi\in\Gamma.
	$$
\end{lemma}

\begin{proof}
	In fact we will prove a stronger result. 
	Let $\Delta_i,\Delta_j\in M$, $\varphi$ be a node expression and $\alpha$ be a path 
	expression.
	\begin{description}
		\item[(IH1)] $\model_\Gamma,\Delta_i\models\varphi ~~ \iff ~~ @_i\varphi \in \Gamma.$
		\item[(IH2)] $\model_\Gamma,\Delta_i,\Delta_j\models\alpha ~~ \iff ~~ @_i\tup{\alpha}j \in \Gamma.$
	\end{description}

The proof is a double induction argument, proceeding first on the
structural complexity of $\varphi$ and $\alpha$, and then on the size
of path-formulas as per Definition~\ref{def:size}.
First, we prove the base cases:

\medskip

	\noindent {- $\alpha=\dowa$:} Suppose $\evpath\models\dowa$, iff $\Delta_i R_{\dowa} \Delta_j$ (by $\models$), 
	iff $@_i\tup{\dowa}j\in\Gamma$ (by definition of extracted model). \smallskip

	\noindent {- $\alpha=@_k$:} Suppose $\evpath\models@_k$, iff $\no(k)=\Delta_j$. But by definition of $\no$,
	$\Delta_j=\Delta_k$, and because we know that $j\in\Delta_j$ we have $j\in\Delta_k$. Then, we have $@_kj\in\Gamma$,
	and by {\em agree}, $@_i@_kj\in\Gamma$. \smallskip

	\noindent {- $\varphi=p$:} $\evnode\models p$ iff $p\in\lbl(\Delta_i)$, iff $@_ip\in\Gamma$. \smallskip

	\noindent {- $\varphi=j$:} $\evnode\models j$ iff $\no(j)=\Delta_i$, iff $\Delta_i=\Delta_j$ iff $j\in\Delta_i$, iff
	$@_ij\in\Gamma$. \medskip

Now we prove the inductive cases:

\medskip

	\noindent {- $\varphi=\psi\wedge\theta$ and $\varphi=\neg\psi$} are direct from (IH1). \smallskip

	\noindent {- $\alpha=[\psi]$:} $\evpath\models[\psi]$ iff $\Delta_i=\Delta_j$ and $\evnode\models\psi$.
	By (IH1), we have $@_i\psi\in\Gamma$ and $j\in\Delta_i$, then $@_ij\in\Gamma$. As $\Gamma$ is an MCS, we have
	$@_i(\psi\wedge j)\in\Gamma$, and by idempotence of the conjunction we have $@_i(\psi\wedge\psi\wedge j \wedge j)\in\Gamma$. 
	Also, $\tup{\epsilon=_e\epsilon}$ is a theorem, by {\em Nec} we have $@_i\tup{\epsilon=_e\epsilon}\in\Gamma$, then
	$@_i(\psi\wedge\psi\wedge j \wedge j \wedge \tup{\epsilon=_e\epsilon})\in\Gamma$.
	Using {\em $=$-test} and {\em $=$-comm} we obtain $@_i\tup{[\psi][j]=_e[\psi][j]}\in\Gamma$
	(which is the same as $@_i\tup{[\psi]}j$) as we wanted. \smallskip

	%\noindent {- $\alpha=\beta\gamma$:} $\evpath\models\beta\gamma$ iff there is some $\Delta_k\in M$ such that
	%$\model_\Gamma,\Delta_i,\Delta_k\models\beta$ and $\model_\Gamma,\Delta_k,\Delta_j,\models\gamma$.
	%By (IH2), we have $@_i\tup\beta k\in\Gamma$ and $@_k\tup\gamma j \in\Gamma$, then
	%$@_i\tup\beta k \wedge @_k\tup\gamma j\in\Gamma$. 
	%By {\em agree}, we have $@_i\tup\beta k \wedge @_i@_k\tup\gamma j\in\Gamma$, and with a very simple	argument purely based on hybrid logic, we get $@_i(\tup\beta k \wedge @_k\tup\gamma j)\in\Gamma$. 
	%By {\em bridge}, we have $@_i(\tup\beta\tup\gamma j)\in\Gamma$, 
	%hence by {\em comp-dist} $@_i(\tup{\beta\gamma}j) \in\Gamma$. 
	
	\noindent {- $\alpha=\dowa\beta$:} For the left to right implication, let $\evpath\models\dowa\beta$ iff there is some $\Delta_k\in M$ such that
	$\model_\Gamma,\Delta_i,\Delta_k\models\dowa$ and $\model_\Gamma,\Delta_k,\Delta_j,\models\beta$.
	By (IH2), we have $@_i\tup\dowa k\in\Gamma$ and $@_k\tup\beta j \in\Gamma$, then
	$@_i\tup\dowa k \wedge @_k\tup\beta j\in\Gamma$. 
	By {\em agree}, we have $@_i\tup\dowa k \wedge @_i@_k\tup\beta j\in\Gamma$, and with a simple	hybrid logic argument, we get $@_i(\tup\dowa k \wedge @_k\tup\beta j)\in\Gamma$. 
	By {\em bridge}, we have $@_i\tup\dowa\tup\beta j\in\Gamma$, 
	hence by {\em comp-dist} $@_i\tup{\dowa\beta}j \in\Gamma$. 
	
	For the other direction, assume $@_i\tup{\dowa\beta}j \in\Gamma$.  Then, by 
  Corollary~\ref{coro:canonrel} there exists some $k$ such that $\Delta_i R_{\dowa} \Delta_k$ and $@_k\tup{\beta}j \in\Gamma$. 
  %i.e., $\tup{@_i\dowa\beta[j] = @_i\dowa\beta[j]} \in\Gamma$.
  %the 
	%Extended Lindenbaum Lemma (Lemma~\ref{lemma:lindenbaum}) we have $@_i\tup\dowa k \wedge  \tup{@_k\beta[j]= @_i\dowa\beta[j]} \in \Gamma$, for some $k$. Hence, both $@_i\tup\dowa k$
	%and $\tup{@_k\beta[j]= @_i\dowa\beta[j]}$ are in $\Gamma$.  From the first, by (IH2) we can easily get that $\model_\Gamma,\Delta_i,\Delta_k\models\dowa$. By 
	%\emph{subpath}, from $\tup{@_k\beta[j]= @_i\dowa\beta[j]} \in \Gamma$ we get $\tup{@_k\beta[j]}\top \in \Gamma$.  By \emph{comp-dist} and \emph{$*$-test}, 
	%$@_k\tup{\beta}j \in \Gamma$. 
  Then, applying (IH2) twice, we get $\model_\Gamma,\Delta_i,\Delta_k\models\dowa$ and $\model_\Gamma,\Delta_k,\Delta_j\models\beta$. 
  Therefore, $\model_\Gamma,\Delta_i,\Delta_j\models\dowa\beta$ as needed. 
	
  \smallskip
  \noindent {- $\alpha=@_k\beta$ and $\alpha=[\psi]\beta$:} these cases are similar to the previous one.
	
	\medskip

\noindent 
For node expressions of the form $\tup{\alpha*\beta}$ we need to do
induction on the size of $\alpha$ and $\beta$.  Notice that by {\em
  $*$-comm}, $@_i\tup{\alpha*\beta}\in\Gamma$ iff
$@_i\tup{\beta*\alpha}\in\Gamma$.  And by the semantic definition,
$\evnode\models\tup{\alpha*\beta}$ iff
$\evnode\models\tup{\beta*\alpha}$.  So we need only carry out the inductive steps for $\alpha$.  Moreover, by {\em comp-neutral},
$\vdash\tup{\alpha*\beta}\lra\tup{\alpha\epsilon*\beta}$ which is also
a validity. So we can assume that every path ends in a test.  The base
case then is when $\size{\alpha}+\size{\beta}=2$, and both $\alpha$
and $\beta$ are tests. \smallskip

	\noindent {- $\varphi=\tup{[\psi]=_e[\theta]}$:} direct from {\em test-dist}. \smallskip

	\noindent {- $\varphi=\tup{[\psi]\neq_e[\theta]}$:} 
By \emph{test-$\bot$} the formula is equivalent to $\bot$, hence
not in $\Gamma$. Moreover, by the defined semantics the formula is unsatisfiable. \smallskip

	\noindent Now, let us consider $\size\alpha+\size\beta \ge 3$: \smallskip

	\noindent {- $\varphi=\tup{\dowa\beta=_e\gamma}$:} Let us prove
        the right to left direction. Suppose
        $@_i\tup{\dowa\beta=_e\gamma}\in\Gamma$. By the Existence Lemma,
        there is $\Delta_j\in M$ such that $\Delta_i R_{\dowa} \Delta_j$,
        and $@_j\tup{\beta=_e@_i\gamma}\in\Gamma$, for some
        $j\in\nom$.  Notice that
        $\size{@_j\tup{\beta=_e@_i\gamma}}\leq\size{@_i\tup{\dowa\beta=_e\gamma}}$.
        Applying (IH1) we obtain
        $\model_\Gamma,\Delta_j\models\tup{\beta=_e@_i\gamma}$, so 
        there exists $\Delta_1,\Delta_2\in M$ such that
	\begin{enumerate}
		\item $\model_\Gamma,\Delta_j,\Delta_1\models\beta$,
		\item $\model_\Gamma,\Delta_j,\Delta_2\models@_i\gamma$,
		\item $\Delta_1\eqrel_e\Delta_2$.
	\end{enumerate}

	From $1$ and $\Delta_i R_{\dowa} \Delta_j$ we get $\model_\Gamma,\Delta_i,\Delta_1\models\dowa\beta$ and
	from $2$ and the semantic interpretation of $@$ we get $\model_\Gamma,\Delta_i,\Delta_2\models\gamma$. 
	Then, together with $3$ we get $\model_\Gamma,\Delta_i\models\tup{\dowa\beta=_e\gamma}$, as wanted.
	\smallskip

	For the other direction, suppose
	 $\evnode\models\tup{\dowa\beta=_e\gamma}$, iff there are $\Delta_j,\Delta_k$ such that
	 $\evnode,\Delta_j\models\dowa\beta$, $\evnode,\Delta_k\models\gamma$ and $\Delta_j\eqrel_e\Delta_k$. First notice that $\evnode,\Delta_j\models\dowa\beta$ iff there exists $\Delta_t \in M$ s.t. $\evnode,\Delta_t \models \dowa$ and 
   $\model_\Gamma,\Delta_t,\Delta_j\models\beta$. Then we have, $\Delta_i R_{\dowa} \Delta_t$, and by definition of $\model_\Gamma$ we get $@_i\tup{\dowa}t \in \Gamma$. On the other hand, by (IH2) on $\model_\Gamma,\Delta_t,\Delta_j\models\beta$ we obtain $@_t\tup{\beta}j\in\Gamma$. Since $\set{@_i\tup{\dowa}t,@_t\tup{\beta}j}\subseteq\Gamma$, by {comp-dist} we have: 

   \begin{enumerate}
    \item $@_i\tup{\dowa\beta}j\in\Gamma$.
   \end{enumerate}

%	 Then, by (IH2) and definition of $\model_\Gamma$ we have:

   By (IH2) on $\evnode,\Delta_k\models\gamma$, and by definition of $\model_\Gamma$ together with $\Delta_j\eqrel_e\Delta_k$ we obtain, respectively:
	 \begin{enumerate}
    \setcounter{enumi}{1}
	 	\item $@_i\tup\gamma k \in \Gamma$, and
		\item $@_j\tup{\epsilon=_e@_k}\in\Gamma$. 
	 \end{enumerate}

	 By $1$ and Corollary~\ref{coro:canonrel} there exists $\Delta_l$ such that 
	 \begin{enumerate}
	 \setcounter{enumi}{3}
	 	\item $@_l\tup\beta j\in\Gamma$. 
	 \end{enumerate}

	 $(\ddagger)$ From $2$ we have $@_i\tup\gamma k\in\Gamma$ and from $3$ we can obtain
	 $@_k\tup{\epsilon=_e@_j}\in\Gamma$, 
	 then we have $\tup{@_i\gamma}k \wedge @_k\tup{\epsilon=_e@_j} \in \Gamma$, by 
	 {\em comp-dist}. By {\em bridge}, $\tup{@_i\gamma}\tup{\epsilon=_e@_j}\in\Gamma$, then 
	 by {\em comp$=$-dist} and {\em back}, we get $\tup{@_i\gamma=_e@_j}\in\Gamma$.
	 Applying {\em $=$-comm, comp-neutral, agree} and {\em $@$-dist},
	 $@_j\tup{\epsilon=_e@_i\gamma}\in\Gamma$.
	 
	 Also, from $4$ we have $@_l\tup{\beta}j\in\Gamma$, then
	 $@_j\tup{\epsilon=_e@_i\gamma}\wedge\tup{@_l\beta}j\in\Gamma$ (by MCS and {\em comp-dist}), 
	 and by {\em bridge} we get $\tup{@_l\beta}\tup{\epsilon=_e@_i\gamma}\in\Gamma$. 
	 By {\em comp$=$-dist and comp-neutral}, $\tup{@_l\beta=_e@_l\beta@_i\gamma}\in\Gamma$,
	 then by {\em back}, {\em agree} and {\em $@{=}$-dist} we have
	 $@_l\tup{\beta=_e@_i\gamma}\in\Gamma$.
	 Then we have
	 \begin{align*}
   & @_l\tup{\beta=_e@_i\gamma}\in\Gamma  \\
  \Ra ~ &  @_i\tup{\dowa}\tup{\beta=_e@_i\gamma}\in\Gamma \tag{$\Delta_i R_{\dowa} \Delta_l$, Prop.~\ref{prop:addprop} item 2} \\
    \Ra ~ &  @_i\tup{\dowa\beta=_e\dowa@_i\gamma}\in\Gamma \tag{{\em comp=-dist}}\\
  	\Ra ~ &  @_i\tup{\dowa\beta=_e@_i\gamma}\in\Gamma \tag{{\em back}} \\  
  	\Ra ~ &  \tup{@_i\dowa\beta=_e@_i@_i\gamma}\in\Gamma \tag{{\em comp=-dist}} \\
    \Ra ~ &  \tup{@_i\dowa\beta=_e@_i\gamma}\in\Gamma \tag{\em back} \\
   	\Ra ~ &  @_i\tup{\dowa\beta=_e\gamma}\in\Gamma \tag{\em $@$=-dist} 
\end{align*} 

	 \noindent {- $\varphi=\tup{[\psi]\beta=_e\gamma}$:} 
   Let us prove the right to left direction. Suppose $@_i\tup{[\psi]\beta=_e\gamma}\in\Gamma$.
   By {\em $=$-test}, we have $@_i(\psi\wedge\tup{\beta=_e\gamma})\in\Gamma$, iff 
   (by using an argument purely based on modal reasoning), $@_i\psi\in\Gamma$ and $@_i\tup{\beta=_e\gamma}\in\Gamma$.
   By (IH1) on both expressions we get $\evnode\models\psi$ and $\evnode\models\tup{\beta=_e\gamma}$. Then, by the semantic interpretation of $\wedge$ and $[\psi]$, we 
   get $\evnode\models\tup{[\psi]\beta=_e\gamma}$, as wanted.
   
   For the other direction, suppose $\evnode\models\tup{[\psi]\beta=_e\gamma}$,  
	 iff there are $\Delta_j,\Delta_k$ such that
	 $\evnode,\Delta_j\models[\psi]\beta$, $\evnode,\Delta_k\models\gamma$ and $\Delta_j\eqrel_e\Delta_k$.
	 Then, by (IH2) and definition of $\model_\Gamma$ we have:
	 \begin{enumerate}
	 	\item $@_i\tup{\beta}j\in\Gamma$,
	 	\item $@_i\tup\gamma k \in \Gamma$,
		\item $@_j\tup{\epsilon=_e@_k}\in\Gamma$, and
		\item $@_i\psi\in\Gamma$. 
	 \end{enumerate}

	 Using the same argument as in $(\ddagger)$ the proof that $@_i\tup{[\psi]\beta=_e\gamma}\in\Gamma$
	 is straightforward. 

	 \smallskip

	\noindent {- $\varphi=\tup{@_j\beta=_e\gamma}$:} For the left to right direction suppose
	$\evnode\models\tup{@_j\beta=_e\gamma}$, iff there are $\Delta_k,\Delta_l$ such that		
	 $\evnode,\Delta_k\models@_j\beta$, $\evnode,\Delta_l\models\gamma$ and $\Delta_k\eqrel_e\Delta_l$.
	 Then, by (IH2) and definition of $\model_\Gamma$ we have:
	 \begin{enumerate}
	 	\item $@_i\tup{@_j\beta}k\in\Gamma$ iff $@_j\tup\beta k\in \Gamma$,
	 	\item $@_i\tup\gamma l \in \Gamma$, and
		\item $@_l\tup{\epsilon=_e@_k}\in\Gamma$, iff $@_k\tup{\epsilon=_e@_l}\in\Gamma$. 
	 \end{enumerate}

	 By $1$ and $3$, applying {\em bridge} we get $@_j\tup\beta\tup{\epsilon=_e@_l}\in\Gamma$, iff (by {\em comp$=$-dist})
	 $@_j\tup{\beta=_e\beta @_l}\in\Gamma$. By {\em back}, we get  $@_j\tup{\beta=_e@_l}\in\Gamma$,
	 which is equivalent to $@_l\tup{\epsilon=_e@_j\beta}\in\Gamma$ (by {\em agree} and {\em comp=-dist}).
	 Together with $2$ and {\em bridge} we get $@_i\tup{\gamma}\tup{\epsilon=_e@_j\beta}\in\Gamma$,
	 hence $@_i\tup{\gamma=_e\gamma @_j\beta}\in\Gamma$ iff (by {\em back} and $=$-comm)
	 $@_i\tup{@_j\beta=_e\gamma}\in\Gamma$. 
	 \smallskip

    For the other direction suppose $@_i\tup{@_j\beta=_e\gamma}\in\Gamma$. 
    First notice that, in all the cases we considered so far, the induction is on the path expression appearing on the left side of the $=$. However, an analogous argument can be applied if we do induction in the path expression on the right side of the $=$, by {\em $=$-comm}. 
    Suppose we proceed as above for the node expression $\tup{@_j\beta=_e\gamma}$. If we apply the exact same steps, we will find out that this time we need
    to do induction on $\gamma$; but, as we mentioned, the cases for $\dowa$ and $[\varphi]$ 
    are symmetric in both sides of the $=$. As a consequence, it all boils down to consider only the case $\gamma=@_k\eta$.
    
    Suppose $@_i\tup{@_j\beta=_e@_k\eta}\in\Gamma$. By Existence Lemma
    we have $@_j\tup{\beta=_e@_k\eta}\in\Gamma$, hence by (IH1)
    $\model_\Gamma,\Delta_j\models\tup{\beta=_e@_k\eta}$. By semantics of $@$, and 
    the fact that $\model_\Gamma$ is named,
    $\evnode\models\tup{@_j\beta=_e@_k\eta}$. 
	\smallskip 

	 \noindent {- The cases involving $\neq_e$ are analogous, using item $1$ from 
	 Proposition~\ref{prop:addprop} to obtain $@_j\tup{\epsilon=_e@_k}\notin\Gamma$ 
	 in item $3$ above.}
	 \end{proof}

\begin{lemma}[Frame Lemma]\label{lemma:frame} 
If $\Gamma$ is a named, pasted and $\Rho$-saturated MCS of $\ahxpd + \Pi + \Rho$,
then the underlying frame of $\model_\Gamma$ satisfies $\fc(\Pi) \wedge \fc(\Rho)$.
\end{lemma}

\begin{proof}
  Since $\model_\Gamma$ is a named model and $\Gamma$ contains all
  instances of elements of $\Pi$, it follows that the underlying frame
  of $\model_\Gamma$ satisfies $\fc(\Pi)$.  Since $\model_\Gamma$ is a
  named model and $\Gamma$ is $\Rho$-saturated, it follows that the
  underlying frame of $\model_\Gamma$ satisfies $\fc(\Rho)$.
\end{proof}

As a result we obtain the completeness result. 

\begin{theorem}[Strong Completeness]\label{th:complete}
  Let $\Pi$ be a set of pure axioms and $\Rho$ a set of existential
  saturation rules.  Let $\mathfrak{C} = \mathsf{Mod}(\mathfrak{F})$
  for $\mathfrak{F}$ the class of all frames satisfying $\fc(\Pi) \wedge \fc(\Rho)$.
  Then, the axiomatic system $\ahxpd+\Pi+\Rho$ is strongly complete for $\mathfrak{C}$.
\end{theorem}

\begin{proof}
  We need to prove that every set of $\hxpd$-formulas $\Sigma$ is
  consistent if and only if $\Sigma$ is satisfiable in an abstract
  hybrid data model satisfying the frame properties defined by $\Pi$
  and $\Rho$.
 
  For any consistent $\Sigma$, we can use the Rule Saturation Lemma to
  obtain $\Sigma^\omega$, which is a named, pasted and $\Rho$-saturated
  MCS in $\hxpd$ extended by a set $\nom'$ of additional nominals.  Let
  $\model_{\Sigma^\omega}=\tup{M,\{\eqrel_e\}_{e\in\eq},
    \{R_{\dowa}\}_{\dowa\in\mo},\lbl,\no}$
  be the extracted model from $\Sigma^\omega$. Let $i\in\Sigma^\omega$, for all $\varphi\in\Sigma$, by \emph{$@$-intro} we have $@_i\varphi\in\Sigma^\omega$ since $\Sigma^\omega$ is MCS.  Then by
  the Truth Lemma, $\model_{\Sigma^\omega},\Delta_i\models\varphi$.  By the Frame Lemma,
  $\model_{\Sigma^\omega}$ satisfies all required frame properties.
%  Because each state is named by some nominal from a countable set
%  $\nom'$, the model is countable.
\end{proof}

Because the class of abstract data models is a conservative abstraction of
concrete data models, we can conclude:

\begin{corollary}\label{coro:concrcompl}
The axiomatic system $\ahxpd+\Pi+\Rho$ is strongly  complete for 
the class of concrete hybrid data models satisfying the frame conditions
$\fc(\Pi)\wedge\fc(\Rho)$.
%(in the corresponding class).	
\end{corollary}

\subsection{Completeness for path formulas}
\label{subsec:pathcomp}

The main contribution of this article is a characterization of (local) semantic
consequence between node formulas by means of an axiomatic
system. Given that soundness of an axiomatic system is usually
granted, this is the main outcome of the strong completeness result we
just proved. In the setting of XPath it also makes sense to discuss
the issue of inference between path formulas. Path inference
can be traced back to~\cite{Pratt91}, in the context of dynamic algebras.  
We will now show how Theorem~\ref{th:complete} can also be used to 
characterize path inference. 

First, recall that previous work like~\cite{cateLM10,AbriolaDFF17}
provided \emph{equational} axiomatizations that characterize
theorems of the form $\vdash \varphi \equiv \psi$ (for $\varphi,\psi$
node expressions) and $\vdash \alpha \equiv \beta$ (for $\alpha,\beta$
path expressions). The first are obviously covered by our results that
characterize theoremhood for arbitrary node expressions (in this
case $\equiv$ is nothing more than $\leftrightarrow$).  For
equivalence between path formulas the following proposition suffices:

\begin{proposition}
Let $\alpha,\beta$ be path expressions in $\hxpd$, then $\models
\alpha \equiv \beta$ iff $\models @_i\tup{\alpha}j \leftrightarrow
@_i\tup{\beta}j$ for $i,j$ not appearing in $\alpha,\beta$.
\end{proposition}

Hence, completeness for equivalence theorems between path formulas follow from
Theorem~\ref{th:complete}.  More interesting is to consider whether a
strong completeness result for path consequence is also possible.  
We need first to introduce this notion. There exists at least two natural
possible definitions:

\begin{definition}
Let $\Lambda \cup \{\alpha\}$ be a set of path formulas in $\hxpd$, 
and let $\mathfrak{C}$ be a class of models. 
Define
\begin{itemize}
\item $\Lambda \models^1_{\mathfrak{C}} \alpha$ iff for any model
  $\model \in \mathfrak{C}$, ($\forall  m,n$, 
$\model,m,n \models \Lambda$ implies $\model,m,n \models \alpha$).
\item $\Lambda \models^2_{\mathfrak{C}} \alpha$ iff for any model
  $\model \in \mathfrak{C}$,
($\forall m,n$, $\model,m,n \models \Lambda$ implies $\forall m,n$, $\model,m,n \models \alpha$).
\end{itemize}
\end{definition}

As before we write $\models$ instead of $\models_{\mathfrak{C}}$ if
the intended class is clear from context.  From the above definition
it is easy to show that $\models^1_{\mathfrak{C}}$ implies
$\models^2_{\mathfrak{C}}$ but not vice versa.  
% Indeed, for example,
% in the class of all models $\{\dowa\dowa\} \models^2 \dowa$
% (i.e., the existence of a twice removed successor implies the
% existence of a successor) while $\{\dowa\dowa\} \not \models^1
% \dowa$. On the other hand, if $\dowa$ is interpreted as a
% transitive relation both $\{\dowa\dowa\} \models^1 \dowa$ and
% $\{\dowa\dowa\} \models^2 \dowa$ hold.
Now, both $\models^1$ and $\models^2$ can be captured using Theorem~\ref{th:complete}.

\begin{theorem}
Let $\Pi$ be a set of pure axioms and $\Rho$ be a set of existential
  saturation rules.  Let $\mathfrak{C} = \mathsf{Mod}(\mathfrak{F})$
  for $\mathfrak{F}$ the class of frames satisfying $\fc(\Pi) \wedge \fc(\Rho)$.
Let $\Lambda \cup  \{\alpha\}$ be a set of path expressions in $\hxpd$
then
\begin{enumerate}
\item $\Lambda \models^1_{\mathfrak{C}}\alpha$ implies $\{@_i\tup{\beta}j
  \mid \beta \in \Lambda\} \vdash @_i\tup{\alpha}j$,
  for $i,j$ not in $\Lambda \cup \{\alpha\}$.
\item $\Lambda \models^2_{\mathfrak{C}}\alpha$ implies $\{@_i\tup{\beta}j
  \mid \beta \in \Lambda\} \vdash @_k\tup{\alpha}l$,
  for $i,j,k,l$ not in $\Lambda \cup \{\alpha\}$.
\end{enumerate}

\end{theorem}

\begin{proof}
The result is a corollary of Theorem~\ref{th:complete} because the
following hold:
\begin{enumerate}
\item $\Lambda \models^1_{\mathfrak{C}}\alpha$ iff $\{@_i\tup{\beta}j
  \mid \beta \in \Lambda\} \models_{\mathfrak{C}}@_i\tup{\alpha}j$,
  for $i,j$ not in $\Lambda \cup \{\alpha\}$.
\item $\Lambda \models^2_{\mathfrak{C}}\alpha$ iff $\{@_i\tup{\beta}j
  \mid \beta \in \Lambda\} \models_{\mathfrak{C}}@_k\tup{\alpha}l$,
  for $i,j,k,l$ not in $\Lambda \cup \{\alpha\}$.\qedhere
\end{enumerate}

\end{proof}

%%% Local Variables: 
%%% mode: latex
%%% TeX-master: "lmcs20"
%%% End: 

\subsection{Some Concrete Examples of Pure Extensions}
\label{subsec:pure}
%!TEX root = jair18.tex
In this section we introduce some extensions of $\hxpd$, and their
corresponding extended axiomatic systems. In all cases, we can apply
the result from previous section and automatically obtain
completeness, since we only use pure axioms and/or existential
saturation rules.

\paragraph{\bf Backwards Navigation. }
We start with the language $\hxpv$, i.e., $\hxpd$ extended with the
path expression $\dowam$ (backwards navigation).  The intuitive semantics
for $\dowam$ is
$$
\model,x,y\models \dowam ~~\iff ~~  x R_{\dowa}^{-1} y,
$$
i.e., $\dowam$ is interpreted on the inverse of the accessibility
relation associated to $\dowa$. Equivalently, 
$$
\model,x,y\models \dowam ~~\iff ~~  y R_{\dowa} x.
$$

In fact, following our presentation, $\dowam$ is an
additional modal operator from $\mo$ (to which all rules and axioms of $\ahxpd$
apply), and in addition we will insist that in models of $\hxpv$, 
this accessibility relation is always interpreted as the inverse of 
the one for $\dowa$.

Formally, consider two modal symbols $\dowa$ and $\dowam$ in $\mo$ (and
assume that their respective accessibility relations in a model are
$R_{\dowa}$ and $R_{\dowam}$).  Define the set $\Pi_1$ of pure axioms
that characterizes the interaction between $\dowa$ and $\dowam$ as

\begin{center}
\begin{tabular}{l}
\toprule 
{\bf Axioms for $\dowa,\dowam$-interaction} \\
\midrule
\begin{tabular}{l@{~~~~}l}
down-up & $\vdash i\ra [\dowa]\tup{\dowam}i$ \\
up-down & $\vdash i\ra [\dowam]\tup{\dowa}i$	
\end{tabular} \\
\bottomrule
\end{tabular}
\end{center}

Since the axioms presented above are pure, the axiom system we obtain
is complete for models obtained from frames satisfying $\fc(\Pi_1)$.
It is a simple exercise to verify that $\fc(\Pi_1)$ is equivalent to
$\forall x.\forall y.(R_{\dowa}(x,y) \leftrightarrow R_{\dowam}(y,x))$.

\begin{proposition}
\label{lemma:inverse}
$\ahxpd + \Pi_1$ is sound and strongly complete  for $\hxpv$.
\end{proposition}

% \begin{proof}
% For the left to the right direction suppose $\Delta_i\ra_b\Delta_j$, iff
% by definition of $\ra_b$, $@_i\tup{\dow_b} j\in\Gamma ~(\dagger)$.	Since $@_ii$ is a
% theorem, by $@_ii\in\Gamma$ and \emph{up-down} we have $@_i[\dow_b]\tup{\dow_a}i\in\Gamma$.
% By \emph{bridge} we get $@_i[\dow_b]\tup{\dow_a}[\dow_b]\tup{\dow_a}i\in\Gamma$.
% Aiming for a contradiction, let us suppose $@_j\tup{\dow_a}i\notin\Gamma$, 
% iff $\neg@_j\tup{\dow_a}i\in\Gamma$ (by $\Gamma$ MCS), iff $@_j\neg\tup{\dow_a}i\in\Gamma$
% by \emph{self-dual}. By \emph{up-down} we have $@_j\neg\tup{\dow_a}[\dow_b]\tup{\dow_a}i\in\Gamma$,
% and by $\dagger$ and \emph{bridge}  we obtain
% $@_i\tup{\dow_b}\neg\tup{\dow_a}[\dow_b]\tup{\dow_a}i\in\Gamma$. Therefore
% $@_i\neg[\dow_b]\tup{\dow_a}[\dow_b]\tup{\dow_a}i\in\Gamma$, which is a contradiction.
% For the other direction, use \emph{down-up}.
% \end{proof}

\paragraph{\bf Sibling Navigation.}
We consider now the extension of $\hxpv$ with sibling navigation $\sib$,
denoted $\hxps$. Intuitively, 
$$
\model,x,y\models \sib ~~\iff ~~ x\neq y \mbox{ and there is some $z$ s.t. $z R_{\dowa} x$ and $zR_{\dowa}y$,}
$$
where $\dowa\in\mo$ is fixed. 

Consider now three modal symbols $\dowa$, $\dowam$ and $\sib$ in $\mo$ (and
their respective accessibility relations $R_{\dowa}$, $R_{\dowam}$ and
$R_{\sib}$).  Define the set $\Pi_2$ of pure axioms
that characterizes their interaction as the axioms in $\Pi_1$ together with

\begin{center}
\begin{tabular}{l}
\toprule 
{\bf Axioms for siblings} \\
\midrule
\begin{tabular}{l@{~~~~}l}
is-sib & $\vdash\tup{\sib} i \ra \tup{\dowam}\tup{\dowa} i$ \\
has-sib & $\vdash\neg i \wedge \tup{\dowam}\tup{\dowa} i \ra \tup{\sib}i$ \\
irref-sib & $\vdash i\ra \neg\tup{\sib} i$	
\end{tabular} \\
\bottomrule
\end{tabular}
\end{center}

As $\Pi_2$ extends $\Pi_1$, $\fc(\Pi_2)$ ensures that $R_{\dowam}$ and
$R_{\dowa}$ are inverses, and in addition that 
$\forall x.\forall y.(R_{\sib}(x,y) \leftrightarrow ((x \not = y)
\wedge (\exists z. (R_{\dowam} (x,z) \wedge R_{\dowa} (z,y)))))$.

\begin{proposition}
$\ahxpd + \Pi_2$ is sound and strongly complete  for $\hxps$.
\end{proposition}

\paragraph{\bf Data Equality Properties.}
As a final, very simple, example let us consider pure axioms defining
the behaviour of equality tests.  Consider a language with two equality
test operators $=_e$ and $=_d$ such that one defines finer
equivalence classes than the other (i.e., if $\sim_e$ and $\sim_d$ are
their respective accessibility relations we want to ensure that
${\sim_e} \subseteq {\sim_d}$). Let us call this logic $\hxpe$. 
Let $\Pi_3$ be the singleton set containing the axiom 

\begin{center}
\begin{tabular}{l}
\toprule 
{\bf Axioms for equality inclusion} \\
\midrule
\begin{tabular}{l@{~~~~}l}
incl & $\vdash\tup{@_i=_e@_j} \ra \tup{@_i=_d@_j}$ \\
\end{tabular} \\
\bottomrule
\end{tabular}
\end{center}

$\fc(\Pi_3)$ is equivalent to $\forall x.\forall y. (D_e (x,y) \to D_d (x,y))$.

\begin{proposition}
$\ahxpd + \Pi_3$ is sound and strongly complete  for $\hxpe$.
\end{proposition}

% \begin{lemma}
% \label{lemma:inclusion}
% Let $\Gamma$ be an $\ahxpd + P_1$-MCS. Let
% $\model_\Gamma=\tup{M, \{\eqrel_e\}_{e\in\eq}, \{\ra_a\}_{a\in\mo}, \lbl, \no}$
% be the extracted model from $\Gamma$, and let $\Delta_i,\Delta_j\in M$.
% $\Delta_i\eqrel_e\Delta_j$ implies $\Delta_i\eqrel_d\Delta_j$.  	
% \end{lemma}

% \begin{proof}
% If $\Delta_i\eqrel_e\Delta_j$ then $@_i\tup{\epsilon=_e@_j}\in\Gamma$. 
% By \emph{comp$=$-dist} and \emph{back}, $\tup{@_i=_e@_j}\in\Gamma$, then by
% \emph{incl} $\tup{@_i=_d@_j}\in\Gamma$. Hence, by \emph{agree} and \emph{$@$=-dist},
% $@_i\tup{\epsilon=_d@_j}\in\Gamma$, iff $\Delta_i\eqrel_d\Delta_j$.
% \end{proof}

%%% Local Variables: 
%%% mode: latex
%%% TeX-master: "lmcs20"
%%% End: 

%%% Local Variables: 
%%% mode: latex
%%% TeX-master: "jair18"
%%% End: 

\section{Hybrid XPath on Trees}
\label{sec:tree}
The Henkin-style model construction from Section~\ref{subsec:complete}
provides us with the tools to obtain complete axiomatizations for a
wide variety of extensions of $\hxpd$.  However there are interesting
cases in which completeness is not a direct corollary of the
results we already presented.  One such case is the class of tree
models which we call $\ct$.

We will devote this section to investigate hybrid \xpath over the
class of tree models. First, in Section~\ref{subsec:treeincomp} we
will show that it is not possible to get a strongly complete finitary
first-order axiomatization over trees.  In
Section~\ref{subsec:treepure} we show that it is possible to obtain a
strongly complete infinite first-order axiomatization over a slightly
larger class of models, while in Section~\ref{subsec:treeweak} we
prove that this axiomatization is weakly complete for $\ct$.

Tree models are interesting in this context since \xpath is commonly used as a query
language over \xml documents; and mathematically, the relational structure of 
an \xml document is a tree. In our setting, this consists of asking that the 
union of all relations in a model has a tree-shape.

Without loss of generality, in the rest of this section we will work
in a mono-modal setting, with a unique basic path expression $\dowa$,
and with single equality/inequality comparisons denoted $=$ and
$\neq$. The path expression $\dowa$ will be interpreted over an
accessibility relation $R_{\dowa}$, representing the union of all the
accessibility relations in the model. To achieve this, we need to work
in a signature where the set $\mo$ is finite. Therefore, we can include the 
followig additional axiom:

\begin{center}
    \begin{tabular}{l}
    %\hline 
    %{\bf } \\
    \toprule
    \begin{tabular}{l@{\ \ \ }l}
    $\dowa$-definition & $\tup{\dowa}\top \lra \tup{\dowa_1\cup\ldots\cup\dowa_k}\top$ \ \ \small{with $\mo=\set{\dowa_1,\ldots,\dowa_k}$} \\
    \end{tabular} \\
    \bottomrule
    \end{tabular}
\end{center}

It is a straightforward exercise to show that this axiom characterizes $R_{\dowa}$ as the union
of all the relations of the model, in a signature where $\mo$ is finite. 
As a consequence, in this section we will focus on axiomatizing $R_{\dowa}$ as the
acessibility relation of a tree. 
In this way, we obtain
the intended meaning for \xpath as a query language for \xml documents.

Hence, a model in this signature is a tuple
$\model=\tup{M,\eqrel,R_{\dowa},\lbl,\no}$, in which
we have a unique data equality relation $\eqrel$ and the 
accessibility relation $R_{\dowa}$.  

\subsection{Failure of Strong Completeness over Trees}
\label{subsec:treeincomp}

Let us start by formally defining the structures we will consider. 

\begin{definition}\label{def:ctree}
  $\ct$ is defined as the class of abstract hybrid data
  models $\tup{M,\eqrel, R_{\dowa},\lbl,\no}$ satisfying:
  \begin{itemize}
  \item there is a unique point in $M$ without predecessors by
    $R_{\dowa}$, which we call the \emph{root};
  \item every node $n\in M$ is reachable from the root by $R_{\dowa}$
    in zero or more steps;
  \item for all $n,l,l'\in M$, if $l R_{\dowa} n$ and
    $l' R_{\dowa} n$, then $l=l'$; and
  \item there is no $n\in M$ such that $n$ is reachable by $R_{\dowa}$
    from itself in one or more steps.
  \end{itemize}
\end{definition}

\noindent
Intuitively, a model
$\tup{M,\eqrel, R_{\dowa},\lbl,\no}$ is in $\ct$ if $\tup{M,R_{\dowa}}$
is a tree.

In order to show that there is no finitary first-order axiomatization
for $\hxpd$ over $\ct$, we will use a standard tool from first-order
model theory: the \emph{compacteness} theorem. In particular, we will
show that $\hxpd$ over $\ct$ is not compact, and consequently, a
finitary axiomatization may not exist (see, e.g.,~\cite{Enderton:2001}
for details).
%We can establish the following result.

\begin{proposition}\label{prop:treenoncompact}
 $\hxpd$ is not compact over the class $\ct$.   
\end{proposition}
\begin{proof}
  We need to show that there exists an infinite set of
  $\hxpd$-formulas such that every finite subset
  $\Gamma^{fin}\subset \Gamma$ is satisfiable in the class $\ct$, but
  $\Gamma$ is not.

    Without loss of generality, we will use natural numbers as names for nominals,
    i.e., $\nom=\nat$. We start by defining the following formula indicating that, at
    the given point, the nominal $k$ is the only one satisfied from the set $\set{0,\ldots,n}$:
    $$
     \only{n}{k}  =  @_k(\underset{0\leq j\neq k \leq n}{\bigwedge} \neg j).   
    $$
    We can define for $n \ge 0$:
    %recursively for $n \ge 1$:
    %$$
    %\begin{array}{lcl}
    % \lin(1) & = &  @_1\tup{\dow}\only{1}{0} \\
    % \lin(n+1) & =& @_{n+1} \tup{\dow}(\only{n+1}{n}\wedge\lin(n))
    %\end{array}
    %$$
    \[
    \lin(n) = \bigwedge_{0< i\le n}@_i(\tup{\dowa}i-1) \wedge \bigwedge_{0 \le i \le n}\only{n}{i}.
    \]
    The formula $\lin(n)$ states that there exists a linear chain of
    named states of length $n$.
    Now let us define  the set $\Gamma_\lin=\set{\lin(n) \mid n\in\nat}$. Notice that $\Gamma_\lin$ enforces structures with at least one infinite branch of the shape
    
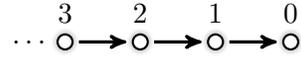
\begin{center}
        \begin{tikzpicture}[>=stealth']
        \tikzstyle myBG=[line width=3pt,opacity=1.0]
        \newcommand{\drawBGT}[2]
        {
            \draw[ar,->,white,myBG]  (#1) -- (#2);
            \draw[ar,->,black,very thick] (#1) -- (#2);
        }
        
        \node (h1) at (0,0) [world,label=left:$\ldots$,label=above:$3$] {} ;
        \node (h2) at (1,0) [world,label=above:$2$] {} ;
        \node (h3) at (2,0) [world,label=above:$1$] {} ;
        \node (h4) at (3,0) [world,label=above:$0$] {} ;
        \drawBGT{h1}{h2};
        \drawBGT{h2}{h3};
        \drawBGT{h3}{h4};
        \end{tikzpicture}
    \end{center}
    
    \noindent
    $\Gamma_\lin$ satisfies the following properties:
    \begin{enumerate}
        \item\label{it:notmodel} $\Gamma_\lin$ does not have a model in $\ct$, since it enforces an infinite chain without a root (since every node has a predecessor), and
        \item\label{it:finmodel} for all finite sets $\Gamma_\lin^{fin}\subset\Gamma_\lin$, $\Gamma_\lin^{fin}$ has a model in $\ct$.
    \end{enumerate}
 Hence, the proposition follows.   
\end{proof}

\begin{theorem}\label{th:treenotstrong}
  There is no finitary first-order axiomatization which is strongly
  complete for $\hxpd$ over $\ct$.
\end{theorem}

\begin{proof}
  Suppose there is some finitary first-order axiomatization $\ah$ which
  is sound and strongly complete for $\hxpd$ over $\ct$, and let
  $\vdash$ be obtained from $\ah$. For any set of formulas $\Gamma$,
  we know:
  \begin{enumerate}
  \item[($\dagger$)]\label{it:consist} if for all finite set
    $\Gamma^{fin}\subseteq\Gamma$, $\Gamma^{fin}$ is consistent (i.e.,
    $\Gamma^{fin}\not\vdash\bot$), then $\Gamma$ is consistent.
  \end{enumerate}
        
  This follows from the fact that any proof is finite, so in order to
  make $\Gamma$ inconsistent, only a finite set of its formulas is
  needed.

  Since $\ah$ is strongly complete, by item~\ref{it:finmodel} in
  Proposition~\ref{prop:treenoncompact}, each set $\Gamma_\lin^{fin}$
  defined therein is consistent. Then, by item ($\dagger$) also
  $\Gamma_\lin$ is consistent, and again by strong completeness of
  $\ah$ we can conclude that $\Gamma_\lin$ has a model in $\ct$,
  contradicting item~\ref{it:notmodel} in
  Proposition~\ref{prop:treenoncompact}.  Therefore, there is no
  finitary first-order axiomatic system $\ah$ strongly complete for
  $\ct$.
\end{proof}

\begin{corollary}
There are no $\Pi$ and $\Rho$ such that $\ahxpd+\Pi+\Rho$ is strongly complete for $\ct$. 
\end{corollary}

%%% Local Variables:
%%% mode: latex
%%% TeX-master: "lmcs20"
%%% End:

\subsection{Axiomatizing the Class $\ctm$ with Pure Axioms}
\label{subsec:treepure}
Define the class $\ctm$ as follows:

\begin{definition}
\label{def:ctreem}
Let $\ctm$ be the class of models
$\tup{M,\eqrel,R_{\dowa},\lbl,\no}$ such that
\begin{itemize}
\item there is no $n\in M$ such that $n$ is reachable by $R_{\dowa}$
  from itself in one or more steps;
\item for all $n \in M, n$ has at most one predecessor by $R_{\dowa}$.
\end{itemize}
\end{definition}

% The class $\ctm$ admits models of the form
% \begin{center}
%     \begin{tikzpicture}[>=stealth']
%     \tikzstyle myBG=[line width=3pt,opacity=1.0]
%     \newcommand{\drawBGT}[2]
%     {
%         \draw[ar,->,white,myBG]  (#1) -- (#2);
%         \draw[ar,->,black,very thick] (#1) -- (#2);
%     }
%     \draw (4,-1)
%     -- (3,-3) node[anchor=north]{}
%     -- (5,-3) node[anchor=south]{}
%     -- cycle;

%     \node() at (5.2,-0.9) {\tiny{$m$}};
%     \draw (5.2,-1)
%     -- (4.9,-1.6) node[anchor=north]{}
%     -- (5.5,-1.6) node[anchor=south]{}
%     -- cycle;

%     \node (h1) at (4,1) [world,label=above:$\vdots$] {} ;
%     \node (h2) at (4,0) [world] {} ;
%     \node (h3) at (4,-1) [world,label=right:$r$] {} ;
%     \drawBGT{h1}{h2};
%     \drawBGT{h2}{h3};

%     \end{tikzpicture}
% \end{center}

Notice that the class $\ctm$ admits forests made of a collection of
models that consist of a tree, possibly extended with an infinite
chain attached to the root.  Models from the class $\ctm$ differ from
those in $\ct$, given that we relaxed the first two conditions from
Definition~\ref{def:ctree}.  We will introduce a strongly complete
axiomatization for $\ctm$. Let $\Pi_{forest^-}$ be the set of axioms
below. Define $\dowa^1$ as $\dowa$, and $\dowa^{n+1}$ as
$\dowa\dowa^n$.

\begin{center}
    \begin{tabular}{l}
    \toprule 
    {\bf Axioms for forests} \\
    \midrule
    \begin{tabular}{l@{~~~~~~~~}l}
    \emph{no-loops} & $\vdash @_i\neg\tup{\dowa^n} i$ \ \ \small{for all $n>0$} \\
    \emph{no-join} & $\vdash @_j\tup{\dowa}i \wedge @_k\tup{\dowa}i \ra @_j k$
    \end{tabular} \\
    \bottomrule
    \end{tabular}
\end{center}

In the table above, \emph{no-loops} is an infinite set of axioms preventing loops of any
size, whereas \emph{no-join} prevents the existence of more than one predecesor. 
Since our axioms are pure and enforce the appropriate
structures, from Theorem~\ref{th:complete} we get:

\begin{theorem}
\label{th:treemcomplete}
The axiomatic system $\ahxpd+\Pi_{forest^-}$ is strongly complete for $\ctm$.    
\end{theorem}

%%% Local Variables:
%%% mode: latex
%%% TeX-master: "lmcs20"
%%% End:

\subsection{A Weakly Complete Axiomatization for Trees}
\label{subsec:treeweak}
In this section we will show that the axiom system
$\ahxpd+\Pi_{forest^-}$ introduced before, is weakly complete for
$\ct$. Recall that for \emph{weak completeness} we need to show that
for any $\hxpd$-formula $\varphi$, $\models\varphi$ implies
$\vdash\varphi$.  In order to achieve this result we need to work on
the extracted model.

In what follows we use $m R_{\dowa}^h n$ to denote $m R_{\dowa} u_1 R_{\dowa} \ldots R_{\dowa} u_{h-1} R_{\dowa} n$, for some sequence of states $u_1,\ldots,u_{h-1}$.

\begin{definition}\label{def:restmodel}
 Let $\model=\tup{M,\eqrel,R_{\dowa},\lbl,\no}$ be a hybrid data model, $N\subseteq\nom$, $N\neq\emptyset$, and $n\geq 0$. We define $\restmodel{\model}{n}{N}$
 as
 $$
 \restmodel{\model}{n}{N} = \tup{S,\eqrel_{\uphr S},{R_{\dowa}}_{\uphr S},\lbl_{\uphr S},\no_{\uphr S}},
 $$
\noindent where
\begin{itemize}
    \item $S=\set{m \mid \mbox{ there is some }i\in N \mbox{ such that } \no(i)\ra^h m, ~ \mbox{ with }h\leq n}$,
    \item ${\eqrel_{\uphr S}} = {\eqrel}\cap(S\times S)$;
    \item ${{R_{\dowa}}_{\uphr S}}={R_{\dowa}}\cap(S\times S)$;
    \item $\lbl_{\uphr S}(m)=~\lbl(m)$, for all $m\in S$; and
    \item $\no_{\uphr S}(i) =
    \left\{
        \begin{array}{ll}
            \no(i)  & \mbox{if } i\in N \\
            m & \mbox{if } i\notin N, ~m\in S \mbox{ arbitrary}.
        \end{array}
    \right.$
\end{itemize}
\end{definition}

Intuitively, $\restmodel{\model}{n}{N}$ is the restriction of $\model$ to the states reached from a nominal in $N$, in at most $n$ steps. It is obvious that $\restmodel{\model}{n}{N}$ has finite depth if $N$ is finite.

In what follows, we will write $\nom(\varphi)$ to denote the set of nominals appearing in the formula $\varphi$.

\begin{proposition}
\label{prop:cuttree}
Let $\varphi$ be an $\hxpd$-formula, and let $\model_\Gamma=\tup{M,\eqrel,R_{\dowa},\lbl,\no}$ be the extracted
 model from some $\ahxpd + \Pi_{forest^-}$ MCS $\Gamma$, such that $\model_\Gamma,\Delta_i\models\varphi$, for some $\Delta_i$. Let $\restmodel{(\model_\Gamma)}{\md(\varphi)}{\nom(\varphi)\cup\set{i}}= \tup{S,\eqrel_{\uphr S},{R_{\dowa}}_{\uphr S},\lbl_{\uphr S},\no_{\uphr S}}$, then,
  \begin{enumerate}
      \item $\Delta_i\in \restmodel{(\model_\Gamma)}{\md(\varphi)}{\nom(\varphi)\cup\set{i}}$, and
      \item $\restmodel{(\model_\Gamma)}{\md(\varphi)}{\nom(\varphi)\cup\set{i}},\Delta_i\models\varphi$.
  \end{enumerate}
\end{proposition}

\begin{proof}
Item $1$ follows from the definition of $\restmodel{(\model_\Gamma)}{\md(\varphi)}{\nom(\varphi)\cup\set{i}}$, since $i\in\Delta_i$ by MCS.
For item $2$, we will prove a stronger result. 
 \begin{description}
    \item[(IH1)] For all $\alpha\in\mathsf{PExp}$ such that $\md(\alpha)\leq\md(\varphi)$ and $\nom(\alpha)\subseteq\nom(\varphi)$, if $\Delta_j,\Delta_k\in\restmodel{(\model_\Gamma)}{\md(\varphi)}{\nom(\varphi)\cup\set{i}}$ then
    $$
    \model_\Gamma,\Delta_j,\Delta_k\models \alpha ~~ \iff ~~
    \restmodel{(\model_\Gamma)}{\md(\varphi)}{\nom(\varphi)\cup\set{i}},\Delta_j,\Delta_k\models\alpha.
    $$
    \item[(IH2)] For all $\psi\in\mathsf{NExp}$ such that $\md(\psi)\leq\md(\varphi)$
 and $\nom(\psi)\subseteq\nom(\varphi)$, if $\Delta_j\in\restmodel{(\model_\Gamma)}{\md(\varphi)}{\nom(\varphi)\cup\set{i}}$ then
 $$
 \model_\Gamma,\Delta_j\models \psi ~~ \iff ~~
 \restmodel{(\model_\Gamma)}{\md(\varphi)}{\nom(\varphi)\cup\set{i}},\Delta_j\models\psi.
 $$
 \end{description}
 The proof is by induction on $\alpha$ and $\psi$.
Assume $\Delta_j,\Delta_k\in \restmodel{(\model_\Gamma)}{\md(\varphi)}{\nom(\varphi)\cup\set{i}}$. Let us start by the base cases.

\medskip

\noindent {- $\alpha=\dowa$:} Suppose $\model_\Gamma,\Delta_j,\Delta_k\models\dowa$, iff $\Delta_j R_{\dowa} \Delta_k$. Since $\Delta_j,\Delta_k\in \restmodel{(\model_\Gamma)}{\md(\varphi)}{\nom(\varphi)\cup\set{i}}$,  by definition we have $\Delta_j\ra_{\uphr S}\Delta_k$. Therefore, $\restmodel{(\model_\Gamma)}{\md(\varphi)}{\nom(\varphi)\cup\set{i}},\Delta_j,\Delta_k\models\dowa$. The other direction is straightforward, since 
$\Delta_j {R_{\dowa}}_{\uphr S}\Delta_k$ implies  $\Delta_j R_{\dowa} \Delta_k$. \smallskip

\noindent {- $\alpha=@_l$:} Suppose $\model_\Gamma,\Delta_j,\Delta_k\models@_l$, iff  $\no(l)=\Delta_k$. By assumption, $\nom(@_l)\subseteq\nom(\varphi)$, then $\no_{\uphr S}(l)=\no(l)=\Delta_k$, therefore $\restmodel{(\model_\Gamma)}{\md(\varphi)}{\nom(\varphi)\cup\set{i}},\Delta_j,\Delta_k\models@_l$. The other direction follows from the fact that $\nom(@_l)=\set{l}\subseteq\nom(\varphi)$,
then $\no(l)=\Delta_k$.\smallskip

\noindent {- $\psi=p$:} Suppose $\model_\Gamma,\Delta_j\models p$, iff $p\in\lbl(\Delta_j)$. By assumption $\Delta_j\in S$, then $p\in\lbl_{\uphr S}(\Delta_j)$. Hence, $\restmodel{(\model_\Gamma)}{\md(\varphi)}{\nom(\varphi)\cup\set{i}},\Delta_j\models p$. The other direction is similar.\smallskip

\noindent {- $\psi=k\in\nom(\varphi)$:}  Suppose $\model_\Gamma,\Delta_j\models k$, iff $\no(k)=\Delta_j$. By assumption $\Delta_j\in S$, then $\no_{\uphr S}=\Delta_j$. Hence, $\restmodel{(\model_\Gamma)}{\md(\varphi)}{\nom(\varphi)\cup\set{i}},\Delta_j\models k$. The other direction is similar.\medskip

Now we prove the inductive cases. \smallskip

\noindent {- $\alpha=[\theta]$:}  Suppose $\model_\Gamma,\Delta_j,\Delta_k\models[\theta]$, iff $\Delta_j=\Delta_k$ and $\model_\Gamma,\Delta_k\models\theta$. By $\Delta_j\in S$ and (IH2), we have 
$\restmodel{(\model_\Gamma)}{\md(\varphi)}{\nom(\varphi)\cup\set{i}},\Delta_j\models \theta$, iff $\restmodel{(\model_\Gamma)}{\md(\varphi)}{\nom(\varphi)\cup\set{i}},\Delta_j,\Delta_k\models [\theta]$.
For the other direction, use similar steps.\smallskip

\noindent {- $\alpha=\beta\gamma$:} Suppose $\model_\Gamma,\Delta_j,\Delta_k\models\beta\gamma$, iff there exists $\Delta_l$ such that 
$\model_\Gamma,\Delta_j,\Delta_l\models\beta$ and $\model_\Gamma,\Delta_l,\Delta_k\models\gamma$. Notice that since $\Delta_j,\Delta_k\in S$, and $\md(\alpha)\leq\md(\varphi)$, we have $\Delta_j\ra^n \Delta_l$, for some $n\leq\md(\varphi)$, then $\Delta_l\in S$. By (IH1), we have $\restmodel{(\model_\Gamma)}{\md(\varphi)}{\nom(\varphi)\cup\set{i}},\Delta_j,\Delta_l\models\beta$ and $\restmodel{(\model_\Gamma)}{\md(\varphi)}{\nom(\varphi)\cup\set{i}},\Delta_l,\Delta_k\models\gamma$. Therefore, $\restmodel{(\model_\Gamma)}{\md(\varphi)}{\nom(\varphi)\cup\set{i}},\Delta_j,\Delta_k\models\beta\gamma$. The other direction is direct from (IH1).
\smallskip

\noindent {- $\psi=\neg\theta$ and $\psi=\theta\wedge\theta'$:} are direct from (IH2).\smallskip

\noindent {- $\psi=\tup{\beta=\gamma}$:}  $\model_\Gamma,\Delta_j\models\tup{\beta=\gamma}$, iff there exists $\Delta_k,\Delta_l$ such that $\model_\Gamma,\Delta_j,\Delta_k\models\beta$, $\model_\Gamma,\Delta_j,\Delta_l\models\gamma$
and $\Delta_k\eqrel\Delta_l$. From the fact that $\md(\psi)\leq\md(\varphi)$,
we have $\md(\beta)\leq\md(\varphi)$ and $\md(\gamma)\leq\md(\varphi)$. Then,
$\Delta_j\ra^n\Delta_k$, and $\Delta_j\ra^m\Delta_l$, for some $n,m\leq\md(\varphi)$. As a consequence, $\Delta_k,\Delta_l\in S$, then by (IH1) we get  $\restmodel{(\model_\Gamma)}{\md(\varphi)}{\nom(\varphi)\cup\set{i}},\Delta_j,\Delta_k\models\beta$ and  $\restmodel{(\model_\Gamma)}{\md(\varphi)}{\nom(\varphi)\cup\set{i}},\Delta_j,\Delta_l\models\gamma$. On the other hand, since $\Delta_k\eqrel\Delta_l$, we also get $\Delta_k\eqrel_{\uphr S}\Delta_l$, therefore  $\restmodel{(\model_\Gamma)}{\md(\varphi)}{\nom(\varphi)\cup\set{i}},\Delta_j\models\tup{\beta=\gamma}$. The other direction follows from (IH1).
\smallskip

\noindent {- $\psi=\tup{\beta\neq\gamma}$:} Analogous to the previous case.
\end{proof}

With this construction at hand, we obtain almost the structure of the model we are looking for.

\begin{lemma}
Let $\model_\Gamma$ be the extracted model for some MCS $\Gamma$, $n\in\nat$ and $N\subseteq\nom$ finite. $\restmodel{(\model_\Gamma)}{n}{N}$ is a forest composed by a finite number of disjoint trees.
\end{lemma}
\begin{proof}
 It is a straightforward exercise to check that from each nominal $i\in N$, we generate a tree (remember that by \emph{no-loops} there are no cycles in $\model_\Gamma$, and by \emph{no-join} each node has at most one precedessor).
Since $N$ is finite, we generate only a finite number of trees.
\end{proof}

\begin{definition}\label{def:itree}
%Let $\restmodel{(\model_\Gamma)}{\md(\varphi)}{\nom(\varphi)\cup\set{i}}$ be as in Proposition~\ref{prop:cuttree}. 
Let $\restmodel{\model}{n}{N}$ be a model as in Definition~\ref{def:restmodel}.
We define $\itree{\model}{n}{N}$ as the tree resulting of adding a new node $r$ to $\restmodel{\model}{n}{N}$, and edges from $r$ to the root to each of its subtrees. Essentially, $\itree{\model}{n}{N}$ is of the form:

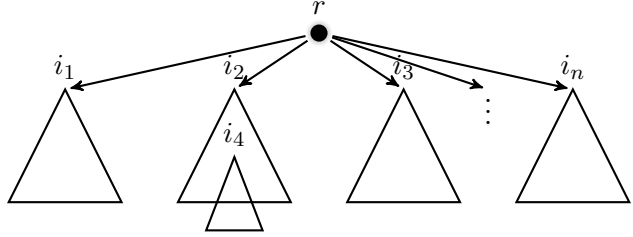
\begin{center}
    \begin{tikzpicture}[scale=.75]
    \node (g1) at (4.5,1) [blackworld,label=above:$r$] {} ;
    
    \draw (0,0) node[anchor=north,label=above:$i_1$]{}
      -- (-1,-2) node[anchor=north]{}
      -- (1,-2) node[anchor=south]{}
      -- cycle;
    
      \draw (3,0) node[anchor=north,label=above:$i_2$]{}
      -- (2,-2) node[anchor=north]{}
      -- (4,-2) node[anchor=south]{}
      -- cycle;

      \draw (3,-1.2) node[anchor=north,label=above:$i_4$]{}
      -- (2.5,-2.5) node[anchor=north]{}
      -- (3.5,-2.5) node[anchor=south]{}
      -- cycle;
    
      \draw (6,0) node[anchor=north,label=above:$i_3$]{}
      -- (5,-2) node[anchor=north]{}
      -- (7,-2) node[anchor=south]{}
      -- cycle;
    
      \draw (7.5,-.25) node[] {\vdots};
    
      \draw (9,0) node[anchor=north,label=above:$i_n$]{}
      -- (8,-2) node[anchor=north]{}
      -- (10,-2) node[anchor=south]{}
      -- cycle;
      \draw[ar,->,black,thick] (g1) -- (0,0);
      \draw[ar,->,black,thick] (g1) -- (3,0);
      \draw[ar,->,black,thick] (g1) -- (6,0);
      \draw[ar,->,black,thick] (g1) -- (7.5,0);
      \draw[ar,->,black,thick] (g1) -- (9,0);
    \end{tikzpicture}
    \end{center}
\end{definition}

\begin{proposition}\label{prop:itreepres}
    Let $\model_\Gamma$ be the extracted model for some $\ahxpd+\Pi_{forest^-}$ MCS $\Gamma$ and let $\varphi$ be an $\hxpd$-formula such that $\model_\Gamma,  \Delta_i\models\varphi$, for some $\Delta_i$. Then,
    \begin{enumerate}
        \item $\itree{(\model_\Gamma)}{\md(\varphi)}{\nom(\varphi)\cup\set{i}}$ is a tree of finite height, and
        \item  $\itree{(\model_\Gamma)}{\md(\varphi)}{\nom(\varphi)\cup\set{i}},\Delta_i\models\varphi$.
    \end{enumerate}
\end{proposition}

\begin{proof}
 For $1$, by construction it is obvious that $\itree{(\model_\Gamma)}{\md(\varphi)}{\nom(\varphi)\cup\set{i}}$ is a tree of finite height. To prove $2$ we can use a bisimulation-based argument. Define $Z$ as the identity relation between the states in $\restmodel{(\model_\Gamma)}{\md(\varphi)}{\nom(\varphi)\cup\set{i}}$ and $\itree{(\model_\Gamma)}{\md(\varphi)}{\nom(\varphi)\cup\set{i}}$ (therefore $r$ is not in the relation). It is a straightforward exercise to show that $Z$ is an $\hxpd$-bisimulation. Since $\restmodel{(\model_\Gamma)}{\md(\varphi)}{\nom(\varphi)\cup\set{i}},\Delta_i\models\varphi$, by Proposition~\ref{prop:invariance} we get
 $\itree{(\model_\Gamma)}{\md(\varphi)}{\nom(\varphi)\cup\set{i}},\Delta_i\models\varphi$. 
\end{proof}

Because we were able to transform the extracted model into a model in $\ct$ which still satisfies the intended formula, we obtain the intended completeness result.

\begin{theorem}
 The system $\ahxpd+\Pi_{forest^-}$ is weakly complete for $\ct$.  
\end{theorem}

Moreover, $\itree{(\model_\Gamma)}{\md(\varphi)}{\nom(\varphi)\cup\set{i}}$ can be transformed into a bounded finite model, and hence, also decidability of the validity problem follows.

\begin{definition}
  Let $\varphi$ be an $\hxpd$-formula, and let
  $\model=\tup{M,\eqrel,R_{\dowa},\lbl,\no}\in\ct$.  We write
  $m\formequiv_{\Gamma}m'$ if for all $n$, $n R_{\dowa} m$ if and only if
  $n R_{\dowa} m'$, and $\model,m\models \varphi$ if and only if $\model,m'\models
  \varphi$, for all $\varphi \in \Gamma$. If $\Gamma$ is finite, $\formequiv_\Gamma$ is an
  equivalence relation of finite index.

  Given an equivalence relation $\formequiv$ over a set $M$, a
  \emph{selection function $s$ for $\formequiv$} is a
  function $s:M_{/{\formequiv}}\to M$ such that
  $s([m])=m'$, for $m'\in[m]$ arbitrary.
\end{definition}

\begin{definition}\label{def:ftree}
Let $\itree{(\model_\Gamma)}{\md(\varphi)}{\nom(\varphi)\cup\set{i}}=\tup{M,\eqrel,R_{\dowa},\lbl,\no}$ be as in Definition~\ref{def:itree}, let $\varphi$ be an $\hxpd$-formula, and let $s$ be a selection function for $\formequiv_{\sub(\varphi)}$. We define $\ftree{(\model_\Gamma)}{\md(\varphi)}{\nom(\varphi)\cup\set{i}}=\tup{F,\eqrel_{\uphr F},{R_{\dowa}}_{\uphr F},\lbl_{\uphr F},\no_{\uphr F}}$ where
\begin{itemize}
    \item $F=\set{s([m]) \mid m\in M}$;
    \item ${\eqrel_{\uphr F}}  = {\eqrel} \cap (F\times F)$;
    \item ${{R_{\dowa}}_{\uphr F}} = {R_{\dowa}} \cap (F \times F)$;
    \item $\lbl_{\uphr F}(m) = ~ \lbl(m)$, for all $m\in F$; and
    \item $\no_{\uphr F}(i) = 
    \left\{
        \begin{array}{ll}
            \no(i)  & \mbox{if } i\in \nom(\varphi) \\
            m & \mbox{if } i\notin \nom(\varphi), m\in F \text{ arbitrary}.
        \end{array}
    \right.$
\end{itemize}
\end{definition}

\begin{proposition}\label{prop:finmodel}
    Let $\model_\Gamma$ be the extracted model for some $\ahxpd+\Pi_{forest^-}$ MCS $\Gamma$ and let $\varphi$ be an $\hxpd$-formula such that $\model_\Gamma,  \Delta_i\models\varphi$, for some $\Delta_i$. Then,
    \begin{enumerate}
        \item $\ftree{(\model_\Gamma)}{\md(\varphi)}{\nom(\varphi)\cup\set{i}}$ is a finite tree, and 
        \item $\ftree{(\model_\Gamma)}{\md(\varphi)}{\nom(\varphi)\cup\set{i}},s([\Delta_i])\models\varphi$.
    \end{enumerate}
\end{proposition}

\begin{proof}
  $1$ follows from the fact that
  $\itree{(\model_\Gamma)}{\md(\varphi)}{\nom(\varphi)\cup\set{i}}$ is
  a tree of finite height, and the selection function selects only one
  representant for each equivalence class, which gives us finite
  width. In order to prove $2$, let $d(m)$ be the distance from $m$ to
  the root, define:
\[
\begin{array}{rl}
m Z_{\md(\varphi)} s([m]) & \text{if } d(m) \le d(\Delta_i) \\
m Z_{\max (\md(\varphi) - (d(\Delta_i) - d(m)), 0)} s([m]) & \text{if }
d(m) > d(\Delta_i).
\end{array}
\]

We claim
  $$
\itree{(\model_\Gamma)}{\md(\varphi)}{\nom(\varphi)\cup\set{i}},\Delta_i\bisim_{\md(\varphi)} \ftree{(\model_\Gamma)}{\md(\varphi)}{\nom(\varphi)\cup\set{i}},s([\Delta_i]),
  $$
  therefore (by Propositions~\ref{prop:linvariance} and~\ref{prop:itreepres}), $\ftree{(\model_\Gamma)}{\md(\varphi)}{\nom(\varphi)\cup\set{i}},s([\Delta_i])\models\varphi$.
\end{proof}

As a consequence we get:
\begin{theorem}
\label{th:tree-bmp}
$\hxpd$ on trees has the bounded model property.
\end{theorem}

The model constructed in Definition~\ref{prop:finmodel} has size at most
exponential in the input formula. Moreover, the algorithm below shows that model checking for $\hxpd$ can be solved in polynomial time. Thus, for a given formula we can guess a model of exponential size and do model checking in polynomial time, solving the problem in \nexpspace (which coincides with the class \expspace).

Let $\model=\tup{M,$
	$\{\eqrel_e\}_{e\in\eq},\{R_{\dowa}\}_{\dowa\in\mo},\lbl,\no}$ be a model and $\varphi$ a formula in $\hxpd$, let $e_1, \ldots, e_n$ be an enumeration, ordered by size, of the subexpressions of $\varphi$. Define the functions $\mlabel_{ne}$ and $\mlabel_{pe}$ inductively as follows.   Initially, 
set $\mlabel_{ne}(m)$ = $\mlabel_{pe}(m_1,m_2)= \{\}$ for any $m$, $m_1$, $m_2$.
Then extend these labeling functions using the following rules.

\begin{enumerate}[align=left]
\item $e_i =  \dowa \in \mo$: $\madd(\mlabel_{pe}(m_1,m_2), a)$ iff $R_{\dowa}(m_1,m_2)$. 

\item $e_i = @_i$: $\madd(\mlabel_{pe}(m_1,m_2), @_i)$ iff $m_2 = \no(i)$.

\item $e_i = [\psi]$: $\madd(\mlabel_{pe}(m,m), [\psi])$ iff $m \in \mlabel_{ne}(\psi)$.

\item $e_i = \dowa\alpha$: $\madd(\mlabel_{pe}(m_1,m_2),\dowa\alpha)$ iff there is $m_3$ s.t.\ $R_{\dowa}(m_1,m_3)$ and $\alpha \in \mlabel_{pe}(m_3,m_2)$.

\item $e_i = @_i\alpha$: $\madd(\mlabel_{pe}(m_1, m_2)@_i\alpha)$ iff $\alpha \in \mlabel_{pe}(\no(i),m_2)$. 

\item $e_i = [\psi]\alpha$: $\madd(\mlabel_{pe}(m_1,m_2),[\psi]\alpha)$ iff $\psi \in \mlabel_{ne}(m_1)$ and $\alpha \in \mlabel_{pe}(m_1,m_2)$.

\item $e_i = p \in \prop$: $\madd(\mlabel_{ne}(m),p)$ iff $p \in V(m)$.

\item $e_i = i \in \nom$: $\madd(\mlabel_{ne}(m), i)$ iff $m = \no(i)$.

\item $e_i = \neg \psi$: $\madd(\mlabel_{ne}(m),\neg \psi)$ iff $\psi \not\in \mlabel_{ne}(m)$.

\item $e_i = \psi_1 \wedge \psi_2$: $\madd(\mlabel_{ne}(m), \psi_1 \wedge \psi_2)$ iff $\psi_1, \psi_2 \in \mlabel_{ne}(m)$.

\item $e_i = \tup{\alpha_1 =_e \alpha_2}$: $\madd(\mlabel_{ne}(m), \tup{\alpha_1 =_e \alpha_2})$ iff there are $m_1, m_2$ s.t.\ $\alpha_1 \in \mlabel_{pe}(m,m_1)$, $\alpha_2 \in \mlabel_{pe}(m,m_2)$ and $m_1 \sim_e m_2$.

\item $e_i = \tup{\alpha_1 \not=_e \alpha_2}$: $\madd(\mlabel_{ne}(m),\tup{\alpha_1 \not=_e \alpha_2})$ iff there are $m_1, m_2$ s.t.\ $\alpha_1 \in \mlabel_{pe}(m,m_1)$, $\alpha_2 \in \mlabel_{pe}(m,m_2)$ and it is not the case that $m_1 \sim_e m_2$.

\end{enumerate}

As a result we obtain the following upper bound for $\hxpd$-satisfiability in $\ct$. 

\begin{corollary}
\label{cor:tree-decid}
The satisfiability problem for $\hxpd$ on trees is decidable in \expspace.
\end{corollary}

%%% Local Variables: 
%%% mode: latex
%%% TeX-master: "jair18"
%%% End: 

% \subsection{A Weakly Complete Axiomatization for $\hxprt$ on Trees}
% \label{subsec:nonpure}
% \input{tree-weak-star}

%%% Local Variables:
%%% mode: latex
%%% TeX-master: "lmcs20"
%%% End:

\section{Decidability via Filtrations}
\label{sec:filt}
% !TEX root = jair18.tex

In the previous section we already obtained decidability for $\hxpd$ on trees. 
In this section we introduce \emph{filtrations}, a notion which will 
% the needed definitions, and
help us to establish a bounded finite model property and an upper
bound for the complexity of satisfiability for
$\hxpv$\footnote{In~\cite{ArecesFS17} it has been proved that the
  satisfiability problem for $\hxpd$ with a single accessibility
  relation and a single data relation is \pspace-complete. However,
  for multiple relations the exact complexity is unknown.} over the
class of all abstract hybrid data models. Recall that $\hxpv$ is the 
extension of $\hxpd$ with inverse modalities.

In Section~\ref{sec:prelim} we introduced path expressions of the form
$\alpha\cup\beta$ as an abbreviation, but this definition potentially
leads to an exponential blow up in the size of the formula. In this
section we will consider $\cup$ as part of the language of $\hxpv$ to
avoid this blow up.

\begin{definition}\label{def:closed}
  Let $\Sigma$ be a set of formulas of $\hxpd$. We say that $\Sigma$
  is {\em closed} if:
  \begin{itemize}
  \item $\neg\varphi\in\Sigma$ implies $\varphi\in\Sigma$
  \item $\varphi\wedge\psi\in\Sigma$ implies $\varphi,\psi\in\Sigma$
  \item $\tup{\alpha*\beta}\in\Sigma$ implies
    $\tup{\alpha}\top,\tup{\beta}\top\in\Sigma$
  \item $\tup{[\varphi]}\top\in\Sigma$ implies $\varphi\in\Sigma$
  \item $\tup{@_i}\top\in\Sigma$ implies $i\in\Sigma$
  \item $\tup{\alpha\beta}\top\in\Sigma$ implies
    $\tup{\alpha}\top,\tup{\beta}\top\in\Sigma$
  \item $\tup{\alpha\cup\beta}\top\in\Sigma$ implies
    $\tup{\alpha}\top,\tup{\beta}\top\in\Sigma$.
  \end{itemize}
\end{definition}

\begin{definition}\label{def:equivfilt}
  Let $\Sigma$ be a closed set of formulas, and $\model$ be an
  abstract hybrid data model with $m,n$ states in $\model$. We define
  $m\fsigma n$ if and only if for all $\varphi\in\Sigma$
  ($\model,m\models\varphi$ iff $\model,n\models\varphi$).
\end{definition}

Now let us introduce filtration, the construction that will help us to
obtain a \emph{small model property} and give us decidability.

\begin{definition}[Filtrations]
\label{def:filtration}
Let $\model=\tup{M,\{\eqrel_e\}_{e\in\eq},\{R_{\dowa}\}_{\dowa\in\mo},\lbl,\no}$, be an abstract hybrid data model, and $\Sigma$ a closed set of formulas. A {\em filtration of $\model$ via $\Sigma$} is any 
$\model^f=\tup{M^f,\{\eqrel^f_e\}_{e\in\eq},\{R^f_{\dowa}\}_{\dowa\in\mo},\lbl^f,\no^f}$, such that
\begin{itemize}
		\item $M^f=M_{/\fsigma}$
		\item If $m R_{\dowa} n$ then $\cls{m} R^f_{\dowa} \cls{n}$
		\item If $\cls{m} R_{\dowa}^f \cls{n}$ then for all $\tup{\dowa}\top\in\Sigma$, 
		for all $m'\in\cls{m}$ there exists $n'\in\cls{n}$ such that $\model,m',n'\models\dowa$
		\item $\cls{m}\eqrel_e^f\cls{n}$ iff $m\eqrel_e n$
		\item $\no^f(i)=\cls{\no(i)}$
		\item $\lbl^f(\cls{m}) = \{p \mid p\in\lbl(m)\}$.
\end{itemize}
The \emph{smallest filtration} is a filtration that satisfies
$$
 		\cls{m} R_{\dowa}^f \cls{n} \text{ iff there exist } m'\in\cls{m},n'\in\cls{n}
 		\text{ such that } m' R_{\dowa} n'.
$$
\end{definition}

%Notice that we first define maximal and minimal conditions to be 
%satisfied by the accessibility relation in a filtration, but also
%we define a particular filtration: \emph{the smallest filtration}.
%Such definition will be helpful later, but for now we will work 
%with the general notion.

\begin{theorem}\label{th:filtration}
Let $\model$ be an abstract hybrid data model and $\model^f$ a filtration of $\model$ for a closed set $\Sigma$. Then
\begin{enumerate}
		\item For all $\varphi\in\Sigma$, for all $m\in\model$, $\model^f,\cls{m} \models\varphi$ 
		iff $\model,m\models\varphi$, and
		\item For all $\tup{\alpha}\top\in\Sigma$, for all $m,n\in\model$,
		$\model^f,\cls{m},\cls{n} \models\alpha$ iff for all $m'\in\cls{m}$ 
		there exists $n'\in\cls{n}$ such that $\model,m',n'\models\alpha$.
\end{enumerate}	
\end{theorem}

\begin{proof} 
  By structural induction. First we prove $1$.  The base and Boolean
  cases are simple.  Let us consider
  $\tup{\alpha=_e\beta}\in\Sigma$.
  $\model^f,\cls m\models\tup{\alpha=_e\beta}$ iff there are
  $\cls n,\cls u$ such that
		\begin{enumerate}[(i)]
			\item $\model^f,\cls m,\cls n\models\alpha$ iff (by IH and $\tup\alpha\top\in\Sigma$)
			there exist $m'\in\cls m,n'\in \cls n$ such that $\model,m', n'\models\alpha$,
			\item $\model^f,\cls m,\cls u\models\beta$ iff (by IH and $\tup\beta\top\in\Sigma$)
			there exist $m''\in\cls m,u'\in \cls u$ such that
			 $\model,m'',u' \models\beta$, and
			\item $\cls n\eqrel_e^f \cls u$ iff (definition of $\model^f$) $n\eqrel_e u$. 
			Since $\cls n=\cls n'$ and $\cls u=\cls u'$ we have $n'\eqrel_e u'$.
		\end{enumerate}
	Notice that the problem is that $m'$ might be different from $m''$. Suppose for contradiction that 
	$\model,m',u'\not\models\beta$. Since $\tup{\beta}\top\in\Sigma$, $m'\in\cls m$, and $u'\in\cls u$, by using $(ii)$ above we get 
	$\model^f,\cls m,\cls u\not\models\beta$, which is a contradiction.
	Therefore, $\model,m'\models\tup{\alpha=_e\beta}$.
		The case for $\tup{\alpha\neq_e\beta}$ is similar.
\medskip

\noindent
Now let us prove $2$:

\noindent
- Case $\tup{[\varphi]}\top\in\Sigma$. For the left to the right
direction, suppose $\model^f,\cls m, \cls n \models [\varphi]$, iff
$\cls m=\cls n$ and $\model^f,\cls m \models \varphi$. By $1$ we have
$\model,m\models\varphi$, then by definition of $\cls m$,
$\model,m'\models\varphi$, for all $m'\in\cls m$. Hence
$\model,m',m'\models[\varphi]$ as wanted, since $m'\in\cls n$.

For the other direction, suppose that for all $m'\in\cls m$ there
exists $n'\in\cls n$ s.t. $\model,m',n'\models[\varphi]$.  Then
$m'=n'$ and $\model,m'\models\varphi$.  By $1$,
$\model^f,\cls{m'}\models\varphi$, and since $\cls{m}=\cls{m'}$,
$\model^f,\cls m\models\varphi$. Therefore
$\model^f,\cls m,\cls n\models[\varphi]$, because $n'=m'$ implies
$\cls{n}=\cls{n'}=\cls{m'}=\cls{m}$.
\smallskip

\noindent
- Case $\tup{@_i}\top\in\Sigma$.  For the left to the right direction
suppose $\model^f,\cls m, \cls n\models @_i$, iff $\no^f(i)=\cls n$.
By definition of $\no^f$, we have $\no^f(i)=\cls{\no(i)}$, then
$\no(i)=n$. Hence $\model,m,n\models @_i$.

For the other direction, suppose that for all $m'\in\cls m$ there
exists $n'\in\cls n$ s.t. $\model,m',n'\models @_i$, iff
$\no(i)=n'$. Then $\cls{n'}=\set{n'}=\cls{n}$, therefore
$\model^f,\cls m,\cls n\models @_i$.
\smallskip

\noindent
- Case $\tup{\dowa}\top\in\Sigma$. For the left to the right
direction suppose $\model^f,\cls m, \cls n \models\dowa$, iff
$\cls m R_{\dowa}^f \cls n$. Then, since $\tup{\dowa}\top\in\Sigma$, for
all $m'\in \cls m$ there exists $n'\in\cls n$ s.t.
$\model,m',n'\models\dowa$ (by the second condition on the definition
of $R_{\dowa}^f$), as wanted.

For the right to the left direction, suppose for all $m'\in \cls m$
there exists $n'\in\cls n$ s.t. $\model,m',n'\models\dowa$. Then
$m' R_{\dowa} n'$, and by the first condition of the definition of
$R_{\dowa}^f$ we have $\cls{m'} R_{\dowa}^f \cls{n}$, iff
$\cls{m} R_{\dowa}^f \cls{n}$.  Hence,
$\model^f,\cls{m},\cls{n}\models\dowa$.
\smallskip

\noindent
- Case $\tup{\alpha\beta}\top\in\Sigma$. For the left to the right
direction suppose $\model^f,\cls m,\cls n\models\alpha\beta$, iff
there exists $\cls z$ such that $\model^f,\cls m,\cls z\models\beta$
and $\model^f,\cls z,\cls n\models\alpha\beta$. By IH, for all
$m'\in \cls m$ there exists $z' \in \cls z$ s.t.
$\model,m',z'\models\alpha$ and for all $z'\in\cls z$ there exists
$n'\in\cls n$ s.t. $\model,z',n'\models\beta$.  This implies that for
all $m'\in\cls m$ there exists $n'\in\cls{n}$
s.t. $\model,m',n'\models\alpha\beta$.

For the other direction, suppose for all $m'\in\cls m$ there exists
$n'\in\cls{n}$ s.t. $\model,m',n'\models\alpha\beta$. Then there
exists $z'$ s.t. $\model,m',z'\models\alpha$ and
$\model,z',n'\models\beta$.  By IH we have
$\model^f,\cls m,\cls{z'}\models\alpha$ and
$\model^f,\cls{z'},\cls{n}\models\beta$.  Hence
$\model^f,\cls{m},\cls{n}\models\alpha\beta$.
\smallskip

\noindent
- Case $\tup{\alpha\cup\beta}\top\in\Sigma$. For the left to the right
direction suppose $\model^f,\cls m,\cls n\models\alpha\cup\beta$, iff
$\model^f,\cls m,\cls n\models\alpha$ or
$\model^f,\cls m,\cls n\models\beta$. By IH, for all $m'\in \cls m$
there exists $n' \in \cls n$ s.t.  $\model,m',n'\models\alpha$ or
$\model,m',n'\models\beta$.  This implies that for all $m'\in\cls m$
there exists $n'\in\cls{n}$ s.t. $\model,m',n'\models\alpha\cup\beta$.

For the other direction, suppose for all $m'\in\cls m$ there exists
$n'\in\cls{n}$ s.t. $\model,m',n'\models\alpha\cup\beta$, iff
$\model,m',n'\models\alpha$ or $\model,m',n'\models\beta$.  By IH we
have $\model^f,\cls m,\cls{n}\models\alpha$ or
$\model^f,\cls{m},\cls{n}\models\beta$.  Hence
$\model^f,\cls{m},\cls{n}\models\alpha\cup\beta$.
\end{proof}

As it happens for many modal logics the previous result gives us a
bounded model property for satisfiability: if a formula $\varphi$ is
satisfiable it is satisfiable in a model of size at most exponential
in the size of $\varphi$. Moreover, filtrations preserve some
structural properties of the accessibility relations. In particular
the following is true:

\begin{proposition}
\label{prop:preservation}
Let $\model=\tup{M,\{\eqrel_e\}_{e\in\eq},\{\ra_a\}_{a\in\mo},\lbl,\no}$ be a 
model, and let $\model^f$ be the smallest filtration of $\model$.
Then, $\dowam=(\dowa)^{-1}$ implies $(\dowam)^f=(\dowa^f)^{-1}$
% \begin{enumerate}
% 		\item if $\dow_b=\dow_a^{-1}$ then $\dow^f_b=(\dow^f_a)^{-1}$, and
% 		\item if $\dow_b$ is the sibling relation of $\dow_a$ then
% 		$\dow_b^f$ is the sibling relation of $\dow_a^f$. 
% 	\end{enumerate}	
\end{proposition}

\begin{proof}
Suppose that $\dowam=(\dowa)^{-1}$.
We want to prove that $(\dowam)^f=(\dowa^f)^{-1}$. Then suppose
$\cls{m}R_{\dowam}^f\cls{n}$. Since $f$ is the smallest filtration, 
there exist $m'\in\cls{m},n'\in\cls{n}$ such that $m' R_{\dowam} n'$.
But by hypothesis, we have $n' R_{\dowa} m'$, and by definition
of filtration we get $\cls{n'} R_{\dowa}^f\cls{m'}$. Since $\cls{m'}=\cls{m}$
and $\cls{n'}=\cls{n}$, we obtain $\cls{n} R_{\dowa}^f\cls{m}$.
% Now let us prove $2$. Suppose  that $\dow_b$ is the sibling relation 
% of $\dow_a$. Suppose now that $\cls{m}\ra_b^f\cls{n}$. Since $f$ is 
% the smallest filtration, there exist $m'\in\cls{m},n'\in\cls{n}$ such 
% that $m'\ra_b n'$. Because $\dow_b$ is the sibling relation of $\dow_a$,
% $m'\neq n'$ and there exists some $z'$ such that $z'\ra_a m'$ and
% $z'\ra_a n'$. 
\end{proof}

As a result we obtain an upper bound for the complexity of 
satisfiability.

\begin{corollary}
The satisfiability problem of $\hxpv$ is in \nexptime.	
\end{corollary}

On the other hand, if we consider sibling navigation the argument does
not work. In the figure below we represent a relation $R_{\dowa}$ by black
arrows, and the sibling relation $R_{\sib}$ by dotted arrows.  Let
$\Sigma=\set{\tup{\dowa}\top}$. The model on the right is the
filtration by $\Sigma$ of the model on the left. As we can see, any
filtration collapses $v$ and $u$ into a unique equivalence class with
a reflexive arrow, but the sibling relation must be
irreflexive. Hence the notion of filtration we introduced
is not appropriate to handle siblings.

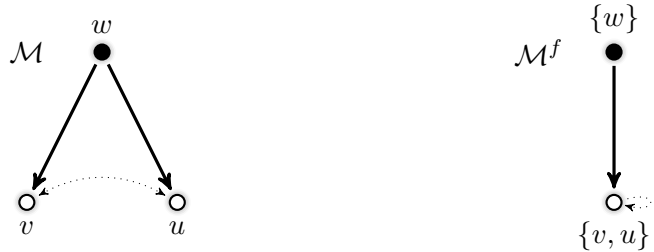
\begin{center}
\begin{tabular}{cc}
\begin{minipage}{.4\linewidth}
\centering
\begin{tikzpicture}[>=stealth']
\tikzstyle myBG=[line width=3pt,opacity=1.0]
\newcommand{\drawBGT}[2]
{
    \draw[ar,->,white,myBG]  (#1) -- (#2);
    \draw[ar,->,black,very thick] (#1) -- (#2);
}

\node () at (3,1) {$\model$};
\node (h1) at (4,1) [blackworld,label=above:$w$] {} ;
\node (h2) at (3,-1) [world,label=below:$v$] {} ;
\node (h3) at (5,-1) [world,label=below:$u$] {} ;
\drawBGT{h1}{h2};
\drawBGT{h1}{h3};
\path [thin,ar,<->] (h2) edge [dotted, bend left] (h3);
\end{tikzpicture}
\end{minipage}
&
\begin{minipage}{.4\linewidth}
\centering
\begin{tikzpicture}[>=stealth']
\tikzstyle myBG=[line width=3pt,opacity=1.0]
\newcommand{\drawBGT}[2]
{
    \draw[ar,->,white,myBG]  (#1) -- (#2);
    \draw[ar,->,black,very thick] (#1) -- (#2);
}
\node () at (3,1) {$\model^f$};

\node (h1) at (4,1) [blackworld,label=above:$\set{w}$] {} ;
\node (h2) at (4,-1) [world,label=below:$\set{v,u}$] {} ;
\drawBGT{h1}{h2};
\path [thin,ar,->] (h2) edge [dotted, loop right,] (h2);
\end{tikzpicture}
\end{minipage}
\end{tabular}
\end{center}

Finally, we show that filtration does not work for proving that the satisfiability
problem for $\hxpd$ over $\ct$ is decidable. This is why, in order to show Corollary~\ref{cor:tree-decid}
we used a more specialized method. Consider again $\Sigma=\set{\tup{\dowa}\top}$. 
The model on the right is the filtration by $\Sigma$ of the model on the left
(assuming all nodes are equal for the data relation).
Any filtration collapses all nodes from $\model$ into a single reflexive node.
Clearly, $\model^f\notin\ct$.

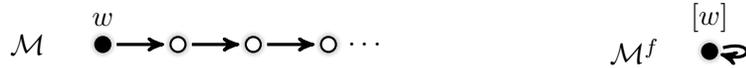
\begin{center}
  \begin{tabular}{cc}
    \begin{minipage}{.5\linewidth}
      \begin{tikzpicture}[>=stealth']
      \tikzstyle myBG=[line width=3pt,opacity=1.0]
      \newcommand{\drawBGT}[2]
      {
          \draw[ar,->,white,myBG]  (#1) -- (#2);
          \draw[ar,->,black,very thick] (#1) -- (#2);
      }

      \node () at (0,0) {$\model$};
      \node (h1) at (4,0) [world,label=right:$\ldots$] {} ;
      \node (h2) at (1,0) [blackworld,label=above:$w$] {} ;
      \node (h3) at (2,0) [world] {} ;
      \node (h4) at (3,0) [world] {} ;
      \drawBGT{h4}{h1};
      \drawBGT{h2}{h3};
      \drawBGT{h3}{h4};
      \end{tikzpicture}
    \end{minipage}
    &
    \begin{minipage}{.2\linewidth}
      \begin{tikzpicture}[>=stealth']
        \tikzstyle myBG=[line width=3pt,opacity=1.0]
        \node () at (3,0) {$\model^f$};
        \node (h2) at (4,0) [blackworld,label=above:\text{[$w$]}] {} ;
        \path [ar,->,black,very thick] (h2) edge [loop right] (h2);
      \end{tikzpicture}
      \end{minipage}
  \end{tabular}
\end{center}

%%% Local Variables:
%%% mode: latex
%%% TeX-master: "jair18"
%%% End:

\section{Discussion}
\label{sec:final}
%!TEX root = lmcs20.tex 
We introduced a sound and strongly complete axiomatization for
$\hxpd$, i.e., the language XPath with forward navigation for multiple
accessibility relations; multiple equality/ inequality data
comparisons; and where node expressions are extended with nominals,
and path expressions are extended with the hybrid operator $@$.  The
{\em hybridization} of XPath allowed us to apply a completeness
argument similar to the one used for the hybrid logic $\hl$ shown in,
e.g.,~\cite{BtC06}. This ensures that certain extensions of the
axiomatic system we introduce are also strongly complete. The
axiomatic systems that can be obtained in this way cover a large
family of hybrid XPath languages over different classes of frames.

Our system extends the calculus introduced previously
in~\cite{ArecesF16}.  The most important improvement is that we
provide a minimal system which can be extended with axioms and rules
of certain kind, such that strong completeness immediately follows.
The kind of axioms and rules we allow in the extensions ensures
completeness with respect to a large family of frame classes. We
showed interesting examples of such extensions.  In particular, we
obtain a strongly complete axiomatization for the logic extended with
backward and sibling navigation, and for data equality inclusion.

One particularly interesting class of structures used in practice is the class
of \emph{tree models}, since the main applications of \xpath in, e.g.,
web semantics are related to XML documents. From a mathematical point
of view, XML documents are trees. In this respect, we investigate
axiomatizations for $\hxpd$ on tree models.
First, we showed that there is no 
finitary first-order axiomatization which is strongly complete for
$\hxpd$ on trees.  Then we discussed two alternatives for dealing with
tree-like structures. On the one hand, it is possible to relax some
condition on the models. We consider the class $\ctm$, which are
forest-like models but possibly with infinite chains with respect to
the predecessor relation. More precisely, we extend the basic axiom
system with an infinite set of pure axioms to strongly axiomatize the
class $\ctm$.  Another alternative is giving up \emph{strong}
completeness. We showed weak completeness for the class $\ct$, using
pure axioms.

We also investigated the status of decidability for the satisfiability
problem of $\hxpd$ and some extensions. We used a standard technique
in modal logics named \emph{filtrations}~\cite{blackburn06}.  The
filtration method is a way to build finite models by taking a large
model and collapsing as many states as possible. We replicate this
technique for $\hxpv$, and obtain a \nexptime upper bound for its
satisfiability problem.  We showed that filtrations do not work on
some extensions of $\hxpd$.  In particular, for the case of tree
models we proved that the satisfiability problem is decidable in \expspace 
using a more specialized approach.

As future work, it would be interesting to investigate $\xpathd$
fragments with the reflexive-transitive path $\dowa^*$. One of the
main limitations of the framework we introduced in this paper is that
it can only axiomatize first-order languages. For that reason,
transitive closure operators cannot be accounted for, and a different
proof strategy is needed.  We conjecture that it is possible to adapt
the results from~\cite{HKT00} for Propositional Dynamic Logic (PDL) to
obtain a weakly complete axiom system for $\hxprt$ (i.e., $\hxpd$
extended with $\dowa^*$) over the class of all models. However, for
axiomatizing $\hxprt$ over trees, the filtration technique used
in~\cite{HKT00} does not seem to work, and new developments are needed, 
as for other logics with fix point operators over tree-like structures
(see, e.g., the case of ${\rm CTL}^*$ in~\cite{Rey01}).

The exact computational complexity of the logics we considered has not
been investigated yet. It has been proved that $\xpathd$ with single
accessibility and data relations extended with hybrid operators is
\pspace-complete~\cite{ArecesFS17}. Moreover, we provided upper bounds for
multi-modal $\hxpv$ over arbitrary models and for $\hxpd$ on trees, but
without giving a tight lower bound. It would be also interesting to
investigate the complexity of the multi-modal languages we studied in
this article, enriched with backward and sibling navigation,
reflexive-transitive closures, and over the class of trees. We
conjecture that we can adapt the automata proof given
in~\cite{FigPhD}, with the method used to account for hybrid operators
presented in~\cite{SattlerVardi-IJCAR} to get exact bounds.

%%% Local Variables: 
%%% mode: latex
%%% TeX-master: "lmcs20"
%%% End: 

\paragraph{\bf Acknowledgments}
We would like to thank the two anonymous reviewers for their helpful comments and
suggestions.
This work is supported by projects ANPCyT-PICTs-2017-1130,
Stic-AmSud 20-STIC-03 `DyLo-MPC', Secyt-UNC, GRFT Mincyt-Cba,
and
by the Laboratoire International Associ\'e SINFIN.

\bibliography{bib}
\bibliographystyle{alpha}

\end{document}